\documentclass[12pt]{article}
\usepackage[utf8]{inputenc}
\usepackage{geometry}
\usepackage{amsfonts,amsthm,bm,dsfont}
\usepackage{amssymb}
\usepackage{upgreek}
\usepackage{setspace}
\usepackage{amsmath}
\usepackage{makecell}
\usepackage{multirow}
\usepackage{bbm}
\usepackage{adjustbox}
\usepackage{booktabs} 
\usepackage{xcolor}
\usepackage[title]{appendix}
\usepackage{url}
\usepackage{algorithm}
\usepackage{algpseudocode}
\usepackage{graphicx}
\usepackage{subcaption}

\numberwithin{equation}{section}

\newcommand{\bs}{\boldsymbol}

\DeclareMathOperator{\tr}{tr}
\newcommand{\vech}{\operatorname{vech}}

\usepackage[compact]{titlesec}
\titlespacing{\section}{0pt}{2ex}{1ex}
\titlespacing{\subsection}{0pt}{1ex}{0ex}
\titlespacing{\subsubsection}{0pt}{0.5ex}{0ex}

\usepackage[authoryear,round]{natbib}

\usepackage{xr}
\makeatletter
\newcommand*{\addFileDependency}[1]{
  \typeout{(#1)}
  \@addtofilelist{#1}
  \IfFileExists{#1}{}{\typeout{No file #1.}}
}
\makeatother

\newcommand*{\myexternaldocument}[1]{
    \externaldocument{#1}
    \addFileDependency{#1.tex}
    \addFileDependency{#1.aux}
}

\usepackage[colorlinks=true,
            linkcolor=blue,
            citecolor=blue,
            urlcolor=blue]{hyperref}

\newcommand{\fnlink}[1]{\hyperlink{fn#1}{#1}}

\providecommand{\keywords}[1]{{\small{\textbf{\textit{Keywords ---}} #1}}}

\myexternaldocument{supplementary}

\theoremstyle{plain}
\newtheorem{theorem}{Theorem}

\newtheorem{assumption}{Assumption}

\newtheorem{remark}{Remark}

\newtheorem{lemma}{Lemma}
\newtheorem{corollary}{Corollary}

\begin{document}

\thispagestyle{empty} \baselineskip=28pt \vskip 5mm
\begin{center}
{\Huge{\bf Mixture-of-Finite-Mixtures Wishart Model for Clustering Covariance Matrices with an Application to Brain Functional Connectivity}}
\end{center}

\baselineskip=12pt \vskip 5mm

\begin{center}\large
Zongyu Li\textsuperscript{\fnlink{1},\fnlink{4}}, 
Stefano Castruccio\textsuperscript{\fnlink{2}}, 
and Zhiyong Zhang\textsuperscript{\fnlink{3}}
\end{center}

\footnotetext[1]{\hypertarget{fn2}{}\baselineskip=11pt Department of Statistics, University of Warwick, Coventry, CV4 7AL, United Kingdom.}

\footnotetext[2]{\hypertarget{fn1}{}\baselineskip=11pt Department of Applied and Computational Mathematics and Statistics,
University of Notre Dame, Notre Dame, IN 46556, USA.}

\footnotetext[3]{\hypertarget{fn3}{}\baselineskip=11pt Department of Psychology,
University of Notre Dame, Notre Dame, IN 46556, USA.}

\footnotetext[4]{\hypertarget{fn4}{}\baselineskip=11pt Corresponding author. Email: \href{mailto:Zongyu.Li@warwick.ac.uk}{Zongyu.Li@warwick.ac.uk}.}

\baselineskip=17pt \vskip 5mm 
\centerline{\today} 

\vskip 5mm
\begin{center}
{\large{\bf Abstract}}
\end{center}

Data represented as covariance-type matrices arise in many fields, including brain functional connectivity and diffusion tensor imaging. We develop the MFM--Wishart, a Bayesian model-based clustering approach for such data that combines Wishart mixture components with a mixture-of-finite-mixtures (MFM) prior, allowing joint posterior inference on both the number of clusters and clustering assignments. Theoretically, we study the properties of Wishart kernels in the context of mixture models and then establish results for posterior consistency for the number of clusters and posterior contraction of the mixing measure under standard regularity conditions. Computationally, we develop an efficient Markov chain Monte Carlo (MCMC) algorithm for posterior inference. Simulation studies show competitive clustering performance and accurate recovery of the number of clusters, even under model misspecification. We apply MFM--Wishart to cluster infants based on functional connectivity during sleep, estimated from functional near-infrared spectroscopy (fNIRS) data, illustrating the practical utility of the model and revealing interpretable heterogeneity.

\newpage

\section{Introduction}\label{sec:intro}

Covariance-type matrices arise naturally in many scientific fields, including neuroscience \citep{allen14, zhu09}, computer vision \citep{tuzel06}, and finance \citep{bn04}. These matrices encode second-order dependence structures among variables and form an important class of symmetric positive definite (SPD) matrix-valued data. Brain functional connectivity provides a particularly important example and is often summarized by covariance or correlation matrices constructed from time series measured across brain regions using electroencephalography (EEG), functional magnetic resonance imaging (fMRI), and functional near-infrared spectroscopy (fNIRS) \citep{niu14,rogers07,sakk11}. Many studies have focused on identifying significant links between brain regions \citep{smith11} or testing group differences across tasks or populations \citep{strain22}.

Beyond these objectives, clustering covariance-type matrices provides a natural way to uncover latent structure or heterogeneity by grouping observations into clusters with distinct patterns. Despite its broad relevance, this problem has received relatively limited attention. In the context of brain functional connectivity, clustering can help identify connectivity phenotypes and subject subtypes that may be relevant to cognition and neuropsychiatric disorders \citep{miranda21,zhang21}. At the same time, clustering covariance-type matrices is challenging because of the positive-definite constraint and the non-Euclidean geometry of the SPD manifold. We next review several existing approaches to clustering covariance-type matrices, with particular attention to methods relevant to brain functional connectivity.

\subsection{Background}\label{sec:related_work}

Many clustering methods proceed by specifying a notion of dissimilarity or similarity between observations and then grouping data based on pairwise distances \citep{irani16, van19}. Classical examples include \(k\)-means \citep{ikotun23} and hierarchical clustering \citep{murtagh12}, as well as graph- and kernel-based methods such as spectral clustering \citep{von07}. When the observations are covariance-type matrices, na\"ive Euclidean distances applied to vectorized entries may fail to capture the geometry of the observations. This has motivated substantial work on non-Euclidean distances for such data \citep{dryden09,jaya13,yin16}. Given such distances, standard distance-based pipelines can be extended to cluster covariance-type matrices, but their performance and interpretability depend sensitively on the chosen distance. Moreover, distance-based methods typically define clusters through an optimization criterion, without assuming an explicit generative model. As a result, they do not naturally support population-level inference, which is important in scientific settings such as phenotyping brain functional connectivity.

An alternative is model-based clustering \citep{bouveyron14,gormley23}, which formulates clustering through an explicit generative model, often via a mixture of distributions, so that latent components induce a partition of the observations. Within this framework, finite mixture models (FMMs) \citep{mclachlan19} assume a fixed number of components \(K\), which is typically interpreted as the number of clusters in clustering applications. For covariance-type matrix-valued observations, FMM-based clustering requires choosing suitable component distributions supported on the space of SPD matrices. The Wishart family provides a natural choice and leads to computationally convenient inference. \citet{hidot10} proposed a Wishart finite mixture model for clustering covariance matrices and developed an expectation–maximization (EM) algorithm for estimation. \citet{capp25} introduced a penalized Wishart mixture model with component-specific sparsity structures to accommodate high-dimensional and sparse covariance-type matrices, and applied it to brain functional connectivity data. More recently, \citet{mai26} developed Bayesian inference for finite Wishart mixtures and a mixture-of-experts extension in which mixing probabilities depend on covariates through a multinomial logistic gating network.

Despite these advantages, FMM-based clustering has an important limitation: the number of components \(K\) must be specified before model fitting. When the number of clusters is unknown, which is common in practice, one typically fits models over a range of \(K\) values and then selects an optimal value by model comparison \citep{gormley23}. However, this post hoc procedure yields only a point estimate and does not automatically propagate uncertainty about \(K\) into downstream clustering summaries. This limitation motivates Bayesian formulations that infer the number of clusters jointly with assignments and component-specific parameters within a unified probabilistic framework.

One Bayesian method is based on Bayesian nonparametrics (BNP), and the Dirichlet process mixture (DPM) is arguably the most widely used example \citep{muller15,teh11}. Under a DPM, the number of clusters is inferred from the data together with the clustering assignments and can grow with sample size. \citet{cherian15} first combined DPM with Wishart kernels for clustering SPD matrices, using the Wishart--inverse-Wishart conjugate pair together with an efficient collapsed Gibbs sampler for posterior inference. \citet{tokuda21} later extended this approach to multiple clustering for functional connectivity data, with the goal of jointly identifying subject clusters and associated brain sub-networks.

Despite its flexibility, \citet{miller13,miller14} showed that the posterior estimate of the number of clusters under a DPM can be inconsistent when the true data-generating mixture has finitely many components, even in simple settings such as univariate normal mixtures with known unit variance. In practice, DPM-based clustering may produce extraneous small clusters, which can hinder interpretation of the inferred partition \citep{wade23}. This phenomenon can also be observed in its application to functional connectivity. For example, the clustering results reported in \citet{tokuda21} contain many small clusters that are difficult to interpret.

To address this issue, \citet{miller18} introduced the mixture-of-finite-mixtures (MFM) framework, which places a prior directly on the number of mixture components by treating \(K\) as a random variable. Unlike the DPM, the MFM assumes a finite number of mixture components almost surely and can consistently estimate the number of clusters under suitable regularity conditions. \citet{guha21} further showed that MFM achieves a posterior contraction rate for the mixing measure that is optimal up to a logarithmic factor under mild conditions. Despite these attractive properties, to our knowledge no existing work has developed an MFM-type model for clustering covariance matrices while also establishing posterior consistency for the number of clusters and posterior contraction, except for one loosely related study that combines an MFM prior with matrix normal distributions to cluster field-goal attempt patterns of professional basketball players \citep{yin23}. Motivated by this gap, we propose an MFM--Wishart model for clustering covariance-type matrices, with a particular focus on matrix-valued brain functional connectivity data.

\subsection{Contributions and organization of the paper}\label{sec:contribution}

This paper contributes to Bayesian model-based clustering of covariance-type matrices, with an application to covariance- and correlation-based brain functional connectivity. We propose the MFM--Wishart model, a Wishart mixture model equipped with an MFM prior, which enables joint posterior inference on the number of clusters and the clustering assignments. We further study the theoretical properties of the Wishart kernel in the context of mixture models and, under standard regularity conditions, establish posterior consistency for the number of clusters as well as posterior contraction of the mixing measure for the proposed MFM--Wishart model. For inference, we develop an efficient Markov chain Monte Carlo (MCMC) algorithm that exploits Wishart--inverse-Wishart conjugacy and the partition representation induced by the MFM prior. We evaluate the empirical performance of the method through simulation studies and apply it to infants' brain functional connectivity during sleep, estimated from fNIRS data, where the method identifies interpretable cluster-specific connectivity patterns.

The remainder of the paper is organized as follows. Section~\ref{sec:method} presents the proposed MFM--Wishart model and discusses its advantages over DPM-based approaches. Section~\ref{sec:theory} establishes theoretical properties of the proposed model. Section~\ref{sec:inference} develops an MCMC algorithm for posterior inference. Section~\ref{sec:simulation} reports extensive simulation studies comparing the proposed method with DPM- and FMM-based models as well as distance-based clustering methods. Section~\ref{sec:application} presents the application to infant fNIRS functional connectivity. Section~\ref{sec:conclusion} concludes with a discussion and directions for future research. Proofs of the theoretical results and additional experiments are provided in the Appendices.

\section{Method}\label{sec:method}

\subsection{Wishart distribution}\label{sec:wishart}

We work with SPD matrices in \(\mathbb S_{++}^p\), the cone of \(p\times p\) SPD matrices. Given the scale matrix \(\mathbf{\Sigma}\in\mathbb S_{++}^p\) and degrees of freedom \(\nu>p-1\), we write \(\mathbf{W}\sim\mathcal W_p(\mathbf{\Sigma},\nu)\) if \(\mathbf{W}\) admits the density
\begin{equation}
    f(\mathbf{W}\mid \mathbf{\Sigma},\nu)
    =
    \frac{1}{2^{\nu p/2}|\mathbf{\Sigma}|^{\nu/2}\Gamma_p\!\left(\nu/2\right)}
    |\mathbf{W}|^{(\nu-p-1)/2}
    \exp\!\left\{-\frac{1}{2}\tr\!\left(\mathbf{\Sigma}^{-1}\mathbf{W}\right)\right\},
    \qquad \mathbf{W}\in\mathbb S_{++}^p.
    \label{eq:wishart_density}
\end{equation}
Here \(\Gamma_p(a):=\pi^{p(p-1)/4}\prod_{i=1}^p\Gamma\!\left(a-\frac{i-1}{2}\right)\) is the multivariate gamma function (defined for \(a>(p-1)/2\)). The mean satisfies \(\mathbb E(\mathbf{W})=\nu\,\mathbf{\Sigma}\), and for any \(1\le i,j,k,\ell\le p\),
\(
\operatorname{Cov}((\mathbf{W})_{ij},(\mathbf{W})_{k\ell})
=\nu\Big((\mathbf{\Sigma})_{ik}(\mathbf{\Sigma})_{j\ell}+(\mathbf{\Sigma})_{i\ell}(\mathbf{\Sigma})_{jk}\Big),
\)
where \((\mathbf W)_{ij}\) denotes the \((i,j)\)-th element of \(\mathbf W\). If \(\boldsymbol{x}_1,\dots,\boldsymbol{x}_T\in\mathbb R^p\) are i.i.d.\ \(\mathcal N(\boldsymbol{\mu},\mathbf{\Sigma})\), \(\mathbf{S}\sim\mathcal W_p(\mathbf{\Sigma},T-1)\) with  \(\bar{\boldsymbol{x}}=\frac{1}{T}\sum_{t=1}^T \boldsymbol{x}_t\) and \(\mathbf{S}=\sum_{t=1}^T (\boldsymbol{x}_t-\bar{\boldsymbol{x}})(\boldsymbol{x}_t-\bar{\boldsymbol{x}})^\top\). Equivalently, the sample covariance \(\widehat{\mathbf{\Sigma}}=\mathbf{S}/(T-1)\) satisfies \((T-1)\,\widehat{\mathbf{\Sigma}}\sim\mathcal W_p(\mathbf{\Sigma},(T-1))\) and \(\mathbb E(\widehat{\mathbf{\Sigma}})=\mathbf{\Sigma}\). This provides the classical motivation for using a Wishart likelihood for covariance-type observations. 

In covariance-based functional connectivity, each observation is typically computed from a multivariate time series \(\{\boldsymbol{x}_t\}_{t=1}^T\) across multiple regions of interest (ROIs) or channels, which often exhibits positive temporal autocorrelation. Consequently, \(\mathbf{S}\) and \(\widehat{\mathbf{\Sigma}}\) are generally not exactly Wishart, even under Gaussianity. Nevertheless, if the mean is known to be zero and one uses the second-moment estimator \(\widetilde{\mathbf \Sigma}=\frac{1}{T}\sum^T_{t=1}\bs x_t\bs x_t^\top \), then we have \(\mathbb E(\widetilde{\mathbf \Sigma})=\mathbf{\Sigma}\), whereas temporal dependence changes the sampling variability of \(\widetilde{\mathbf \Sigma}\). In Supplementary Material~\ref{app:temporal_moments}, we derive the second moments of \(\widetilde{\mathbf \Sigma}\) under a jointly Gaussian, second-order stationary model and show that autocorrelation can inflate \(\operatorname{Cov}((\widetilde{\mathbf \Sigma})_{ij},(\widetilde{\mathbf \Sigma})_{k\ell})\) relative to the i.i.d.\ case.

Motivated by this, we can treat \(\nu\) as unknown and infer it from the data, allowing the likelihood to adapt to the levels of temporal dependence. In this case, one can interpret \(\nu\) as an effective sample size, and \(\nu\) is often estimated to be much smaller than \(T\) when temporal autocorrelation is strong and positive, as shown in Appendix~\ref{app:temporal_moments}. Similar approaches using Wishart likelihoods have been studied previously for functional connectivity, with \(\nu\) interpreted as an effective sample size; see, for example, \citet{niel17}, \citet{tokuda21}, and \citet{seil17}.

\subsection{Mixture of finite mixtures with Wishart kernels}\label{sec:MFM--Wishart}

The preceding discussion motivates the use of Wishart densities as working likelihoods for covariance-type matrices. For clustering, a single Wishart distribution is insufficient because the observed matrices may arise from multiple latent subpopulations with distinct covariance or connectivity patterns. This motivates the use of Wishart densities as component kernels in a mixture model. Let \(\mathbf{W}_1,\ldots,\mathbf{W}_n\) denote \(n\) observed covariance-type matrices in \(\mathbb{S}_{++}^p\), and suppose that they arise from \(K\) mixture components. Let \(z_i\in\{1,\ldots,K\}\) denote the cluster label for observation \(i\). In this model-based clustering framework, clustering amounts to inferring the labels \(z_1,\ldots,z_n\).

We note that, with a finite sample of size \(n\), it is possible that only \(K_{+,n}\) of the \(K\) components are occupied by the observations, with \(K_{+,n}\le K\). From now on, we will refer to \(K_{+,n}\) as the number of occupied components, or simply, the number of clusters in the data. When \(K\) is unknown, one must also infer it and thus the number of clusters \(K_{+,n}\). Bayesian methods provide a natural way to infer the number of clusters and the clustering assignments jointly. A widely used model is the DPM, which induces a prior on the clustering assignments through the Chinese restaurant process (CRP):
\begin{equation}\label{eq:crp}
    p(z_i=c\mid z_1,\dots,z_{i-1})=\begin{cases}
        \frac{|c|}{i-1+\alpha},\quad\text{if \(c\) is an existing cluster},\\
        \frac{\alpha}{i-1+\alpha},\quad \text{if \(c\) is a new cluster},
    \end{cases}
\end{equation}
where \(|c|\) denotes the size of an existing cluster \(c\) in \(\{z_1,\dots,z_{i-1}\}\), and \(\alpha>0\) is the concentration parameter controlling the probability of creating a new cluster. When \(\alpha\) is small, observation \(i\) is more likely to be assigned to an existing cluster; when \(\alpha\) is large, it is more likely to initiate a new cluster. In addition, larger existing clusters receive higher prior mass, producing the ``rich-get-richer'' effect. The prior in \eqref{eq:crp} assigns positive mass to partitions with countably many clusters. For any finite sample, however, only finitely many clusters can exist, and this number can thus be interpreted as the inferred number of clusters in the data.

As mentioned in Section~\ref{sec:related_work}, the CRP prior in \eqref{eq:crp} can produce extraneous small clusters and the posterior estimate of the number of clusters can be inconsistent when the true data-generating mixture has finitely many components. In many applications, it is more natural to assume that the number of components is finite but unknown. For example, under a fixed experimental condition, one may expect only finitely many biologically meaningful brain activity patterns. In such settings, a prior that assumes infinitely many components may be less appropriate. To address this issue, the MFM was introduced \citep{miller18}, which is specified hierarchically as
\begin{equation}\label{eq:mfm}
\begin{aligned}
K &\sim p_K,\\
\boldsymbol{\pi}=(\pi_1,\dots,\pi_k)\mid K=k &\sim \mathrm{Dirichlet}(\gamma,\dots,\gamma), \\
z_i\mid \boldsymbol{\pi}, K=k &\stackrel{\mathrm{i.i.d.}}{\sim} \mathrm{Discrete}(\pi_1,\dots,\pi_k), \qquad i=1,\dots,n,
\end{aligned}
\end{equation}
where \(p_K\) denotes a prior distribution supported on the positive integers, \(\mathrm{Dirichlet}(\gamma,\dots,\gamma)\) is a symmetric Dirichlet distribution with concentration parameter \(\gamma>0\), and \(\mathrm{Discrete}(\pi_1,\dots,\pi_k)\) denotes a categorical distribution over \(\{1,\dots,k\}\) with probability vector \(\boldsymbol{\pi}\).

Like the CRP in \eqref{eq:crp}, the MFM prior also induces the following prior full conditional for the cluster labels \(\{z_i\}_{i=1}^n\) \citep{miller18}:
\begin{equation}\label{eq:prior_mfm}
    p(z_i=c\mid \boldsymbol{z}_{-i})\propto\begin{cases}
        |c|+\gamma,&\text{if \(c\) is an existing cluster},\\
        \frac{V_n\left(K^*+1\right)}{V_n\left(K^*\right)}\gamma,\ &\text{if \(c\) is a new cluster},
    \end{cases}
\end{equation}
where \(\boldsymbol{z}_{-i} = (z_1,\dots,z_{i-1},z_{i+1},\dots,z_n)^\top\), and \(K^*\) is the number of clusters in \(\boldsymbol{z}_{-i}\). The term \(V_n(t)\) is defined as
\begin{equation}
    V_n(t) = \sum_{k=1}^{\infty}\frac{k_{(t)}}{(\gamma k)^{(n)}}p_K(k),
\end{equation}
where \(k^{(m)}= k(k + 1)\cdots(k + m - 1)\) and \(k_{(m)}= k(k - 1)\cdots(k - m + 1)\), with \(k^{(0)}=k_{(0)}=1\). Compared with the CRP prior in \eqref{eq:crp}, the ratio \(V_n\left(K^*+1\right)/V_n\left(K^*\right)\) tends to slow the creation of new clusters and hence helps avoid many small extraneous clusters. In practice, the infinite sum defining \(V_n(t)\) is truncated at a sufficiently large value of \(k\) so that the relative truncation error in the ratios
$V_n(t+1)/V_n(t)$ is negligible.

As argued in \citet{miller18}, relative to the DPM, the MFM places a prior directly on the finite but unknown number of mixture components, making prior beliefs about cluster complexity more transparent and easier to calibrate. In contrast, under a DPM the induced prior behavior on the number of components is controlled only indirectly through the concentration parameter and is also sensitive to the sample size \(n\). 

By combining the MFM prior with Wishart kernels, we propose the MFM--Wishart model:
\begin{equation}\label{eq:mfm_wishart_model}
\begin{aligned}
K &\sim p_K,\\
\boldsymbol{\pi}=(\pi_1,\ldots,\pi_k)\mid K=k &\sim \mathrm{Dirichlet}\!\left(\gamma,\ldots,\gamma\right),\\
z_i \mid \boldsymbol{\pi},K=k &\stackrel{\mathrm{i.i.d.}}{\sim} \mathrm{Discrete}(\pi_1,\ldots,\pi_k),\qquad i=1,\ldots,n,\\
\mathbf{\Sigma}_1,\ldots,\mathbf{\Sigma}_k \mid K=k &\stackrel{\mathrm{i.i.d.}}{\sim} \mathcal{IW}_p(\mathbf{\Psi}_0,\kappa_0),\\
\nu_1,\ldots,\nu_k \mid K=k &\stackrel{\mathrm{i.i.d.}}{\sim} \mathrm{Uniform}(\nu_L,\nu_U),\\
\mathbf{W}_i \mid z_i,\{(\mathbf{\Sigma}_j,\nu_j)\}_{j=1}^k &\sim \mathcal{W}_p(\mathbf{\Sigma}_{z_i},\nu_{z_i}),\qquad i=1,\ldots,n,
\end{aligned}
\end{equation}
where \(\mathcal{IW}_p(\mathbf{\Psi}_0,\kappa_0)\) denotes an inverse-Wishart distribution with hyperparameters \(\mathbf{\Psi}_0\in\mathbb{S}_{++}^p\) and \(\kappa_0>p-1\), and \(\mathrm{Uniform}(\nu_L,\nu_U)\) denotes the uniform distribution on \([\nu_L,\nu_U]\), with \(p-1<\nu_L<\nu_U<\infty\). The inverse-Wishart prior is chosen because it is conjugate to the Wishart likelihood for the scale matrix, which facilitates efficient posterior sampling in Section~\ref{sec:inference}. The inverse-Wishart density is parameterized as follows:
\begin{equation}
p(\mathbf{\Sigma}\mid\mathbf{\Psi}_0,\kappa_0)=\frac{|\mathbf{\Psi}_0|^{\kappa_0/2}}{2^{\kappa_0p/2}\Gamma_p\left(\kappa_0/2\right)}|\mathbf{\Sigma}|^{-(\kappa_0+p+1)/2}\exp{\left\{-\frac{1}{2}\tr\left(\mathbf{\Psi}_0\mathbf{\Sigma}^{-1}\right)\right\}},\quad \mathbf{\Sigma}\in \mathbb{S}_{++}^p.
\end{equation}

For \(p_K\), we use a Poisson distribution with parameter \(\lambda\) shifted to the positive integers, that is, \(K-1\sim\mathrm{Poisson}(\lambda)\). Following \citet{miller18}, we set \(\gamma=1\) and \(\lambda=1\).

For model~(\ref{eq:mfm_wishart_model}), both the scale matrix \(\mathbf{\Sigma}_k\) and the degrees-of-freedom parameter \(\nu_k\) are component-specific. A specialization of model~(\ref{eq:mfm_wishart_model}) is obtained by imposing a shared degrees-of-freedom parameter across mixture components, \(\nu_1=\cdots=\nu_K=\nu.\) In previous DPM-based Wishart mixture models, such as those of \citet{cherian15} and \citet{tokuda21}, this shared-\(\nu\) assumption is imposed. However, we note that this assumption may be violated when the temporal correlation of brain signals differs among individuals. In broader applications, cluster-specific degrees of freedom may be more flexible and robust, as considered in \citet{mai26}. We will focus on model~(\ref{eq:mfm_wishart_model}) to study its large-sample properties in Section~\ref{sec:theory} and propose an inference algorithm in Section~\ref{sec:inference}.

\section{Theoretical Results}\label{sec:theory}

In this section, we study Wishart kernels in the context of mixture models and establish large-sample theory for posterior consistency of the number of mixture components and posterior contraction of the mixing measure for our MFM--Wishart model formulated in (\ref{eq:mfm_wishart_model}). We focus on a reparameterization of \(\mathbf{\Sigma}\) by the precision matrix \(\mathbf{\Lambda}=\mathbf{\Sigma}^{-1}\). Following a common restriction in the theoretical analysis of MFM models \citep{guha21}, we develop our theory on the following compact parameter space:
\begin{equation}\label{eq:Theta_star}
\Theta^*
=
\Big\{\mathbf{\Lambda}\in \mathbb{S}_{++}^p:\ \underline{\lambda}\mathbf{I}_p\preceq \mathbf{\Lambda}\preceq \overline{\lambda}\mathbf{I}_p\Big\}
\times [\underline{\nu},\,\overline{\nu}],
\end{equation}
where \(\underline{\lambda}\), \(\overline{\lambda}\), 
\(\underline{\nu}\), and \(\overline{\nu}\) are fixed constants satisfying
\(
0<\underline{\lambda}<\overline{\lambda}<\infty
\)
,
\(
p+1<\underline{\nu}<\overline{\nu}<\infty,
\)
and
\(
[\underline{\nu}, \overline{\nu}] \subseteq [\nu_L,\nu_U].
\)
We write \(\mathbf{I}_p\) for the \(p\times p\) identity matrix. For symmetric matrices \(\mathbf{A}\) and \(\mathbf{B}\), \(\mathbf{A}\preceq \mathbf{B}\) means that \(\mathbf{B}-\mathbf{A}\) is positive semidefinite. 

To work with Euclidean distances, we identify each precision matrix \(\mathbf{\Lambda}\in\mathbb{S}_{++}^p\) with its half-vectorization \(\boldsymbol{\eta}:=\vech(\mathbf{\Lambda})\in\mathbb{R}^{p(p+1)/2}\), which stacks the lower-triangular entries of \(\mathbf{\Lambda}\) into a vector. Let \(d=p(p+1)/2\) and define the component parameter by \(\boldsymbol{\theta}=\big(\boldsymbol{\eta}^\top,\nu\big)^\top \in \mathbb{R}^{d+1}\). With a slight abuse of notation, we view \(\Theta^*\) as a compact subset of \(\mathbb{R}^{d+1}\) through this reparameterization. For \(\boldsymbol{\theta}=(\boldsymbol{\eta}^\top,\nu)^\top\) and \(\boldsymbol{\theta}'=(\boldsymbol{\eta}'^\top,\nu')^\top\), we use the Euclidean metric \(\|\boldsymbol{\theta}-\boldsymbol{\theta}'\|_2 := \left(\|\boldsymbol{\eta}-\boldsymbol{\eta}'\|_2^2 + |\nu-\nu'|^2\right)^{1/2}\).

Let \(G_0=\sum_{j=1}^{k_0}\pi_j^0\,\delta_{\boldsymbol{\theta}_j^0}\) be the true but unknown discrete mixing measure with unknown but fixed \(k_0\) support points, where  \(\delta_{\boldsymbol{\theta}}\) denotes a point mass at \(\boldsymbol{\theta}\) and \(\boldsymbol{\theta}_j^0\in \Theta^*\) for all \(j=1,\dots,k_0\). Let \(f(\mathbf{W}\mid\boldsymbol{\theta})\) denote the Wishart density under this precision parameterization. We denote the corresponding data-generating distribution by \(P_{G_0}\), with the density \(p_{G_0}(\mathbf{W})=\sum_{j=1}^{k_0}\pi_j^0\,f(\mathbf{W}\mid\boldsymbol{\theta}_j^0)\), which generates i.i.d.\ samples \(\mathbf{W}_1,\dots,\mathbf{W}_n\).

Next, we introduce the Wasserstein distance between two mixing measures. Consider another mixing measure \(G=\sum_{i=1}^{k}\pi_i\,\delta_{\boldsymbol{\theta}_i}\) with \(\boldsymbol{\theta}_i\in \Theta^*\) for all \(i=1,\dots,k\). A coupling of \(\boldsymbol{\pi}=(\pi_1,\dots,\pi_k)\) and \(\boldsymbol{\pi}^0=(\pi_1^0,\dots,\pi_{k_0}^0)\) is any joint distribution on \(\{1,\dots,k\}\times\{1,\dots,k_0\}\) with marginals \(\boldsymbol{\pi}\) and \(\boldsymbol{\pi}^0\). We represent a coupling by a matrix \(\mathbf{Q}=(q_{ij})_{1\le i\le k,\ 1\le j\le k_0}\in[0,1]^{k\times k_0}\) satisfying \(\sum_{i=1}^{k} q_{ij}=\pi^0_j\), \(j=1,\dots,k_0\), and \(\sum_{j=1}^{k_0} q_{ij}=\pi_i\), \(i=1,\dots,k\). Let \(\mathcal{Q}(\boldsymbol{\pi},\boldsymbol{\pi}^0)\) denote the collection of all such matrices \(\mathbf{Q}\). For any \(r\ge1\), we denote the order-\(r\) Wasserstein distance between \(G\) and \(G_0\) by
\begin{equation}\label{eq:Wasserstein_def}
W_r(G,G_0) = \inf_{\mathbf{Q}\in \mathcal{Q}(\boldsymbol{\pi},\boldsymbol{\pi}^0)}
\left(\sum_{i=1}^{k}\sum_{j=1}^{k_0}q_{ij}\|\boldsymbol{\theta}_i-\boldsymbol{\theta}_j^0\|_2^r\right)^{1/r}.
\end{equation}

We begin by presenting two lemmas that state some properties of the class of restricted Wishart kernels \(\{f(\mathbf{W}\mid \boldsymbol{\theta}):\,\boldsymbol{\theta}\in\Theta^*\}\). Lemma~\ref{lem:first_order_identifiability} establishes first-order identifiability, and Lemma~\ref{lem:first_order_uniform_lipschitz} establishes the first-order uniform Lipschitz property. These two results are key ingredients for obtaining posterior consistency for the number of mixture components \(K\) and posterior contraction of the mixing measure \(G\) under our MFM--Wishart model. We note that these properties have been studied for many kernels in the mixture-model literature \citep{guha21, ho16}. However, to the best of our knowledge, they have not been established for Wishart kernels.

\begin{lemma}[First-order identifiability]
\label{lem:first_order_identifiability}
Assume \(p\ge2\). For any \(k\ge1\), assume \(\boldsymbol{\theta}_1, \boldsymbol{\theta}_2,\dots, \boldsymbol{\theta}_k\in\Theta^*\) are distinct, and that \(\{\alpha_i\}_{i=1}^k\subset\mathbb{R}\) and \(\{\boldsymbol{\beta}_i\}_{i=1}^k\subset\mathbb{R}^{d+1}\). If 
\begin{equation}
\sum^{k}_{i=1}\alpha_i f(\mathbf{W}\mid \boldsymbol{\theta}_i) + \sum^k_{i=1}\boldsymbol{\beta}_i^\top \nabla_{\boldsymbol{\theta}} f(\mathbf{W}\mid \boldsymbol{\theta}_i)=0
\qquad\text{for almost all }\mathbf{W}\in\mathbb{S}_{++}^p,
\end{equation}
then \(\alpha_i=0\) and \(\boldsymbol{\beta}_i=\boldsymbol{0}\) for \(i=1,2,\dots,k\).
\end{lemma}

\begin{remark}
The condition \(p\ge2\) in Lemma~\ref{lem:first_order_identifiability} is not merely a technical artifact. When \(p=1\), the Wishart family reduces to a Gamma family, for which first-order identifiability can fail unless the parameter space is further restricted; see Remark~\ref{rem:p1_gamma_failure} in the Appendix for a simple counterexample.
\end{remark}

\begin{lemma}[First-order uniform Lipschitz property]
\label{lem:first_order_uniform_lipschitz}
There exist constants \(\delta>0\) and \(C<\infty\), independent of \(\mathbf{W}\), \(\boldsymbol{\theta}_1\), and \(\boldsymbol{\theta}_2\), such that for any \(\boldsymbol{\theta}_1,\boldsymbol{\theta}_2\in\Theta^*\) and any \(\mathbf{W}\in\mathbb{S}_{++}^p\),
\begin{equation}
\|\nabla_{\boldsymbol{\theta}}f(\mathbf{W}\mid\boldsymbol{\theta}_1)
-
\nabla_{\boldsymbol{\theta}}f(\mathbf{W}\mid\boldsymbol{\theta}_2)\|_2
\le
C\|\boldsymbol{\theta}_1-\boldsymbol{\theta}_2\|_2^\delta.
\end{equation}
\end{lemma}

We now state the assumptions used in Theorem~\ref{thm:contraction}.

\begin{assumption}[Data-generating distribution]
\label{ass:data_generating}
The observations \(\mathbf{W}_1,\dots,\mathbf{W}_n\in \mathbb{S}_{++}^p\) with \(p\ge 2\) are i.i.d. from the finite Wishart mixture distribution \(P_{G_0}\), where
\(
G_0=\sum_{j=1}^{k_0}\pi_j^0\delta_{\boldsymbol{\theta}_j^0}
\)
has exactly \(k_0\) distinct support points \(\boldsymbol{\theta}_1^0,\dots,\boldsymbol{\theta}_{k_0}^0\in\Theta^*\). The true mixing weights satisfy \(\pi_j^0>0\) for \(j=1,\dots,k_0\) and \(\sum_{j=1}^{k_0}\pi_j^0=1\).
\end{assumption}

\begin{assumption}[Prior on the number of components]
\label{ass:prior_K}
The prior on \(K\) assigns positive mass to every positive integer, that is, \(p_K(k)>0\) for every \(k\in\mathbb{N}\). For example, the shifted Poisson prior \(K-1\sim\mathrm{Poisson}(\lambda)\), with \(\lambda>0\), satisfies this condition.
\end{assumption}

\begin{assumption}[Base prior on component parameters]
\label{ass:base_prior}
The induced base prior on \(\boldsymbol{\theta}=(\boldsymbol{\eta}^\top,\nu)^\top\), after restriction to \(\Theta^*\), has a continuous density with respect to Lebesgue measure that is strictly positive on \(\Theta^*\). We note that the priors specified in model~\eqref{eq:mfm_wishart_model}, after reparameterization and restriction to \(\Theta^*\), satisfy this condition.
\end{assumption}

Let \(\Pi_n(\cdot\mid \mathbf{W}_1,\dots,\mathbf{W}_n)\) be the posterior distribution obtained from the restricted version of model~\eqref{eq:mfm_wishart_model}, after reparameterizing by \((\mathbf{\Lambda},\nu)\) and restricting the component parameter space to \(\Theta^*\). The following theorem follows by verifying the regularity conditions of \citet[Theorem 3.1]{guha21} for the Wishart kernel on the compact parameter space \(\Theta^*\).

\begin{theorem}[Posterior consistency and contraction]
\label{thm:contraction}
Suppose Assumptions~\ref{ass:data_generating}--\ref{ass:base_prior} hold. For the MFM--Wishart model~(\ref{eq:mfm_wishart_model}) restricted to \(\Theta^*\), as \(n\rightarrow \infty\), we have
\begin{itemize}
\item[(a)] \(\Pi_n(K=k_0\mid \mathbf{W}_1,\dots,\mathbf{W}_n)\rightarrow 1\) almost surely under \(P_{G_0}\);
\item[(b)] there exists a constant \(C>0\), independent of \(n\), such that
\begin{equation}
\Pi_n\!\left(
W_1(G,G_0)\le C\left(\frac{\log n}{n}\right)^{1/2}
\ \middle|\ 
\mathbf{W}_1,\dots,\mathbf{W}_n
\right)
\rightarrow 1
\end{equation}
in \(P_{G_0}\)-probability.
\end{itemize}
\end{theorem}

Theorem~\ref{thm:contraction} establishes that, under the compact parameter space \(\Theta^*\) and the regularity conditions stated above, the posterior probability assigned to the true number of mixture components converges to one as the sample size grows. Moreover, the posterior distribution over mixing measures contracts around the true mixing measure \(G_0\) with respect to the order-\(1\) Wasserstein distance at rate \((\log n/n)^{1/2}\).

Theorem~\ref{thm:contraction}(a) concerns the number of mixture components \(K\). In clustering applications, however, one is often more interested in the number of clusters \(K_{+,n}\) in the observed data. \citet[Theorem 5.2]{miller18} show that, under the MFM prior, for each fixed \(k\) with \(p_K(1),\ldots,p_K(k)>0\), the posterior masses assigned to the events \(\{K=k\}\) and \(\{K_{+,n}=k\}\) differ by \(o(1)\) as \(n\rightarrow\infty\), pointwise in the observed data. Combining this pointwise equivalence with Theorem~\ref{thm:contraction}(a) yields posterior consistency for the number of clusters.

\begin{corollary}[Posterior consistency of the number of clusters]
\label{cor:cluster_consistency}
Under the same assumptions as Theorem~\ref{thm:contraction}, as \(n\rightarrow \infty\), we have
\begin{equation}
\Pi_n(K_{+,n}=k_0\mid \mathbf{W}_{1},\dots,\mathbf{W}_n)\rightarrow 1
\end{equation}
almost surely under \(P_{G_0}\).
\end{corollary}

\section{Inference}\label{sec:inference}

We develop a partially collapsed auxiliary-variable Metropolis--Hastings-within-Gibbs algorithm for posterior inference under the MFM--Wishart model introduced in Section~\ref{sec:MFM--Wishart}. We analytically integrate out the total number of mixture components \(K\), the mixing weights \(\boldsymbol{\pi}\), and the cluster-specific scale matrices \(\boldsymbol{\Sigma}_k\). Integrating out the scale matrices gives closed-form predictive densities conditional on the degrees of freedom. For a new cluster, the additional integral over its degrees-of-freedom parameter is not available in closed form. We therefore introduce temporary auxiliary values for this parameter, following a partially collapsed MFM adaptation of Algorithm~8 of \citet{neal00}. We first describe the resulting sampler and then the Dahl's method \citep{dahl06}, which is used to obtain a representative clustering from its output.

\subsection{Posterior sampling}
\label{subsec:inference-sampling}

The state of the sampler is \((\boldsymbol z,\{\nu_c:c\in\mathcal C\})\), where \(\mathcal C=\{c:\exists i\text{ such that }z_i=c\}\) denotes the set of occupied clusters. The partition can be initialized randomly or with each observation in its own cluster. For each initially occupied cluster \(c\), we initialize \(\nu_c\) randomly within \([\nu_L,\nu_U]\). In our algorithm, each iteration consists of an auxiliary-variable Gibbs sweep over the cluster labels, followed by separate random-walk Metropolis--Hastings updates of the occupied cluster-specific degrees of freedom. The number of these parameters can change as clusters are created or become empty.

\subsubsection{Updating the cluster labels}
\label{subsubsec:inference-labels}

To update \(z_i\), we temporarily remove observation \(i\) from its current cluster. If observation $i$ was the only member of its cluster, we remove the now-empty cluster. Let \(\mathcal C_{-i}=\{c:\exists j\ne i\text{ such that }z_j=c\}\) denote the clusters occupied by the remaining observations, and let \(K^*=|\mathcal C_{-i}|\). For each \(c\in\mathcal C_{-i}\), define
\(
n_{c,-i}=\sum_{j\ne i}\mathbf{1}(z_j=c)
\),
\(
\mathbf S_{c,-i}=\sum_{j:z_j=c,\,j\ne i}\mathbf W_j.
\)

Let \(M\ge1\) be the number of auxiliary degrees-of-freedom values. If the previous cluster of observation $i$ remains occupied, we draw \(\widetilde\nu_1,\ldots,\widetilde\nu_M \overset{\mathrm{i.i.d.}}{\sim} \operatorname{Uniform}(\nu_L,\nu_U)\). If observation \(i\) is a singleton with degrees-of-freedom parameter \(\nu_{\mathrm{old}}\), we instead set \(\widetilde\nu_1=\nu_{\mathrm{old}}\) and independently draw the remaining \(M-1\) values from the same prior. 

Conditional on these auxiliary parameters, each allocation weight combines an MFM allocation factor with a collapsed predictive density. The weight for assigning observation \(i\) to an existing cluster is
\begin{equation}
w_c=(n_{c,-i}+\gamma)\, p\!\left( \mathbf W_i\mid c,\nu_c, \{\mathbf W_j\}_{j:z_j=c,\,j\ne i}\right),\qquad c\in\mathcal C_{-i},
\label{eq:inf-allocation-existing}
\end{equation}
and the weight for assigning it to the \(r\)-th auxiliary new-cluster candidate is
\begin{equation}
w_{\mathrm{aux},r} = \frac{\gamma}{M} \frac{V_n(K^*+1)}{V_n(K^*)} m(\mathbf W_i\mid\widetilde\nu_r),\qquad r=1,\ldots,M.
\label{eq:inf-allocation-auxiliary}
\end{equation}

In Equation~(\ref{eq:inf-allocation-existing}), for an existing cluster, integrating out its scale matrix gives
\begin{equation}
\begin{aligned}
&p\!\left(\mathbf W_i\mid c,\nu_c,\{\mathbf W_j\}_{j:z_j=c,\,j\ne i}\right)\\
&\quad=\frac{\Gamma_p\!\left(\frac{\kappa_0+(n_{c,-i}+1)\nu_c}{2}\right)}{\Gamma_p\!\left(\frac{\kappa_0+n_{c,-i}\nu_c}{2}\right)\Gamma_p\!\left(\frac{\nu_c}{2}\right)}
\frac{|\mathbf W_i|^{(\nu_c-p-1)/2}|\boldsymbol{\Psi}_0+\mathbf S_{c,-i}|^{(\kappa_0+n_{c,-i}\nu_c)/2}}{|\boldsymbol{\Psi}_0+\mathbf S_{c,-i}+\mathbf W_i|^{(\kappa_0+(n_{c,-i}+1)\nu_c)/2}}.
\end{aligned}
\label{eq:inf-predictive-existing}
\end{equation}

In Equation~(\ref{eq:inf-allocation-auxiliary}), for an auxiliary new-cluster candidate, the corresponding
predictive density is
\begin{equation}
m(\mathbf W_i\mid\widetilde\nu_r)=
\frac{\Gamma_p\!\left(\frac{\widetilde\nu_r+\kappa_0}{2}\right)}{\Gamma_p\!\left(\frac{\widetilde\nu_r}{2}\right)\Gamma_p\!\left(\frac{\kappa_0}{2}\right)}
\frac{|\mathbf W_i|^{(\widetilde\nu_r-p-1)/2}|\boldsymbol{\Psi}_0|^{\kappa_0/2}}{|\mathbf W_i+\boldsymbol{\Psi}_0|^{(\widetilde\nu_r+\kappa_0)/2}},
\label{eq:inf-predictive-new}
\end{equation}
and we give the derivations of (\ref{eq:inf-predictive-existing}) and (\ref{eq:inf-predictive-new}) in Appendix~\ref{sec:mcmc_dev}.

Next, we normalize these \(K^*+M\) weights and sample a destination for observation \(i\). If an existing cluster is selected, observation \(i\) is assigned to that cluster. If auxiliary candidate \(r\) is selected, a new occupied cluster is created with degrees-of-freedom parameter \(\nu_{\mathrm{new}}=\widetilde\nu_r\). All unselected auxiliary values are then discarded. In the singleton case, selecting the first auxiliary candidate retains the current degrees-of-freedom value, while selecting another auxiliary candidate replaces it with a new value.

\subsubsection{Updating the degrees of freedom}
\label{subsubsec:inference-nu}

After completing a full sweep over the cluster labels, we update the degrees-of-freedom parameters conditional on the resulting partition. For each cluster \(c\in\mathcal C\), define \(n_c=\sum_{i=1}^n\mathbf{1}(z_i=c)\), \(\mathbf S_c=\sum_{i:z_i=c}\mathbf W_i\), and \(A_c=\sum_{i:z_i=c}\log|\mathbf W_i|.\)

The independent priors and cluster-specific likelihoods imply that the collapsed full conditional distribution factorizes across occupied clusters:
\begin{equation}
p\!\left(\{\nu_c:c\in\mathcal C\}\mid\boldsymbol z,\{\mathbf W_i\}_{i=1}^n\right)\propto
\prod_{c\in\mathcal C}g_c(\nu_c),
\label{eq:inf-nu-factorization}
\end{equation}
where
\begin{equation}
g_c(\nu_c)=\mathbf{1}\!\left(\nu_c\in[\nu_L,\nu_U]\right)\frac{\Gamma_p\!\left(\frac{\kappa_0+n_c\nu_c}{2}\right)}{\Gamma_p\!\left(\frac{\nu_c}{2}\right)^{n_c}}\exp\!\left\{\frac{\nu_c}{2}\left[A_c-n_c\log|\boldsymbol{\Psi}_0+\mathbf S_c|\right]\right\}.
\label{eq:inf-nu-kernel}
\end{equation}

Equivalently, up to terms that do not depend on \(\nu_c\), the log conditional density within $[\nu_L,\nu_U]$ is
\begin{equation}
\begin{aligned}
\ell_c(\nu_c)=
\log\Gamma_p\!\left(\frac{\kappa_0+n_c\nu_c}{2}\right)-n_c\log\Gamma_p\!\left(\frac{\nu_c}{2}\right)
+\frac{\nu_c-p-1}{2}A_c-\frac{\kappa_0+n_c\nu_c}{2}\log|\boldsymbol{\Psi}_0+\mathbf S_c|.
\end{aligned}
\label{eq:inf-nu-log-kernel}
\end{equation}

Each conditional distribution is univariate but does not have a standard form, so we update the parameters separately using random-walk Metropolis--Hastings steps. For each \(c\in\mathcal C\), we propose \(\nu_c^\star=\nu_c+\varepsilon_c\), where \(\varepsilon_c\sim N(0,\sigma_\nu^2)\). A proposal outside \([\nu_L,\nu_U]\) is rejected. Otherwise, it is accepted with probability \(a_c=\min\!\left\{1,\frac{g_c(\nu_c^\star)}{g_c(\nu_c)}\right\}\). Algorithm~\ref{alg:inference-mfm-wishart} summarizes the allocation and degrees-of-freedom updates.

\begin{algorithm}[h]
\caption{Partially collapsed MFM--Wishart sampler}
\label{alg:inference-mfm-wishart}
\begin{algorithmic}[1]

\Statex \textbf{Input:} Observations $\{\mathbf W_i\}_{i=1}^n$;
prior hyperparameters
$(\gamma,\lambda,\boldsymbol{\Psi}_0,\kappa_0,\nu_L,\nu_U)$;
number of auxiliary values $M$;
number of iterations $N_{\mathrm{iter}}$;
proposal standard deviation $\sigma_\nu$.

\Statex \textbf{Output:} MCMC draws
\(\{(\boldsymbol z^{(s)},
\{\nu_c^{(s)}:c\in\mathcal C^{(s)}\})\}_{s=1}^{N_{\text{iter}}}\).

\State Precompute \(V_n(t)\) for \(t=1,\ldots,n+1\).

\State Initialize a partition and
\(\nu_c\in[\nu_L,\nu_U]\) for each occupied cluster.

\For{\(s=1,\ldots,N_{\text{iter}}\)}
    \For{\(i=1,\ldots,n\)}
        \State Save the current cluster \(c_{\mathrm{old}}\) and its parameter \(\nu_{\mathrm{old}}\), and temporarily remove \(i\). Let \(\mathcal C_{-i}\) be the remaining occupied clusters and set \(K^*=|\mathcal C_{-i}|\).

        \If{\(c_{\mathrm{old}}\in\mathcal C_{-i}\)}
            \State Draw \(\widetilde\nu_1,\ldots,\widetilde\nu_M \overset{\mathrm{i.i.d.}}{\sim} \operatorname{Uniform}(\nu_L,\nu_U)\).
        \Else
            \State Set \(\widetilde\nu_1=\nu_{\mathrm{old}}\) and independently draw \(\widetilde\nu_2,\ldots,\widetilde\nu_M\) from \(\operatorname{Uniform}(\nu_L,\nu_U)\).
        \EndIf

        \State Sample a destination using the normalized weights in \eqref{eq:inf-allocation-existing} and \eqref{eq:inf-allocation-auxiliary}.

        \State Assign \(i\) to the selected existing cluster, or create a new singleton cluster with \(\nu_{\mathrm{new}}=\widetilde\nu_r\) if auxiliary candidate \(r\) is selected.

        \State Remove any empty cluster and its parameter, and discard all unused auxiliary values.

    \EndFor

    \For{each currently occupied cluster \(c\in\mathcal C\)}

        \State Propose \(\nu_c^\star=\nu_c+\varepsilon_c\), where \(\varepsilon_c\sim N(0,\sigma_\nu^2)\).

        \State Reject if \(\nu_c^\star\notin[\nu_L,\nu_U]\); otherwise, accept with probability \(\min\{1,g_c(\nu_c^\star)/g_c(\nu_c)\}\), with \(g_c\) given in \eqref{eq:inf-nu-kernel}.

    \EndFor

    \State Record
    \((\boldsymbol z^{(s)}, \{\nu_c^{(s)}:c\in\mathcal C^{(s)}\})\).

\EndFor
\end{algorithmic}
\end{algorithm}

\subsection{Post-processing of the MCMC samples}

We use Dahl's method \citep{dahl06} to post-process the MCMC samples and infer clustering assignments. Dahl's method was originally proposed in the context of the DPM, and has recently been widely used in MFM-based clustering \citep{yin23,pan24,zhu25}.

For each retained MCMC draw \(l=1,\ldots,L\) after burn-in, let \(\boldsymbol{z}^{(l)}=(z_1^{(l)},\ldots,z_n^{(l)})\) denote the sampled cluster labels and form the corresponding membership matrix \(\mathbf{A}^{(l)}\), where \((\mathbf{A}^{(l)})_{ij}=\mathbbm{1}\!\big(z_i^{(l)}=z_j^{(l)}\big)\). We then estimate the posterior mean membership matrix \(\bar{\mathbf{A}}=\frac{1}{L}\sum_{l=1}^L \mathbf{A}^{(l)}\). Dahl's estimator selects the most representative clustering by choosing the draw whose membership matrix is closest to \(\bar{\mathbf{A}}\) with respect to the squared Frobenius norm, namely,
\begin{equation}
l^\star=\arg\min_{1\le l\le L}\ \sum_{i=1}^n\sum_{j=1}^n\big((\mathbf{A}^{(l)})_{ij}-(\bar{\mathbf{A}})_{ij}\big)^2.
\end{equation}
The posterior estimate of the clustering is then taken as \(\boldsymbol{z}^{(l^\star)}\), and cluster-specific summaries can then be reported based on, or aligned with, this representative draw. We emphasize that Dahl's method only provides a point estimate of the number of clusters and the clustering assignments. If one wants the posterior distribution of the number of clusters \(K_{+,n}\), one can compute \(K_{+,n}\) for each retained MCMC draw and thereby approximate its posterior distribution.

\section{Simulation}\label{sec:simulation}

We conduct simulation studies to evaluate the clustering performance of our MFM--Wishart model under two data-generating mechanisms. In the well-specified settings, observations are generated from finite mixtures of Wishart distributions. In the misspecified settings, observations are lag-0 covariance matrices computed from temporally dependent multivariate time series within each cluster, so the resulting within-cluster distributions are not exactly Wishart. The well-specified settings are described in detail below, while the misspecified setting is briefly discussed in Section~\ref{sec:mis} and detailed in Appendix~\ref{app:misspecified}.

\subsection{Simulation settings}\label{sec:sim_setting}

We consider three matrix-size settings: (1) small matrices of size \(3\times 3\), (2) medium matrices of size \(6\times 6\), and (3) large matrices of size \(12\times 12\). We note that the small matrix size is of practical importance, since in diffusion tensor imaging, a diffusion tensor is represented by a \(3\times3\) SPD matrix estimated at each voxel in the brain \citep{lan21}. The medium- and large-matrix settings are motivated by ROI-based and subnetwork-level functional connectivity studies; see \citet{bul20} and \citet{dil24}.

For the small- and medium-matrix settings, we consider two cases for the true number of clusters: (1) \(k_0=3\) and (2) \(k_0=5\). We also consider two cluster-size configurations: (1) a balanced setting, in which the observations are distributed as evenly as possible across clusters, and (2) an unbalanced setting, in which the cluster sizes are fixed to approximate the proportions \((0.2,0.4,0.4)\) for the \(k_0=3\) scenario and \((0.1,0.1,0.2,0.3,0.3)\) for the \(k_0=5\) scenario. Data are generated from heterogeneous Wishart distributions with the corresponding number of components and cluster-size configuration, and we consider three total sample sizes: \(n=50,100,200\). For the cluster-specific scale matrices \(\mathbf{\Sigma}_k\), we set all diagonal elements equal to \(1\) and vary the off-diagonal elements according to different patterns to create distinct clusters. The cluster-specific scale matrices \(\mathbf{\Sigma}_k\) are visualized in Figure~\ref{fig:small_sigma} for the small-matrix setting and in Figure~\ref{fig:medium_sigma} in the Appendix for the medium-matrix setting. For the degrees-of-freedom parameter \(\nu\), we set \(\nu_1=8\), \(\nu_2=10\) and \(\nu_3=12\) for \(k_0=3\) and \(\nu_1=10\), \(\nu_2=15\), \(\nu_3=20\), \(\nu_4=25\) and \(\nu_5=30\) for \(k_0=5\).

\begin{figure}[h!]
\centering

\begin{subfigure}{0.95\textwidth}
  \centering
  \includegraphics[width=\textwidth]{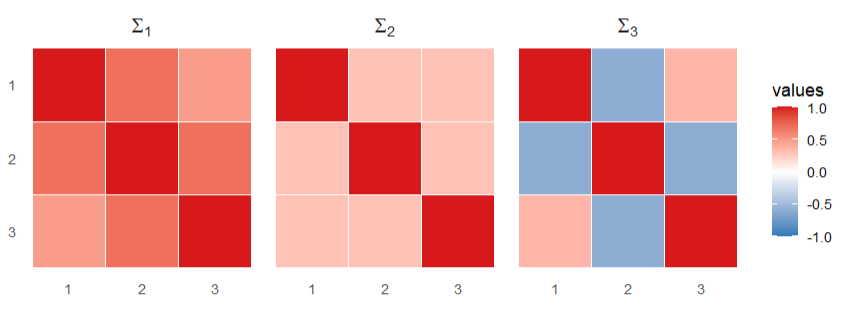}
  \caption{\(k_0=3\): cluster-specific \(\mathbf{\Sigma}_k\).}
\end{subfigure}

\begin{subfigure}{0.95\textwidth}
  \centering
  \includegraphics[width=\textwidth]{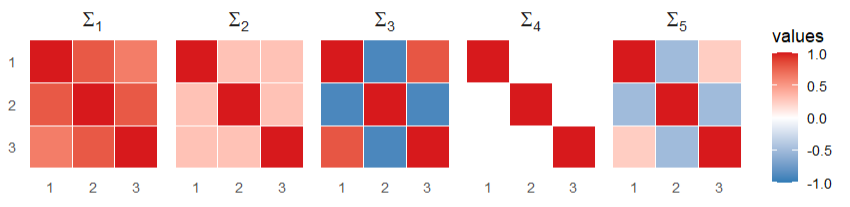}
  \caption{\(k_0=5\): cluster-specific \(\mathbf{\Sigma}_k\).}
\end{subfigure}

\caption{Cluster-specific scale matrix \(\mathbf{\Sigma}_k\) in the small-matrix (\(p=3\)) simulations.}
\label{fig:small_sigma}
\end{figure}

For the large-matrix setting, we consider \(k_0=3\) and a balanced setting. In this setting, we follow the general spirit of \citet{capp25} and impose structured sparsity on the cluster-specific scale matrices. Specifically, \(\mathbf{\Sigma}_1\) and \(\mathbf{\Sigma}_2\) are constructed to have block-sparse correlation structures, as shown in Figure~\ref{fig:large_S}. The third scale matrix, \(\mathbf{\Sigma}_3\), is obtained by standardizing a sample from \(\mathcal{W}_{12}(\mathbf{I}_{12},24)\) into a correlation matrix. Thus, \(\mathbf{\Sigma}_3\) varies across replicates. For each replicate, we then draw the \(12\times12\) observations independently from the corresponding cluster-specific Wishart distributions with corresponding degrees of freedom \(\nu_1=15\), \(\nu_2=15\) and \(\nu_3=20\). We consider three total sample sizes: \(n=50\), \(100\), and \(200\).

For our MFM--Wishart model, as described in Section~\ref{sec:method}, we set \(\gamma=1\) and use a shifted Poisson prior with \(\lambda=1\) on the number of mixture components \(K\), that is, \(K-1\sim\mathrm{Poisson}(1)\). For the inverse-Wishart prior on the cluster-specific scale matrices, we set \(\mathbf{\Psi}_0=\mathbf{I}_p\) and \(\kappa_0=p+2\), so that the prior mean of \(\mathbf{\Sigma}_k\) is \(\mathbf{I}_p\). For the uniform prior on the degrees-of-freedom parameter \(\nu_k\), we use \(\nu_L=p+1\) and \(\nu_U=50\) in all settings.

For each setting, we generate \(100\) replicated datasets. We run a single MCMC chain for \(10{,}000\) iterations, discard the first \(4{,}000\) iterations as burn-in, initialize the chain with \(n\) singleton clusters, and use a Gaussian random-walk proposal with standard deviation \(1\) to update \(\nu_k\) in each of the clusters during the MCMC sampling. We use \(M=10\) auxiliary values for each allocation update. The final clustering estimate for MFM--Wishart is obtained by applying Dahl's method to the retained MCMC samples.

Following \citet{yin23}, we compare our MFM--Wishart model with a Dirichlet process mixture of Wishart kernels, hereafter referred to as DPM--Wishart. Specifically, we replace \(n_{c,-i}+\gamma\) in (\ref{eq:inf-allocation-existing}) with \(n_{c,-i}\), and \((\gamma/M)V_n(K^*+1)/V_n(K^*)\) in (\ref{eq:inf-allocation-auxiliary}) with \(\gamma/M\). Algorithm~\ref{alg:inference-mfm-wishart} is modified accordingly to perform posterior inference under DPM--Wishart. We use the same inverse-Wishart prior for \(\mathbf{\Sigma}_k\), the same uniform prior for \(\nu_k\), and the same MCMC settings as for MFM--Wishart.

We also consider two frequentist model-based approaches: the finite mixture of Wishart distributions \citep{hidot10} and its penalized counterpart \citep{capp25}, hereafter referred to as FMM and penalized FMM, respectively. For both methods, we fit models over a range of candidate values of \(K\) using EM algorithms and select the final model by the Bayesian information criterion (BIC). For the penalized FMM, BIC is also used to select the penalty parameter. The FMM and penalized FMM are fitted using the R packages \texttt{moewishart} \citep{mai26} and \texttt{sparsemixwishart} \citep{capp25}, respectively.

In addition, we consider two distance-based clustering methods, hierarchical clustering (HC) and partitioning around medoids (PAM), using the Riemannian distance as the dissimilarity measure, following \citet{dryden09}. For HC and PAM, the number of clusters is fixed at the true value in each simulation replicate. Further details on the baseline settings are provided in Appendix~\ref{app:baseline}.

\subsection{Evaluation metrics}

Since posterior consistency for the number of clusters is a central motivation for our MFM--Wishart model and is not generally guaranteed for DPM--Wishart, we assess the ability of both methods to recover the true number of clusters. Specifically, for each simulation replicate, we select a representative clustering using Dahl's method and take the number of clusters in this partition as a point estimate of \(K_{+,n}\). We then report the proportion of the \(100\) replicates in which this estimate equals the true number of clusters \(k_0\). This metric measures the accuracy of the cluster-count estimate derived from Dahl's method.

Besides, both Bayesian models also provide a posterior distribution for \(K_{+,n}\). For each model and each replicate, we therefore evaluate the posterior probability assigned to the true number of clusters, \(\mathbb{P}(K_{+,n}=k_0\mid \text{data})\), estimated based on the retained MCMC samples. Higher values are preferred, since they indicate greater posterior support for the true number of clusters. This also aligns with Theorem~\ref{thm:contraction}, and it's useful to check whether the posterior probability of the number of clusters concentrates more on the true value as the sample size increases.

To evaluate clustering performance across all methods, we report the adjusted Rand index (ARI), which measures agreement between the estimated partition and the true partition after adjusting for chance. The ARI is based on pairwise agreement between the true and estimated sample allocations. A higher ARI indicates better clustering performance. An ARI of \(1\) indicates perfect recovery of the true clustering, an ARI close to \(0\) indicates chance-level agreement, and a negative ARI indicates performance worse than chance.

\subsection{Results of small- and medium-matrix settings}

\subsubsection{Estimating the number of clusters} \label{sec:est_cluster_num}

\begin{figure}[!h]
    \centering
    \includegraphics[width=\linewidth]{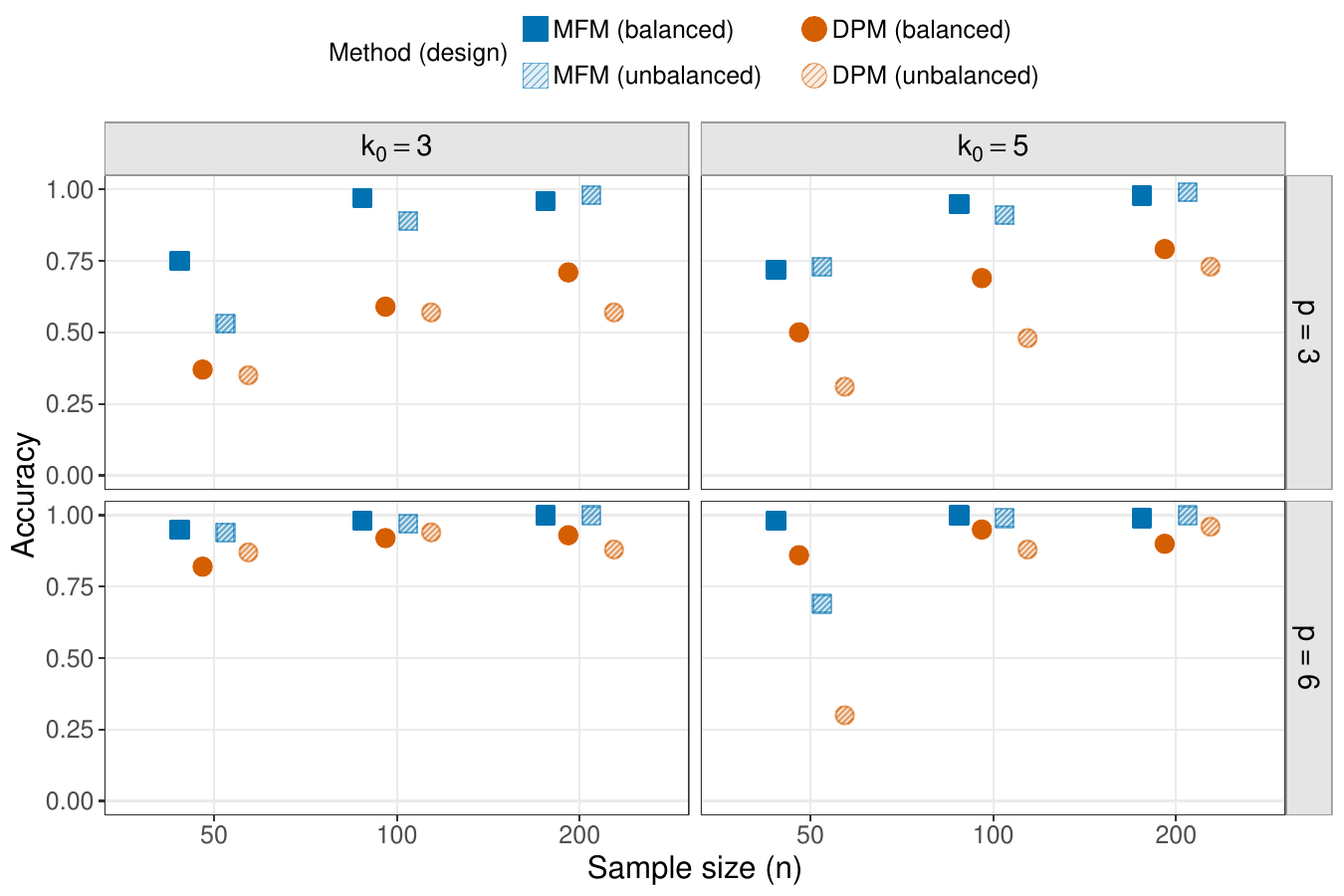}
    \caption{Accuracy of recovering the true number of clusters using Dahl's representative clustering over \(100\) simulation replicates. Columns correspond to \(k_0=3\) and \(k_0=5\), and rows correspond to \(p=3\) and \(p=6\). Blue squares and orange circles represent MFM--Wishart and DPM--Wishart, respectively. Solid and hatched markers indicate balanced and unbalanced designs, respectively.}
    \label{fig:acc}
\end{figure}

Figure~\ref{fig:acc} compares the accuracy of recovering the true number of clusters using Dahl's representative clustering. MFM--Wishart achieves higher recovery rates than DPM--Wishart across all combinations of conditions considered. The accuracy generally improves with sample size. The advantage is particularly pronounced for \(p=3\), where the accuracy of DPM--Wishart remains below \(0.8\) even at \(n=200\). For \(p=6\), DPM--Wishart performs more competitively, although its recovery rates remain below those of MFM--Wishart. Unequal cluster sizes can substantially reduce recovery accuracy in small samples. Nevertheless, MFM--Wishart achieves nearly perfect recovery in this setting once \(n=200\). These results provide empirical evidence of the stronger finite-sample performance of MFM--Wishart in estimating the number of clusters under the simulation settings considered.

We further investigate whether the favorable performance of MFM--Wishart can be solely explained by the MFM prior assigning greater probability to the true number of clusters. We also examine whether the finite-sample results exhibit the concentration of posterior mass at the true cluster count suggested by Theorem~\ref{thm:contraction}. Figure~\ref{fig:prior_prob} shows the induced prior probabilities of the event \(\{K_{+,n}=k_0\}\) under the MFM and DPM priors. When \(k_0=3\), the MFM prior assigns a probability of approximately \(0.18\) to this event across the sample sizes considered, whereas the corresponding DPM probability decreases from approximately \(0.18\) to \(0.08\) as \(n\) increases from \(50\) to \(200\). When \(k_0=5\), however, the DPM prior assigns substantially greater probability to the true cluster count, with probabilities around \(0.19\)--\(0.21\), compared with approximately \(0.01\)--\(0.02\) under the MFM prior.

Despite this prior disadvantage when \(k_0=5\), MFM--Wishart generally assigns greater posterior probability to the true number of clusters, as shown in Figure~\ref{fig:posterior_prob}. Its posterior probabilities typically increase and become less variable across simulation replicates as the sample size grows, with most values close to one at \(n=200\). The difference between the two models is particularly pronounced for \(p=3\), where DPM--Wishart continues to assign appreciably less posterior mass to the true cluster count. For \(p=6\), both models tend to assign higher posterior probabilities to the true count, although substantial variability remains in some small-sample settings, particularly under the unbalanced design. These findings indicate that the advantage of MFM--Wishart cannot be explained simply by greater marginal prior mass at the true cluster count. They also provide finite-sample evidence of increasing posterior concentration at the true value under MFM--Wishart.

\begin{figure}[!htbp]
    \centering

    \begin{subfigure}{0.85\textwidth}
        \centering
        \includegraphics[width=\linewidth]{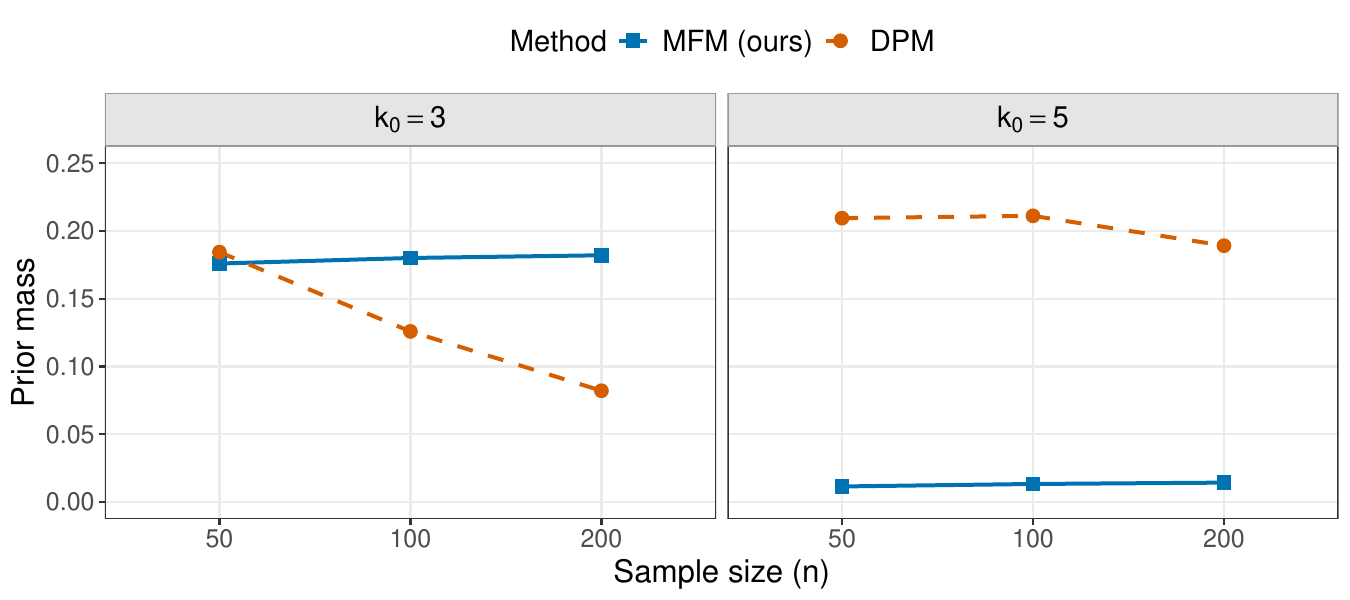}
        \caption{Induced prior probability \(\mathbb{P}(K_{+,n}=k_0)\) under the MFM and DPM priors across sample sizes \(n\).}
        \label{fig:prior_prob}
    \end{subfigure}

    \vspace{0.5em}

    \begin{subfigure}{0.95\textwidth}
        \centering
        \includegraphics[width=\linewidth]{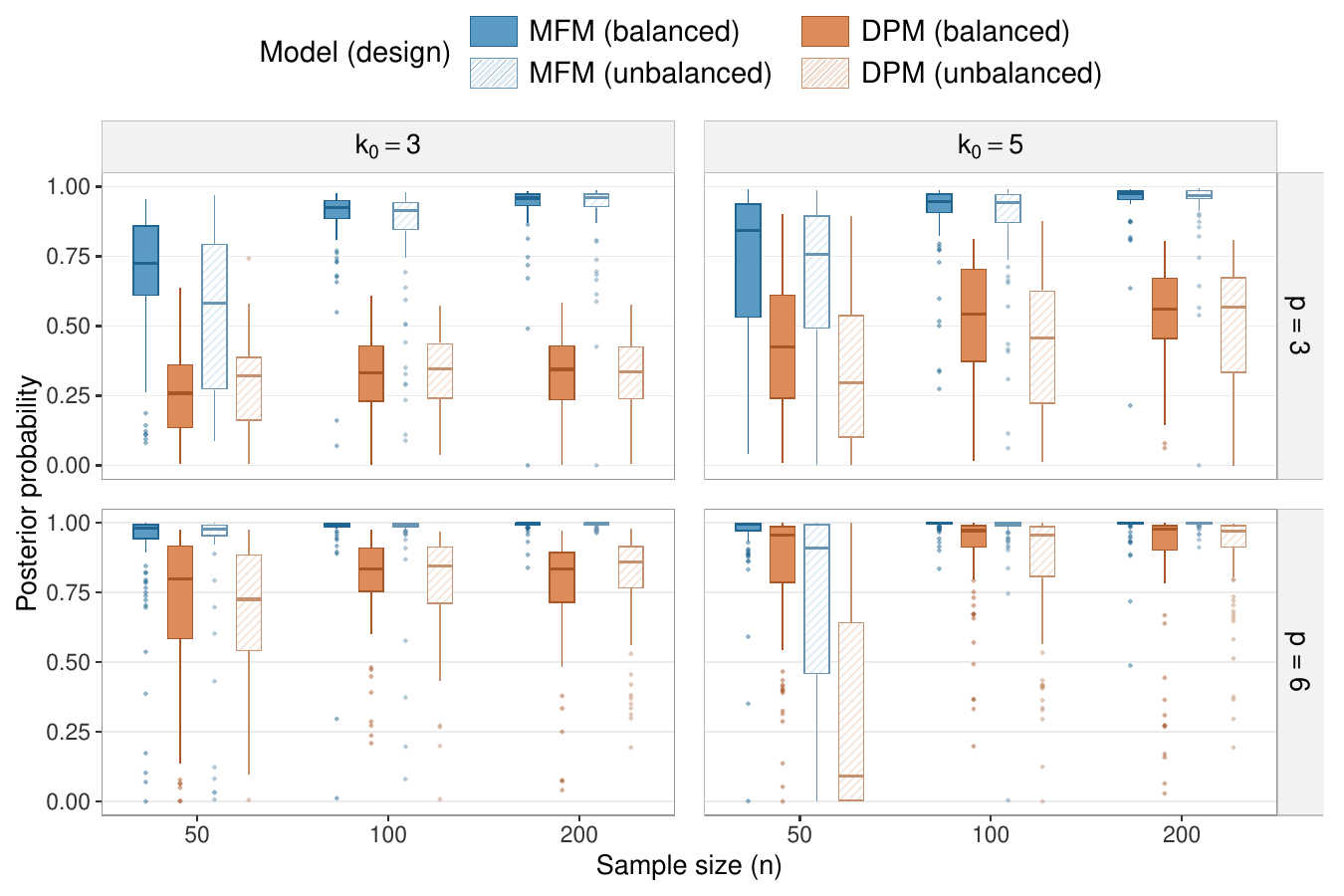}
        \caption{Posterior probability \(\mathbb{P}(K_{+,n}=k_0\mid\text{data})\) under MFM--Wishart and DPM--Wishart. Boxplots summarize these probabilities across \(100\) simulation replicates for each setting. Solid and hatched boxes indicate balanced and unbalanced designs, respectively.}
        \label{fig:posterior_prob}
    \end{subfigure}

    \caption{Prior and posterior probabilities assigned to the true number of clusters under MFM--Wishart and DPM--Wishart.}
    \label{fig:prior_posterior}
\end{figure}

\subsubsection{Clustering performance}

Figures~\ref{fig:ari-small-medium-balanced} and~\ref{fig:ari-small-medium-unbalanced} compare the clustering performance of the six methods under balanced and unbalanced cluster-size configurations. MFM--Wishart delivers competitive ARI values across the settings considered, with performance broadly comparable to that of DPM--Wishart. Its clustering accuracy generally improves with sample size, and its ARI values are close to one in the \(p=6\) settings when \(n\geq100\). These results show that MFM--Wishart maintains competitive partition recovery under both cluster-size configurations while also providing the advantages in cluster-count estimation discussed above.

The penalized FMM performs less favorably in several settings, particularly in the balanced scenarios with \(k_0=5\). Examination of the selected models reveals that it frequently selects fewer clusters than the true number. One possible contributing factor is the absence of pronounced sparsity in the component scale matrices in these small- and medium-dimensional scenarios, which may limit the benefits of sparsity-inducing regularization.

\begin{figure}[h!]
    \centering
    \includegraphics[width=0.98\linewidth]{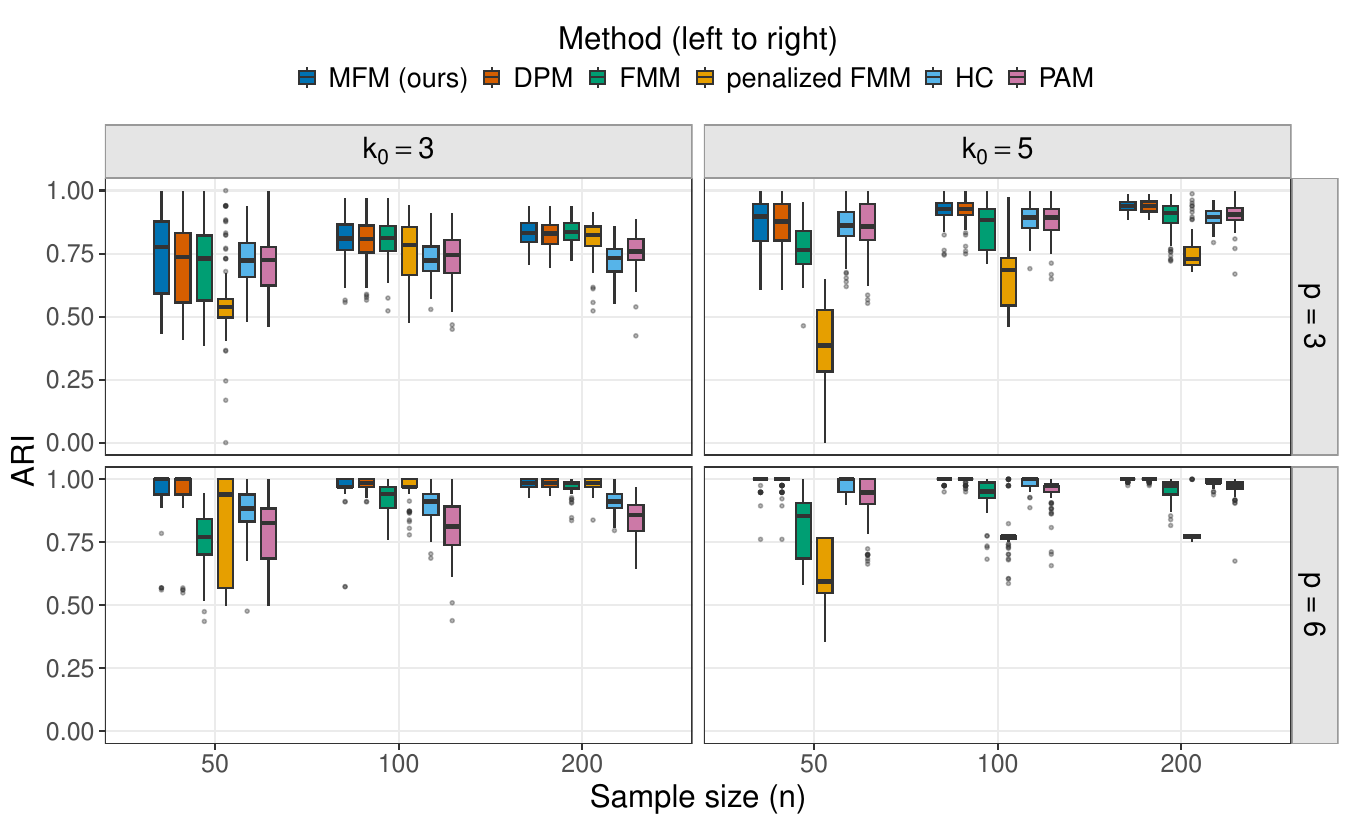}
    \caption{Adjusted Rand index (ARI) under the balanced cluster-size configuration. Columns correspond to \(k_0=3\) and \(k_0=5\), and rows correspond to \(p=3\) and \(p=6\). The x-axis shows the sample size \(n\). Each boxplot summarizes the ARI across \(100\) simulation replicates.}
    \label{fig:ari-small-medium-balanced}
\end{figure}

\begin{figure}[h!]
    \centering
    \includegraphics[width=0.98\linewidth]{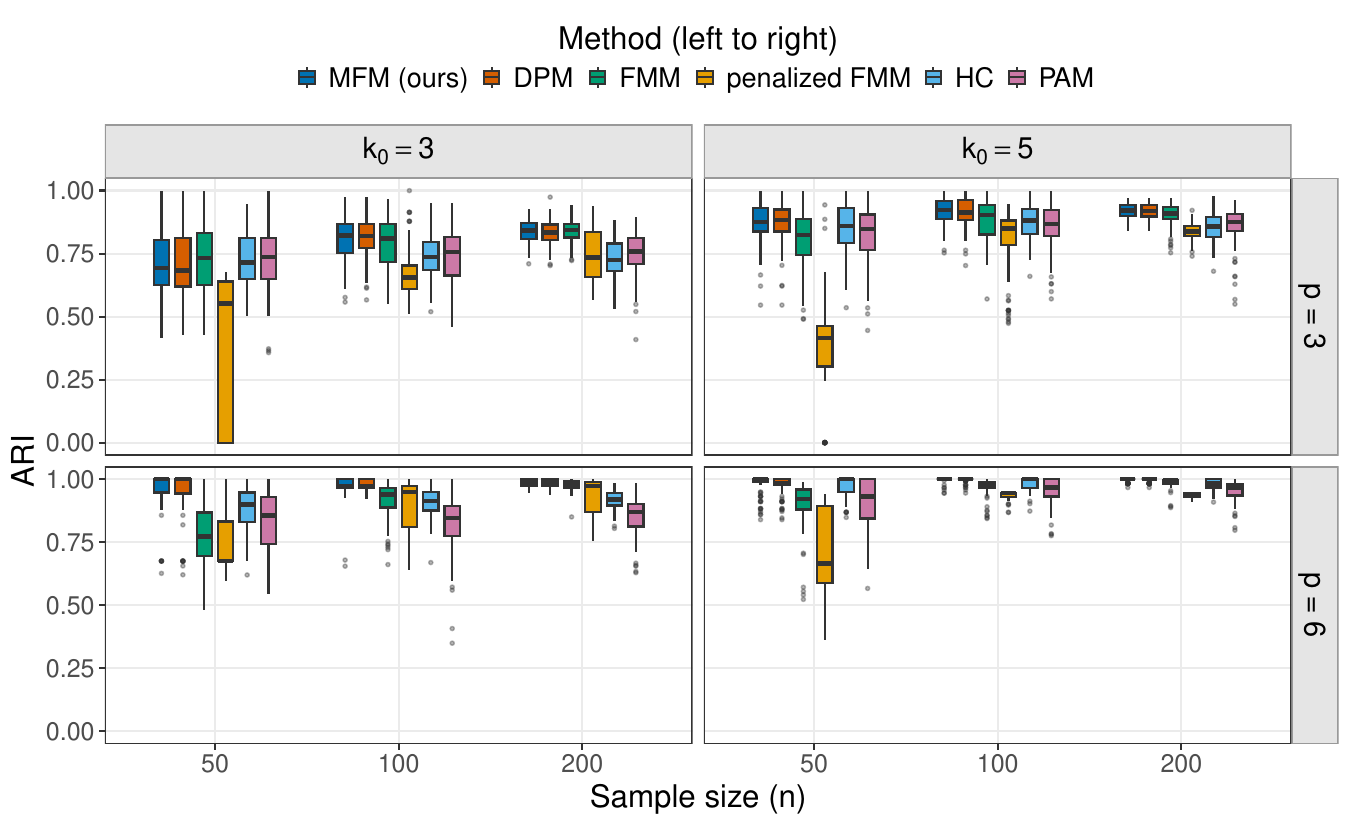}
    \caption{Adjusted Rand index (ARI) under the unbalanced cluster-size configuration. Columns correspond to \(k_0=3\) and \(k_0=5\), and rows correspond to \(p=3\) and \(p=6\). The x-axis shows the sample size \(n\). Each boxplot summarizes the ARI across \(100\) simulation replicates.}
    \label{fig:ari-small-medium-unbalanced}
\end{figure}

\subsection{Results for the large-matrix setting}

Table~\ref{tab:large_result} reports the mean and standard deviation of the ARI across \(100\) simulation replicates in the large-matrix setting. MFM--Wishart and DPM--Wishart exhibit similar clustering performance, with MFM--Wishart attaining slightly higher mean ARI at \(n=50\) and \(n=100\). At \(n=50\), MFM--Wishart achieves a mean ARI of \(0.901\), close to the highest value of \(0.903\) attained by HC. The penalized FMM achieves the highest mean ARI at \(n=100\) and \(n=200\). Its strong performance at these sample sizes may reflect the benefits of sparsity-inducing regularization under the sparse scale-matrix structure of this setting.

\begin{table}[!htbp]
\centering
\caption{Clustering performance in the large-matrix setting
(\(p=12\)). Entries report the mean and standard deviation (SD)
of the ARI across \(100\) simulation replicates. For each sample size, the highest mean ARI is shown in bold and the second-highest is underlined.}
\label{tab:large_result}
\small
\begin{tabular}{l|cc|cc|cc|cc}
\hline
& \multicolumn{2}{c|}{\(n=50\)} & \multicolumn{2}{c|}{\(n=100\)} & \multicolumn{2}{c|}{\(n=200\)} & \multicolumn{2}{c}{Overall}\\
\cline{2-9}
Method & Mean & SD & Mean & SD & Mean & SD & Mean & SD  \\
\hline
MFM (ours)    & \underline{0.901} & 0.169 & \underline{0.961} & 0.115 & \underline{0.990} & 0.043 & \textbf{0.951} & 0.126  \\
DPM           & 0.885 & 0.177 & 0.957 & 0.122 & \underline{0.990} & 0.043 & 0.944 & 0.133   \\
FMM           & 0.672 & 0.164 & 0.892 & 0.097 & 0.975 & 0.023 & 0.847 & 0.169 \\
Penalized FMM & 0.848 & 0.194 & \textbf{0.992} & 0.016 & \textbf{0.996} & 0.008 & \underline{0.945} & 0.132 \\
HC            & \textbf{0.903} & 0.090 & 0.915 & 0.055 & 0.926 & 0.044 & 0.914 & 0.067 \\
PAM           & 0.732 & 0.120 & 0.739 & 0.103 & 0.756 & 0.078 & 0.743 & 0.102  \\
\hline
\end{tabular}
\end{table}

For estimation of the number of clusters, Table~\ref{tab:large-matrix-k-accuracy} in the appendix shows that both MFM--Wishart and DPM--Wishart recover the true number of clusters with high accuracy. The two methods exhibit similar recovery accuracy across the sample sizes considered, with accuracy increasing toward one as the sample size increases.

Although MFM--Wishart does not explicitly impose sparsity on the component scale matrices, it remains competitive across all three sample sizes. Across the three sample sizes, MFM--Wishart attains the highest overall mean ARI of \(0.951\), followed by the penalized FMM at \(0.945\). Extending MFM--Wishart to accommodate sparsity in the component scale matrices is a promising direction for future work.

\subsection{Computational efficiency}
\label{sec:computational-efficiency}

We also evaluate the MCMC computation time of MFM--Wishart and DPM--Wishart. The two samplers use the same Wishart likelihood, the same collapsed predictive densities, the same update for the degrees-of-freedom parameters \(\nu_k\), and the same MCMC settings. The MFM-specific term \(V_n(K^*+1)/V_n(K^*)\) is precomputed, so the additional cost of using the MFM prior is only a scalar lookup when updating cluster labels.

In our simulations, the running time is mainly affected by the sample size \(n\) and, to a lesser extent, by the number of clusters. Larger \(n\) requires more cluster-label updates per MCMC iteration, while more clusters require more predictive likelihood evaluations for each label update. By contrast, within the small- and medium-matrix settings considered here, increasing the matrix dimension from \(p=3\) to \(p=6\) has little visible effect on computation time.

The results in Appendix~\ref{app:add_sim} support this interpretation. Figure~\ref{fig:mcmc-time-balanced-small-medium} shows that MFM--Wishart and DPM--Wishart have similar runtime under the balanced cluster-size configuration. The large-matrix results in Table~\ref{tab:large-matrix-mcmc-time} show the same pattern. For both models, runtime increases with sample size, consistent with the sample-wise allocation updates used in our MCMC algorithms. However, runtime seems to be stable with matrix dimension and the true number of clusters. These findings support the computational feasibility of MFM–Wishart, although further experiments would be needed to assess its scalability to substantially higher dimensions.

\subsection{Misspecified setting with temporal dependence}\label{sec:mis}

As an additional robustness check, we consider a misspecified setting in which the observed covariance-type matrices are computed from temporally dependent multivariate time series, so that the within-cluster distributions are not exactly Wishart. The detailed simulation design is given in Appendix~\ref{app:misspecified}. The resulting performance is broadly similar to that in the well-specified setting. MFM--Wishart remains generally competitive in ARI, while showing more reliable recovery of the true number of clusters than DPM--Wishart.

\section{Application to Infant fNIRS Functional Connectivity}
\label{sec:application}

\subsection{Data and construction of functional connectivity matrices}
\label{sec:app_data}

We analyze a publicly available task-free fNIRS dataset \citep{blanco22} of hemodynamic activity in 4-month-old infants during natural sleep. fNIRS is a noninvasive optical neuroimaging technique that uses near-infrared light to monitor cortical hemodynamics as a proxy for neural activity. The measurements allow estimation of relative concentration changes in oxyhemoglobin (HbO) and deoxyhemoglobin (HbR).

The original dataset contains task-free fNIRS recordings from 104 healthy, full-term infants. Each recording session lasted between 9 and 25 minutes and began after clear signs of natural sleep were observed. Five infants were excluded from the original analysis because of insufficient data quality, resulting in a final sample of 99 infants (female: 51; male: 48) from three language backgrounds: 30 Spanish monolingual, 33 Basque monolingual, and 36 Basque--Spanish bilingual infants. Data were collected using a NIRx NIRScout system with 16 sources and 24 detectors, yielding 52 channels for each hemoglobin signal. The system used two wavelengths, 760 and 850 nm, and sampled at 8.93 Hz. Optodes were placed on an Easycap according to the international 10--20 system, with source--detector separations of approximately 20--45 mm, covering bilateral frontal, temporal, parietal, and occipital regions.

We use the preprocessed HbO data for the final sample of 99 infants from \citet{blanco22}. The data are openly accessible at \url{https://osf.io/7fzkm/overview}. The preprocessed signals are provided in the folder \texttt{mat}, the participant information is provided in \texttt{BCBL\_RS4\_participant-info.csv}, and the data-format description is provided in \texttt{BCBL\_RS4\_osf-data-format.rtf}. Interested readers can refer to \citet{blanco22} for details of the preprocessing pipeline. For each infant, the multivariate HbO time series across all channels was manually cropped to a common length of \(T=5{,}000\) samples, corresponding to approximately 560 seconds, so that all infants contributed equally to downstream analyses.

Previous studies using this dataset have examined group differences and subnetwork-level patterns in infant brain activity. For example, \citet{wang23} reported sex differences in brain activity, with stronger differences in the frontoparietal, somatomotor, visual, and dorsal networks. In their functional-connectivity analysis of the same infant sleep fNIRS dataset, \citet{wang23} computed pairwise Pearson correlations between channel time courses and displayed group-level correlation adjacency matrices. These matrices descriptively suggested weaker overall connectivity in male infants than in female infants, with the difference appearing most pronounced around channels 40--46. However, this observation was primarily based on visual inspection of the correlation adjacency matrices.

Motivated by this, we conduct a targeted secondary analysis of the subject-specific connectivity matrices for channels 40--46. We use Pearson correlation matrices rather than covariance matrices for two reasons. First, this choice maintains direct comparability with the correlation-based functional-connectivity analysis of \citet{wang23}. Second, correlations standardize out channel-specific and subject-specific marginal scale differences in the HbO signals, allowing the clustering analysis to focus on scale-free co-fluctuation patterns among channels. Our goal is not to directly replicate the group-level sex-difference analysis of \citet{wang23}. Instead, we ask whether subject-specific connectivity patterns in this region form latent heterogeneous clusters, and then evaluate whether the resulting clusters are associated with recorded sex as an external post hoc covariate.

For infant \(i=1,\ldots,99\), let \(\boldsymbol{x}_i(t)\in\mathbb{R}^{p}\), \(t=1,\ldots,T\), denote the vector of HbO signals at the selected channels at time \(t\), where \(p=7\) for channels 40--46 and \(T=5{,}000\). Let
\(
\bar{\boldsymbol{x}}_i
=
\frac{1}{T}\sum_{t=1}^{T}\boldsymbol{x}_i(t)
\)
denote the sample mean vector. The lag-0 sample covariance matrix is computed as
\(
\widehat{\mathbf{\Sigma}}_i
=
\frac{1}{T-1}
\sum_{t=1}^{T}
\bigl(\boldsymbol{x}_i(t)-\bar{\boldsymbol{x}}_i\bigr)
\bigl(\boldsymbol{x}_i(t)-\bar{\boldsymbol{x}}_i\bigr)^{\top}.
\)
The corresponding sample correlation matrix is
\(
\widehat{\mathbf{R}}_i
=
\mathbf{D}_i^{-1/2}
\widehat{\mathbf{\Sigma}}_i
\mathbf{D}_i^{-1/2}
\)
, where
\(\mathbf{D}_i
=
\operatorname{diag}\!\left\{
(\widehat{\mathbf{\Sigma}}_i)_{11},\ldots,
(\widehat{\mathbf{\Sigma}}_i)_{pp}
\right\}.
\)
We take
\(
\mathbf{W}_i=\widehat{\mathbf{R}}_i
\)
, 
\(
i=1,\ldots,99,
\)
as the matrix-valued observations in the MFM--Wishart model. Although Pearson correlation matrices are not exactly Wishart-distributed because of their unit-diagonal constraint, using a Wishart likelihood as a working model for correlation-based functional-connectivity matrices has precedent in prior Wishart-mixture analyses \citep{tokuda21,tokuda21psy,capp25}.

\subsection{MFM--Wishart analysis}
\label{sec:app_analysis}

We use the same MFM prior settings as in the simulation studies, namely \(\gamma=1\) and \(K-1\sim\mathrm{Poisson}(1)\). For the cluster-specific degrees of freedom, we assign independent priors \(\nu_k\sim\mathrm{Uniform}(10,100)\), independently of the component scale matrices. The prior mean of each \(\nu_k\) is \(\nu_0=55\), and the matrix dimension is \(p=7\). For the inverse-Wishart prior \(\mathbf{\Sigma}_k\sim\mathcal{IW}_p(\mathbf{\Psi}_0,\kappa_0)\), we set \(\kappa_0=12\) and
\( \mathbf{\Psi}_0=(\kappa_0-p-1)\frac{1}{\nu_0}\mathbf{I}_p=
\frac{4}{55}\mathbf{I}_7.\)
These choices imply
\(
\mathbb{E}(\mathbf{\Sigma}_k)=\frac{1}{55}\mathbf{I}_7,\)
\(
\mathbb{E}(\nu_k\mathbf{\Sigma}_k)=\mathbf{I}_7.
\)
Thus, the prior mean of each component mean matrix is the identity matrix, providing a baseline with unit diagonal entries and zero off-diagonal entries.

Posterior computation follows the MCMC algorithm in Section~\ref{sec:inference}, with Dahl's method used to obtain a representative clustering. We run a single MCMC chain for \(10{,}000\) iterations, discard the first \(4{,}000\) iterations as burn-in, and use a Gaussian random-walk proposal with standard deviation \(1\) to update each occupied cluster's degrees-of-freedom parameter \(\nu_k\). We use \(M=10\) auxiliary values for each allocation update.

\begin{figure}[h!]
\centering
\includegraphics[width=0.98\textwidth]{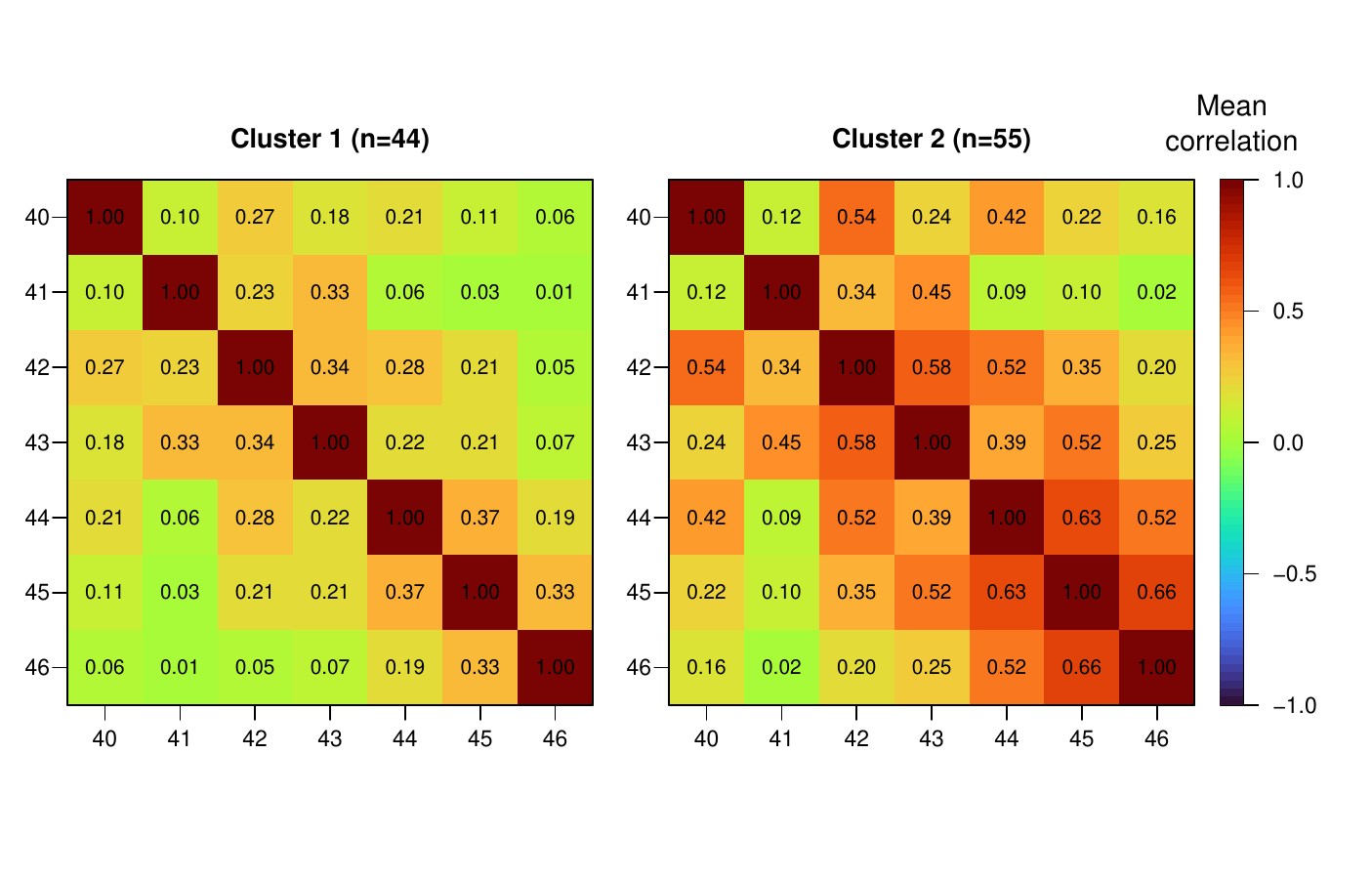}
\caption{Cluster-wise mean correlation matrices for the Dahl partition obtained by applying MFM--Wishart to the \(7\times 7\) correlation submatrices for channels 40--46. Cluster 2 exhibits stronger positive connectivity throughout this region.}
\label{fig:app_channels4046_mean}
\end{figure}

Applying MFM--Wishart to the \(7\times7\) correlation submatrices yields a Dahl representative clustering with two clusters containing \(44\) and \(55\) infants, respectively. Figure~\ref{fig:app_channels4046_mean} displays the cluster-wise mean correlation matrices, obtained by averaging the observed correlation matrices within each inferred cluster. Both cluster averages exhibit positive correlations among channels 40--46, with Cluster 2 showing systematically stronger correlations than Cluster 1. The representative clustering therefore primarily distinguishes infants by overall connectivity strength within this subnetwork. The degrees-of-freedom values in the MCMC draw selected by Dahl’s method are \(24.52\) and \(28.43\) for Clusters 1 and 2, respectively.

Figure~\ref{fig:application-k-posterior} displays the estimated posterior distribution of the number of occupied clusters. The posterior mode is \(K_{+,n}=2\), consistent with the two-cluster Dahl representative clustering. The estimated posterior probabilities for two, three, and four clusters are approximately \(0.74\), \(0.24\), and \(0.02\), respectively. Thus, although two clusters receive the greatest posterior support, uncertainty remains between two and three clusters. By assigning posterior probabilities to alternative cluster counts, MFM--Wishart quantifies uncertainty about the number of clusters under the assumed model and priors.

\begin{figure}[!htbp]
    \centering
    \includegraphics[width=0.70\textwidth]{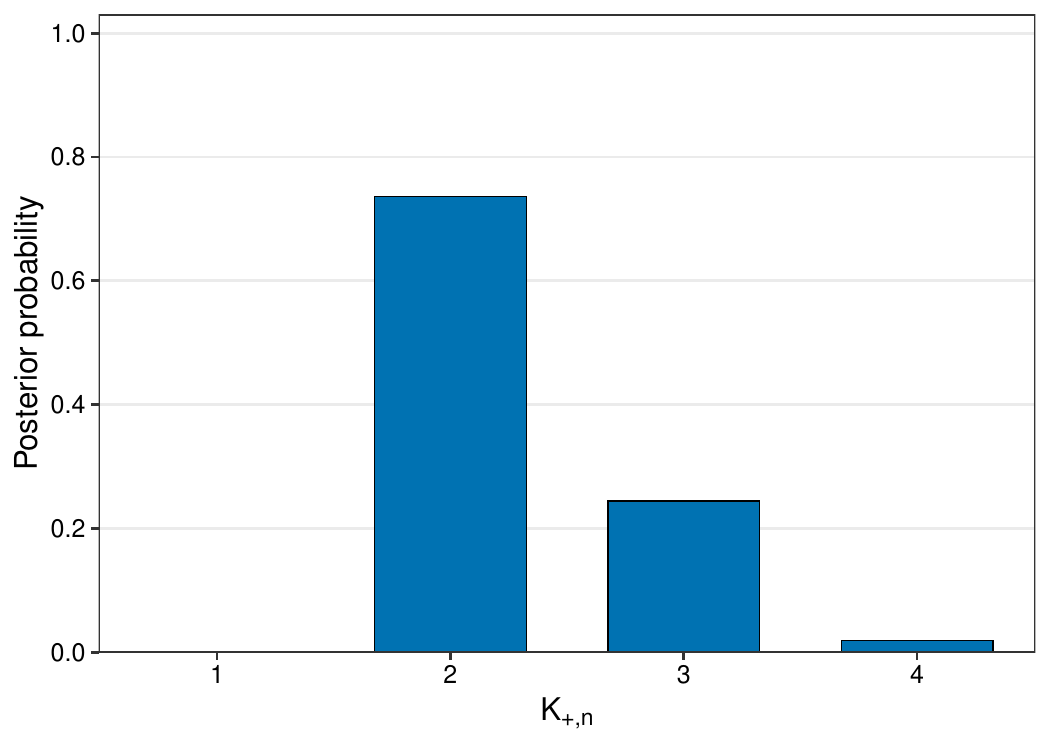}
    \caption{Estimated posterior distribution of the number of occupied clusters \(K_{+,n}\) under MFM--Wishart in the application analysis. Bar heights represent relative frequencies calculated from \(6{,}000\) post-burn-in MCMC draws.}
    \label{fig:application-k-posterior}
\end{figure}

To examine the MCMC behavior in the application, Appendix~\ref{app:application_trace} presents MCMC diagnostics, including trace plots of \(K_{+,n}\) and label-invariant summaries of the cluster-specific degrees of freedom, together with effective sample size estimates. These diagnostics revealed no obvious signs of poor mixing in the monitored quantities. We note that, to assess the sensitivity of the MCMC output to initialization and to examine whether the reported clustering was driven by a particular random seed or starting configuration, we repeated the application analysis 100 times using different random seeds. In each repeated run, the initial number of clusters was randomly selected from \(\{1,\ldots,99\}\). All \(100\) runs produced the same representative clustering, up to label permutation, providing empirical support for the stability of the reported partition.

As an exploratory post hoc analysis, we examine the association between the inferred clusters and sex. Table~\ref{tab:app_channels4046_gender} reports the corresponding contingency table. Cluster 1 contains \(25\) female and \(19\) male infants, whereas Cluster 2 contains \(26\) female and \(29\) male infants. A two-sided Fisher's exact test yields a \(p\)-value of \(0.420\), providing no statistically significant evidence of an association between this representative clustering and sex.

\begin{table}[h!]
\centering
\caption{Sex composition of the inferred clusters in channels 40--46.}
\label{tab:app_channels4046_gender}
\begin{tabular}{lcc}
\hline
 & Female & Male \\
\hline
Cluster 1 (\(n=44\)) & 25 & 19 \\
Cluster 2 (\(n=55\)) & 26 & 29 \\
\hline
\end{tabular}
\end{table}

Overall, this analysis provides a complementary perspective on the earlier qualitative observation. Modeling the full \(7\times7\) connectivity pattern jointly identifies a partition characterized primarily by differences in overall connectivity strength, while the exploratory comparison does not provide evidence that the inferred groups are associated with sex.

\section{Conclusion}\label{sec:conclusion}

In this paper, we proposed MFM--Wishart, a Bayesian model-based clustering approach for covariance and correlation matrices. By combining Wishart mixture components with an MFM prior, the proposed model enables joint posterior inference on the clustering assignments and the number of clusters. Although DPM--Wishart can also be used for the joint inference, it does not in general provide posterior consistency for the number of clusters when the data are generated from a finite mixture. This motivates the MFM formulation in our setting. On the theoretical side, we studied Wishart kernels in the context of mixture models and established posterior consistency for the number of components together with posterior contraction for the mixing measure under standard regularity conditions. We also developed an efficient MCMC algorithm for posterior computation.

From a high-level modeling perspective, the MFM prior assumes a finite but unknown number of mixture components, whereas the DPM prior allows a countably infinite collection of potential components. As argued in \citet{miller18}, when the scientific question suggests a finite collection of meaningful latent subtypes, the MFM formulation may be a more natural choice and offers more direct prior control over the number of clusters. By contrast, under a DPM prior, the induced behavior of the number of clusters is controlled only indirectly through the concentration parameter and its interaction with the sample size.

The simulation studies showed that MFM--Wishart provides competitive clustering performance across a range of settings, while yielding more reliable recovery of the true number of clusters than DPM--Wishart when the data are generated from a finite mixture. The misspecified-case simulations further suggest that the method remains effective for clustering when the Wishart likelihood is interpreted as a working model rather than an exact distributional assumption. In the application to infant fNIRS functional connectivity, MFM--Wishart identified interpretable heterogeneity in subject-specific connectivity patterns, illustrating the practical value of the model for matrix-valued data analysis. In practice, MFM--Wishart may also be useful as an exploratory tool for suggesting a plausible number of clusters or an initial partition to be refined by other procedures.

For the current MFM--Wishart model, several directions merit further study. One is to extend MFM--Wishart to incorporate explicit sparsity regularization in the cluster-specific scale matrices, which may be beneficial in higher-dimensional settings and for better interpretability. Another is to incorporate spatial dependence into the model, as considered in \citet{lan21} and \citet{zhu25}, which may be particularly relevant for applications such as clustering voxels in diffusion tensor imaging. Beyond neuroimaging, the proposed framework may also be applied to other covariance-type data, such as covariance matrices derived from multivariate intensive longitudinal data.

Extensions beyond the Wishart kernel are also worth exploring. One natural extension is an MFM model with inverse-Wishart kernels for clustering precision-type matrices. For correlation matrices, it would also be of interest to replace the Wishart kernel, which serves as a working likelihood in our study, with a distribution supported directly on the space of SPD matrices with unit diagonals, such as an LKJ-type construction \citep{lew09}, although developing a sufficiently flexible model and an efficient MCMC algorithm for such nonconjugate mixtures remains challenging.

\vspace{0.5cm}
\paragraph*{Code availability}

The code used to implement the MFM--Wishart model, reproduce the simulation studies, and carry out the infant fNIRS functional-connectivity application is available at \url{https://github.com/Zongyu-Li/MFM-Wishart}.

\newpage
\bibliographystyle{asa}
\bibliography{references}

\newpage

\begin{appendices}

\renewcommand{\thefigure}{\thesection.\arabic{figure}}
\renewcommand{\thetable}{\thesection.\arabic{table}}
\counterwithin{figure}{section}
\counterwithin{table}{section}

\section{Derivation of the Hessian of the Wishart Distribution}\label{sec:Hessian}

\subsection{Notation and facts}

We consider the following Wishart density parameterized by the precision matrix \(\mathbf{\Lambda}\) and the degrees of freedom \(\nu\):

\[
f(\mathbf{W}\mid \mathbf{\Lambda},\nu)=\frac{|\mathbf{\Lambda}|^{\nu/2}}{2^{\nu p/2}\Gamma_p(\nu/2)}|\mathbf{W}|^{(\nu-p-1)/2}\exp\left\{-\frac{1}{2}\tr(\mathbf{\Lambda}\mathbf{W})\right\},\;\mathbf{W}\in\mathbb{S}_{++}^p
\]
where the parameters \(\mathbf{\Lambda}\in\mathbb{S}_{++}^p\) and \(\nu>p-1\).

Next, we consider the reparameterization stated in Section~\ref{sec:theory}, where \(\boldsymbol{\eta}:=\vech{(\mathbf{\Lambda})}\in \mathbb{R}^{d}\) with \(d=p(p+1)/2\). We let \(D_p\in\mathbb{R}^{p^2\times d}\) be the duplication matrix such that \(\text{vec}(\mathbf{S})=D_p\vech{(\mathbf{S})}\) for all \(\mathbf{S}\in\mathbb{S}_{++}^p\), where \(\text{vec}(\mathbf{S})\) stacks \(\mathbf{S}\) by columns into a vector. We denote by \(\mathbb{S}^p\) the space of all \(p\times p\) symmetric matrices.

So, the reparameterized parameter space is
\[
\Theta := \left\{(\boldsymbol{\eta},\nu):\mathbf{\Lambda}(\boldsymbol{\eta})\in\mathbb{S}_{++}^p,\nu>p-1\right\}
\]
where we write \(\mathbf{\Lambda}(\boldsymbol{\eta})\) as the unique matrix \(\mathbf{\Lambda}\) such that \(\vech{(\mathbf{\Lambda})}=\boldsymbol{\eta}\).

We now state some facts from matrix calculus and matrix algebra that will be used. For conformable matrices, the following identities will be used:
\begin{align}
& d\,\text{vec}(\mathbf{\Lambda}) = D_p\, d \boldsymbol{\eta}, \label{fact:vech}\\
&d\log|\mathbf{A}|=\tr(\mathbf{A}^{-1}\,d\mathbf{A}), \label{fact:logdet}\\
&d\,\tr(\mathbf{A}\mathbf{X})=\tr(\mathbf{X}\,d\mathbf{A})\quad(\text{with }\mathbf{X}\ \text{fixed}), \label{fact:trace}\\
&\tr(\mathbf{M}^\top \mathbf{N})=\operatorname{vec}(\mathbf{M})^\top \operatorname{vec}(\mathbf{N}), \label{fact:vec-tr}\\
&\operatorname{vec}(\mathbf{A}\mathbf{X}\mathbf{B})=(\mathbf{B}^\top\otimes \mathbf{A})\,\operatorname{vec}(\mathbf{X}), \label{fact:vec-AXB}\\
&d(\mathbf{A}^{-1})=-\mathbf{A}^{-1}(d\mathbf{A})\mathbf{A}^{-1}, \label{fact:inv}
\end{align}
where \(\otimes\) denotes the Kronecker product. Furthermore, by (\ref{fact:vech}), (\ref{fact:vec-AXB}) and (\ref{fact:inv}), it is easy to show that, for \(\mathbf{\Lambda}\in \mathbb{S}_{++}^p\),
\begin{align}
    d\,\text{vec}(\mathbf{\Lambda}^{-1})=-(\mathbf{\Lambda}^{-1}\otimes \mathbf{\Lambda}^{-1}) D_p d\boldsymbol{\eta} \label{fact:vec-inv}
\end{align}

Throughout the Appendix, for a real symmetric matrix
\(\mathbf B\in\mathbb S^q\), we write \(\lambda_1(\mathbf B),\ldots,\lambda_q(\mathbf B)\) for its real eigenvalues and define \(\lambda_{\max}(\mathbf B):=\max_{1\le j\le q}\lambda_j(\mathbf B)\). For a matrix \(\mathbf A\), we use \(\|\mathbf A\|_{\mathrm{op}}\) and \(\|\mathbf A\|_F\) to denote its operator norm and Frobenius norm, respectively:
\[
\|\mathbf A\|_{\mathrm{op}}
:=
\sup_{\|\boldsymbol x\|_2=1}\|\mathbf A\boldsymbol x\|_2
=
\sqrt{\lambda_{\max}(\mathbf A^\top\mathbf A)},
\qquad
\|\mathbf A\|_F
:=
\{\tr(\mathbf A^\top\mathbf A)\}^{1/2}
=
\left(\sum_{i,j}(\mathbf A)_{ij}^2\right)^{1/2}.
\]
In particular, if \(\mathbf A\in\mathbb S^p\) is symmetric, then \(\|\mathbf A\|_{\mathrm{op}}=\max_{1\le j\le p}|\lambda_j(\mathbf A)|\) and \(\|\mathbf A\|_{\mathrm{op}}\le \|\mathbf A\|_F\). Moreover, for any symmetric matrix \(\mathbf A\), \(\|\mathbf A\|_F\le \sqrt{2}\|\vech(\mathbf A)\|_2\). If \(\mathbf A\in\mathbb S_{++}^p\), then \(\|\mathbf A\|_F\le \tr(\mathbf A).\)

\subsection{Gradients}\label{sec:grad}

We first consider the log-density of the Wishart kernel, denoted by
\[
l(\boldsymbol{\eta},\nu):= \frac{\nu}{2}\log|\mathbf{\Lambda}| - \frac{\nu p}{2}\log2 -\log \Gamma_{p}\left(\frac{\nu}{2}\right) +\frac{\nu-p-1}{2}\log|\mathbf{W}| - \frac{1}{2}\tr(\mathbf{\Lambda W}).
\]
where we write \(\mathbf{\Lambda}\) for \(\mathbf{\Lambda}(\boldsymbol{\eta})\) for notational convenience.

\subsubsection{\(\nabla_{\nu}\,l\)}

The gradient with respect to \(\nu\) is
\[
\nabla_\nu\,l=\frac{1}{2}\log|\mathbf{\Lambda}| - \frac{p}{2}\log 2 - \frac{1}{2}\psi_p\left(\frac{\nu}{2}\right) + \frac{1}{2}\log|\mathbf{W}|
\]
where \(\psi_p(x):=\frac{d}{dx}\log\Gamma_p(x)\) and we let \(\psi'_p(x)\) denote its derivative.

\subsubsection{\(\nabla_{\boldsymbol{\eta}}\,l\)}

By (\ref{fact:logdet}) and (\ref{fact:trace}), we have
\begin{align*}
    d_{\mathbf{\Lambda}}l =\frac{\nu}{2}d\,\log|\mathbf{\Lambda}| - \frac{1}{2}d\tr(\mathbf{\Lambda W}) = \frac{\nu}{2}\tr(\mathbf{\Lambda}^{-1}d\mathbf{\Lambda})-\frac{1}{2}\tr(\mathbf{W}d\mathbf{\Lambda})=\frac{1}{2}\tr\Big((\nu\mathbf{\Lambda}^{-1}-\mathbf{W})d\mathbf{\Lambda}\Big)
\end{align*}

Note that \(\nu\mathbf{\Lambda}^{-1}-\mathbf{W}\) is symmetric. Then, by (\ref{fact:vec-tr}), 
\[
d_{\mathbf{\Lambda}}l = \frac{1}{2}\text{vec}(\nu\mathbf{\Lambda}^{-1}-\mathbf{W})^\top d\,\text{vec}(\mathbf{\Lambda}) = \boldsymbol{g}^\top\,d\,\text{vec}(\mathbf{\Lambda}),
\]
where we define \(\boldsymbol{g}:=\frac{1}{2}\text{vec}(\nu\mathbf{\Lambda}^{-1}-\mathbf{W})\). Then, using the \(\boldsymbol{\eta}\)-parameterization,
\[
\nabla_{\boldsymbol{\eta}}l = D_p^\top\boldsymbol{g}
\]

\subsection{Hessian}\label{sec:hessian}

\subsubsection{\(\nabla_{\nu\nu}^2l\)}

The Hessian block \(\nabla_{\nu\nu}^2l\) is easily derived as
\[
\nabla_{\nu\nu}^2l = -\frac{1}{4}\psi'_p\left(\frac{\nu}{2}\right)
\]

\subsubsection{\(\nabla_{\boldsymbol{\eta}\nu}^2l\) and \(\nabla_{\nu\boldsymbol{\eta}}^2l\)}

We have
\[
\nabla_{\boldsymbol{\eta}\nu}^2l = \frac{\partial(\nabla_{\boldsymbol{\eta}}l)}{\partial \nu} = \frac{1}{2}D_p^\top\text{vec}(\mathbf{\Lambda}^{-1})
\]
and the transposed block is 
\[
\nabla_{\nu\boldsymbol{\eta}}^2l = \frac{1}{2}\text{vec}(\mathbf{\Lambda}^{-1})^\top D_p
\]

\subsubsection{\(\nabla_{\boldsymbol{\eta}\boldsymbol{\eta}}^2l\)}

By (\ref{fact:inv}), for \(\boldsymbol{g}:=\frac{1}{2}\text{vec}(\nu\mathbf{\Lambda}^{-1}-\mathbf{W})\), we have
\begin{align*}
    &d_{\mathbf{\Lambda}}\boldsymbol{g}=\frac{\nu}{2}\text{vec}\big(d(\mathbf{\Lambda}^{-1})\big)=-\frac{\nu}{2}\left(\mathbf{\Lambda}^{-1}\otimes\mathbf{\Lambda}^{-1}\right)d\,\text{vec}(\mathbf{\Lambda})\\
    \Rightarrow  \quad &d_{\boldsymbol{\eta}}\boldsymbol{g}=-\frac{\nu}{2}\left(\mathbf{\Lambda}^{-1}\otimes\mathbf{\Lambda}^{-1}\right)D_p\,d\boldsymbol{\eta}\\
     \Rightarrow  \quad & \nabla_{\boldsymbol{\eta}\boldsymbol{\eta}}^2l = -\frac{\nu}{2} D_p^\top\left(\mathbf{\Lambda}^{-1}\otimes\mathbf{\Lambda}^{-1}\right)D_p
\end{align*}

\subsubsection{Hessian of \(f\)}

Write \(f:= f(\mathbf{W}\mid \mathbf{\Lambda},\nu)\) for notational convenience. We use the fact that \(\nabla^2f=f\cdot\left(\nabla^2l+\nabla l \nabla l^\top\right)\). Applying this identity blockwise gives the Hessian of \(f\). We organize the Hessian as follows:
\[
\nabla^2f=\begin{bmatrix}
    \nabla_{\boldsymbol{\eta}\boldsymbol{\eta}}^2f & \nabla_{\boldsymbol{\eta}\nu}^2f\\
     \nabla_{\nu\boldsymbol{\eta}}^2f & \nabla_{\nu\nu}^2f
\end{bmatrix}
\]
with
\begin{align*}
    & \nabla_{\boldsymbol{\eta}\boldsymbol{\eta}}^2f = f\cdot D_p^\top \left(\boldsymbol{g}\boldsymbol{g}^\top - \frac{\nu}{2}\left(\mathbf{\Lambda}^{-1}\otimes \mathbf{\Lambda}^{-1}\right)\right)D_p, \\
    & \nabla_{\boldsymbol{\eta}\nu}^2f = f\cdot D_p^\top \left((\nabla_\nu l)\boldsymbol{g} + \frac{1}{2}\text{vec}(\mathbf{\Lambda}^{-1})\right),\quad \nabla_{\nu\boldsymbol{\eta}}^2f = \left(\nabla_{\boldsymbol{\eta}\nu}^2f\right)^\top, \\
    & \nabla_{\nu\nu}^2f = f\cdot\left[(\nabla_\nu l)^2-\frac{1}{4}\psi'_p\left(\frac{\nu}{2}\right)\right]
\end{align*}

\section{Proof of Theoretical Results}

Throughout this section, we always assume that Assumptions~\ref{ass:data_generating}--\ref{ass:base_prior} presented in Section~\ref{sec:theory} hold, and we focus on the MFM--Wishart model~(\ref{eq:mfm_wishart_model}) restricted to \(\Theta^*\), defined in Section~\ref{sec:theory}.

\subsection{Auxiliary lemmas}\label{sec:lemmas}


We first prove a technical lemma that will be used to show that the Wishart kernel satisfies the first-order uniform Lipschitz property when verifying (P.1).

\begin{lemma}\label{eq:amgm}
Let \(p\) be a positive integer. For any scalars \(a,b>0\), \(m\ge 0\), \(r>0\), \(c\ge 0\), and \(\alpha>0\), there exists a constant
\(M<\infty\) depending only on \((a,b,c,m,r,\alpha,p)\) such that
\[
\sup_{\mathbf{W}\in \mathbb{S}_{++}^p}
|\mathbf{W}|^\alpha \exp\!\bigl\{-a\,\tr(\mathbf{W})\bigr\}\,
\tr(\mathbf{W})^m\,
\bigl[c+b|\log|\mathbf{W}||\bigr]^r
\le M .
\]
\end{lemma}

\noindent\textbf{Proof:}
Let \(t:=\tr(\mathbf{W})>0\) and \(x:=|\mathbf{W}|>0\).
By the AM--GM inequality applied to the eigenvalues of \(\mathbf{W}\),
\[
x=|\mathbf{W}| \le \left(\frac{\tr(\mathbf{W})}{p}\right)^p=\left(\frac{t}{p}\right)^p,
\qquad
\log x\le p\log t-p\log p.
\]

\emph{Case 1: \(x=|\mathbf{W}|\ge 1\).}
Then \(1\le x\le (t/p)^p\) implies \(t\ge p\), hence \(\log t\ge 0\).
Moreover, \(\log x\ge 0\) so \(|\log x|=\log x\).
Thus
\[
x^\alpha \le (t/p)^{\alpha p},
\qquad
c+b|\log x| = c+b\log x \le c+b(p\log t-p\log p)\le C_1(1+\log t),
\]
for a constant \(C_1\) depending only on \((b,c,p)\).
Therefore,
\[
x^\alpha\,[c+b|\log x|]^r \le C_2\, t^{\alpha p}(1+\log t)^r,
\]
for a constant \(C_2\) depending only on \((b,c,\alpha,r,p)\), and hence
\[
x^\alpha e^{-a t} t^m [c+b|\log x|]^r
\le
C_2\, t^{m+\alpha p}(1+\log t)^r e^{-a t}.
\]
Since the exponential term dominates any polynomial and logarithmic growth as \(t\to\infty\),
we have
\[
\sup_{t\ge p} t^{m+\alpha p}(1+\log t)^r e^{-a t}<\infty.
\]

\emph{Case 2: \(0<x=|\mathbf{W}|<1\).}
Define \(\phi(x):=x^\alpha[c+b|\log x|]^r\) on \((0,1)\).
For \(x\in(0,1)\), let \(x=e^{-y}\) with \(y>0\). Then
\[
\phi(e^{-y})=e^{-\alpha y}(c+by)^r\to 0
\qquad \text{as } y\to\infty.
\]
Since \(\phi\) is continuous on \((0,1)\), \(\phi(x)\to 0\) as \(x\downarrow 0\), and
\(\phi(x)\to c^r\) as \(x\uparrow 1\), it follows that
\[
\sup_{0<x<1}\phi(x)=:C_0<\infty
\]
for a constant \(C_0\) depending only on \((b,c,\alpha,r)\).
Hence, for \(0<x<1\),
\[
x^\alpha[c+b|\log x|]^r \le C_0,
\]
and therefore
\[
x^\alpha e^{-a t} t^m [c+b|\log x|]^r
\le C_0\, t^m e^{-a t}.
\]
Finally,
\[
\sup_{t>0} t^m e^{-a t}<\infty.
\]

Combining the two cases proves the claim.
\qed\newline


We next prove an auxiliary lemma for Lemma~\ref{thm:P3}, which will later be used in the proof of Theorem~\ref{thm:contraction}. 

\begin{lemma}\label{thm:P3_lemma}
    We use the same notation as in Section~\ref{sec:theory}. For any \(j\in\{1,\dots,k_0\}\), given any radius \(\omega>0\), we define an open ball 
    \[
    B_j:=B(\boldsymbol{\theta}_j^0,\omega)=\left\{\boldsymbol{\theta}:\|\boldsymbol{\theta}-\boldsymbol{\theta}_j^0\|_2<\omega\right\},
    \]
    and we denote the mass of \(B_j\) under the measure \(G\) by \(G(B_j)=\sum_{i:\boldsymbol{\theta}_i\in B_j}\pi_i\). We have the following inequality:
    \[
    W_r(G,G_0)^r\ge \omega^r(\pi_j^0-G(B_j))_+
    \]
    where \(x_+\) is the positive part of \(x\in\mathbb{R}\).
\end{lemma}

\noindent\textbf{Proof:} Fix any \(j\in\{1,\dots,k_0\}\). We will show that, for any coupling matrix \(\mathbf{Q}\in\mathcal{Q}(\boldsymbol{\pi},\boldsymbol{\pi}^0)\),
\[
\sum^k_{i=1}\sum^{k_0}_{l=1}q_{il}\|\boldsymbol{\theta}_i-\boldsymbol{\theta}_l^0\|_2^r\ge \omega^r\left(\pi_j^0-G(B_j)\right)_+
\]

For notational convenience, we define \(S_{\text{in}}:=\{i:\boldsymbol{\theta}_i\in B_j\}\) and \(S_{\text{out}}:=\{i:\boldsymbol{\theta}_i\notin B_j\}\). Then, we have
\begin{align}
    \sum_{i\in S_{\text{in}}}q_{ij} + \sum_{i\in S_{\text{out}}}q_{ij} = \pi_j^0. \label{eq:q_decompose} 
\end{align}

Recall that \(\sum^{k_0}_{l=1}q_{il}=\pi_i\), so for any \(i\), we have \(q_{ij}\le\pi_i\). Thus,
\begin{align}
    \sum_{i\in S_{\text{in}}}q_{ij} \le \sum_{i\in S_{\text{in}}}\pi_i = G(B_j). \label{eq:GB_control}
\end{align}

From (\ref{eq:q_decompose}) and (\ref{eq:GB_control}), we have
\[
\sum_{i\in S_{\text{out}}}q_{ij} = \pi_j^0 - \sum_{i\in S_{\text{in}}}q_{ij} \ge \pi_j^0-G(B_j)\quad \Rightarrow \quad \sum_{i\in S_{\text{out}}}q_{ij} \ge \left(\pi_j^0-G(B_j)\right)_+
\]

Then,
\[
\sum^k_{i=1}\sum^{k_0}_{l=1}q_{il}\|\boldsymbol{\theta}_i-\boldsymbol{\theta}_l^0\|_2^r \ge \sum^k_{i=1}q_{ij}\|\boldsymbol{\theta}_i-\boldsymbol{\theta}_j^0\|_2^r \ge \sum_{i\in S_{\text{out}}}q_{ij}\|\boldsymbol{\theta}_i-\boldsymbol{\theta}_j^0\|_2^r \ge \omega^r\sum_{i\in S_{\text{out}}}q_{ij}\ge \omega^r\left(\pi_j^0-G(B_j)\right)_+
\]

Thus,
\[
W_r(G,G_0)^r = \inf_{\mathbf{Q}\in \mathcal{Q}(\boldsymbol{\pi},\boldsymbol{\pi}^0)}
\sum_{i=1}^{k}\sum_{l=1}^{k_0}q_{il}\|\boldsymbol{\theta}_i-\boldsymbol{\theta}_l^0\|_2^r \ge \omega^r\left(\pi_j^0-G(B_j)\right)_+
\]
\qed

\begin{remark}
    Lemma~\ref{thm:P3_lemma} states an obvious fact: if \(W_r(G,G_0)\) is small, then \(G(B_j)\) will not be too small. More formally, let \(\pi_{\min}^0:=\min_{j}\pi_j^0\), if one wants \(G(B_j)\ge \frac{\pi_j^0}{2}\) for all \(j\), one can choose \(W_r(G,G_0)\le \varepsilon_0\) with \(\varepsilon_0^r\le \omega^r\frac{\pi_{\min}^0}{2}\).
\end{remark}

Next, we prove a lemma used in the proof of Theorem~\ref{thm:contraction}.

\begin{lemma}\label{thm:P3}
    Let \(\mathcal{O}_{k_0}:=\mathcal{O}_{k_0}(\Theta^*)\) denote the space of all mixing measures with at most \(k_0\) support points, all in \(\Theta^*\), and let \(\mu\) be the Lebesgue measure. Then, there exists \(\varepsilon_0>0\) such that \(\int \frac{(p_{G_0}(\mathbf{W}))^2}{p_G(\mathbf{W})}d\mu(\mathbf{W})\le M(\varepsilon_0)\) as long as \(W_1(G,G_0)\le\varepsilon_0\), for any \(G\in\mathcal{O}_{k_0}\), where \(M(\varepsilon_0)\) only depends on \(\varepsilon_0\), \(G_0\) and \(\Theta^*\). \label{eq:P3}
\end{lemma}

\noindent\textbf{Proof:} Recall the notations \(G=\sum^k_{i=1}\pi_i\delta_{\boldsymbol{\theta}_i}\) with \(\pi_i\ge0\) for all \(i\) and \(\sum^k_{i=1}\pi_i=1\), and \(G_0=\sum^{k_0}_{j=1}\pi^0_j\delta_{\boldsymbol{\theta}^0_j}\) with \(\pi^0_j>0\) for all \(j\) and \(\sum^{k_0}_{j=1}\pi^0_j=1\). 

We use the following setup to prove this Lemma. We will show the case where \(k_0\ge2\), but the \(k_0=1\) case can be handled similarly. For \(j=1,\dots, k_0\), we define an open ball \(B_j:=\left\{\boldsymbol{\theta}\in\Theta^*:\|\boldsymbol{\theta}-\boldsymbol{\theta}_j^0\|_2<\omega\right\}\) for a radius \(\omega>0\). Denote \(\pi_{\min}^0:=\min_{1\le j\le k_0}\pi_j^0>0\) and \(\Delta_0:=\min_{i\neq j}\|\boldsymbol{\theta}_i^0-\boldsymbol{\theta}^0_j\|_2>0\). We note that if we take a radius \(\omega < \frac{\Delta_0}{4}\), the balls \(B_1,\dots,B_{k_0}\) are disjoint.

We note that for any symmetric matrix \(\mathbf{A}\), we have \(\|\mathbf{A}\|_{\text{op}}\le \|\mathbf{A}\|_F\le \sqrt{2}\|\vech{(\mathbf{A})}\|_2\). Thus, for any \(i\), if \(\boldsymbol{\theta}\in B_i\), we have \(\|\mathbf{\Lambda}-\mathbf{\Lambda}_i^0\|_{\text{op}}\le \sqrt{2}\|\boldsymbol{\eta}-\boldsymbol{\eta}_i^0\|_2\le \sqrt{2}\|\boldsymbol{\theta}-\boldsymbol{\theta}_i^0\|_2\le \sqrt{2}\omega\). Also, recall that we work in the restricted parameter space \(\Theta^*\). If we want \(\|\mathbf{\Lambda}-\mathbf{\Lambda}_i^0\|_{\text{op}}\le \frac{\underline{\lambda}}{2}\) for all \(i\), we can take \(\omega \le \frac{\underline{\lambda}}{2\sqrt{2}}\). In addition, if \(|\nu-\nu_i^0|\le \underline{\nu}-p\), we have \(\nu\le \nu_i^0+\underline{\nu}-p\le 2\nu_i^0-p\). Thus, if we want \(2\nu^0_i-\nu\ge p\) for all \(i\), we can take \(\omega \le \underline{\nu}-p\).

From now on, we take \(\omega=\min{\left\{\frac{\Delta_0}{4}, \frac{\underline{\lambda}}{2\sqrt{2}}, \underline{\nu}-p\right\}}\) and \(\varepsilon_0=\frac{\omega\pi_{\min}^0}{2}\). By Lemma~\ref{thm:P3_lemma}, if \(W_1(G,G_0)\le \varepsilon_0\), we have 
\[
G(B_j)\ge\pi_j^0-\frac{W_1(G,G_0)}{\omega}\ge \pi_j^0-\frac{\varepsilon_0}{\omega}\ge \pi_j^0-\frac{\pi_{\min}^0}{2}\ge \frac{\pi^0_j}{2}
\]

Note that \(G\in\mathcal{O}_{k_0}\) and \(G(B_j)\ge\frac{\pi_j^0}{2}>0\), so each \(B_j\) contains at least one atom of \(G\). Also, since \(\{B_j\}_{j=1}^{k_0}\) are disjoint, \(G\) has at least \(k_0\) atoms. However, \(G\in\mathcal{O}_{k_0}\), which means that \(G\) can have at most \(k_0\) atoms. So, \(G\) must have exactly \(k_0\) atoms (after discarding zero-mass atoms), each in one distinct \(B_j\), \(j=1,\dots,k_0\).

Thus, we can rewrite \(G\) (after relabeling the atoms of \(G\), we may assume that \(\boldsymbol{\theta}_i\in B_i\)) as 
\[
G=\sum^{k_0}_{i=1}\pi_i\delta_{\boldsymbol{\theta}_i}, \;\text{ with } \boldsymbol{\theta}_i\in B_i \text{ and } \pi_i= G(B_i)\ge \frac{\pi_i^0}{2},\,\forall\,i.
\]

Next, we focus on the Wishart kernel \(f(\mathbf{W}\mid \boldsymbol{\theta})\) and the following mixtures,
\[
p_G(\mathbf{W})=\sum^{k_0}_{i=1}\pi_if(\mathbf{W}\mid \boldsymbol{\theta}_i),\quad p_{G_0}(\mathbf{W})=\sum^{k_0}_{i=1}\pi_i^0f(\mathbf{W}\mid \boldsymbol{\theta}_i^0).
\]

By the Cauchy-Schwarz inequality,
\[
\frac{\left(p_{G_0}(\mathbf{W})\right)^2}{p_G(\mathbf{W})}\le \sum^{k_0}_{i=1}\frac{\left(\pi_i^0\right)^2}{\pi_i}\cdot \frac{\left(f(\mathbf{W}\mid\boldsymbol{\theta}_i^0)\right)^2}{f(\mathbf{W}\mid \boldsymbol{\theta}_i)}.
\]

Let \(\mu\) be the Lebesgue measure on the \(d=\frac{p(p+1)}{2}\)-dimensional linear space of symmetric matrices, and we have
\begin{align*}
    \int \frac{\left(p_{G_0}(\mathbf{W})\right)^2}{p_G(\mathbf{W})}d\mu(\mathbf{W}) &\le  \sum^{k_0}_{i=1}\frac{\left(\pi_i^0\right)^2}{\pi_i} \int \frac{\left(f(\mathbf{W}\mid\boldsymbol{\theta}_i^0)\right)^2}{f(\mathbf{W}\mid \boldsymbol{\theta}_i)} d\mu(\mathbf{W})\\
    &\le 2 \sum^{k_0}_{i=1}\pi_i^0 \int \frac{\left(f(\mathbf{W}\mid\boldsymbol{\theta}_i^0)\right)^2}{f(\mathbf{W}\mid \boldsymbol{\theta}_i)} d\mu(\mathbf{W}),
\end{align*}
where the last inequality holds since \(\pi_i\ge \frac{\pi_i^0}{2}\), \(\forall\,i\). 

Hence, what we need to show is that, when \(\boldsymbol{\theta}_i\) is in a small-enough neighborhood of \(\boldsymbol{\theta}_i^0\), for all \(i\),
\begin{align}
    I(\boldsymbol{\theta}_i^0,\boldsymbol{\theta}_i):=\int \frac{\left(f(\mathbf{W}\mid\boldsymbol{\theta}_i^0)\right)^2}{f(\mathbf{W}\mid \boldsymbol{\theta}_i)} d\mu(\mathbf{W})\le C,\quad \forall\, i \text{ and for some constant \(C\)}. \label{eq:P3_target}
\end{align}

From now on we drop the index \(i\) since the following arguments hold for all \(i\). Consider the Wishart kernel
\[
f(\mathbf{W}\mid \mathbf{\Lambda}, \nu)=\kappa(\mathbf{\Lambda},\nu)|\mathbf{W}|^{\frac{\nu-p-1}{2}}\exp{\left[-\frac{1}{2}\tr(\mathbf{\Lambda W})\right]},\quad \text{where }\kappa(\mathbf{\Lambda},\nu):=\frac{|\mathbf{\Lambda}|^{\nu/2}}{2^{\nu p/2}\Gamma_p(\nu/2)},
\]
and we compute
\begin{align}
  \frac{(f(\mathbf{W}\mid \mathbf{\Lambda}_0,\nu_0))^2}{f(\mathbf{W}\mid \mathbf{\Lambda},\nu)} = \frac{\kappa(\mathbf{\Lambda}_0,\nu_0)^2}{\kappa(\mathbf{\Lambda},\nu)}|\mathbf{W}|^{\frac{m-p-1}{2}}\exp{\left[-\frac{1}{2}\tr(\mathbf{BW})\right]},  \label{eq:turn_wishart}
\end{align}
with \(\mathbf{B}:=2\mathbf{\Lambda}_0-\mathbf{\Lambda}\) and \(m:=2\nu_0-\nu\). Recall that we took \(\omega=\min{\left\{\frac{\Delta_0}{4}, \frac{\underline{\lambda}}{2\sqrt{2}}, \underline{\nu}-p\right\}}\) and \(\varepsilon_0=\frac{\omega\pi_{\min}^0}{2}\). This implies that \(\mathbf{B}\succ\mathbf{0}\) and \(m>p-1\). Hence, the \(|\mathbf{W}|^{\frac{m-p-1}{2}}\exp{\left[-\frac{1}{2}\tr(\mathbf{BW})\right]}\) in Equation~(\ref{eq:turn_wishart}) is a Wishart kernel with the precision matrix \(\mathbf{B}\) and degrees of freedom \(m\).

Thus, we have
\[
\int |\mathbf{W}|^{\frac{m-p-1}{2}}\exp{\left[-\frac{1}{2}\tr(\mathbf{BW})\right]}d\mu(\mathbf{W})=2^{\frac{mp}{2}}\Gamma_p\left(\frac{m}{2}\right)|\mathbf{B}|^{-\frac{m}{2}}
\]

Hence, it suffices to evaluate \(I(\boldsymbol{\theta}_i^0,\boldsymbol{\theta})\):
\begin{align}
    I(\boldsymbol{\theta}_i^0,\boldsymbol{\theta})&=\frac{\kappa(\mathbf{\Lambda}_0,\nu_0)^2}{\kappa(\mathbf{\Lambda},\nu)}\cdot2^{\frac{mp}{2}}\Gamma_p\left(\frac{m}{2}\right)|\mathbf{B}|^{-\frac{m}{2}}\\
    &=\frac{|\mathbf{\Lambda}_0|^{\nu_0}}{2^{\nu_0p}\Gamma_p\left(\frac{\nu_0}{2}\right)^2}\cdot \frac{2^{\frac{\nu p}{2}}\Gamma_p\left(\frac{\nu}{2}\right)}{|\mathbf{\Lambda}|^{\frac{\nu}{2}}}\cdot 2^{\frac{mp}{2}}\Gamma_p\left(\frac{m}{2}\right)|\mathbf{B}|^{-\frac{m}{2}}\\
    &=\frac{\Gamma_p\left(\frac{\nu}{2}\right)\Gamma_p\left(\frac{2\nu_0-\nu}{2}\right)}{\Gamma_p\left(\frac{\nu_0}{2}\right)^2}\cdot \frac{|\mathbf{\Lambda}_0|^{\nu_0}}{|\mathbf{\Lambda}|^{\frac{\nu}{2}}|2\mathbf{\Lambda}_0-\mathbf{\Lambda}|^{\frac{2\nu_0-\nu}{2}}}
\end{align}

For each \(i\), define the closed neighborhood \(K_i:=\left\{\boldsymbol{\theta}\in\Theta^*:\|\boldsymbol{\theta}-\boldsymbol{\theta}_i^0\|\le\omega\right\}\). Note that \(\boldsymbol{\theta}\mapsto I(\boldsymbol{\theta}_i^0,\boldsymbol{\theta})\) is continuous on \(K_i\). Since \(K_i\) is compact, we have \(C_i:=\sup_{\boldsymbol{\theta}\in K_i}I(\boldsymbol{\theta}_i^0,\boldsymbol{\theta})<\infty\). Thus, \(C:=\max_{1\le i \le k_0}C_i<\infty\). Hence,
\[
\int \frac{\left(p_{G_0}(\mathbf{W})\right)^2}{p_G(\mathbf{W})}d\mu(\mathbf{W})\le 2 \sum^{k_0}_{i=1}\pi_i^0 I(\boldsymbol{\theta}_i^0,\boldsymbol{\theta}_i)\le 2 \sum^{k_0}_{i=1}\pi_i^0 C = 2C
\]

We can take \(M(\varepsilon_0)=2C\), and this completes the proof. \qed
\subsection{Proof of Lemma 1}

The left-hand side is continuous on the open cone \(\mathbb S_{++}^p\). If it were nonzero at some point, it would be nonzero on a neighborhood of positive Lebesgue measure, contradicting the almost-everywhere identity. Hence the identity holds for all \(\mathbf W\in\mathbb S_{++}^p\).

Note that, for \(\boldsymbol{\theta}_i=(\boldsymbol{\eta}_i^\top,\nu_i)^\top\), we can decompose the coefficients \(\boldsymbol{\beta}_i\) in the identity by
\begin{align*}
    \boldsymbol{\beta}_i=\begin{bmatrix}
        \boldsymbol{u}_i \\
        v_i
    \end{bmatrix},\quad \boldsymbol{u}_i\in\mathbb{R}^d,\,v_i\in\mathbb{R},
\end{align*}
and we define \(\mathbf{H}_i:=\vech^{-1}(\boldsymbol{u}_i)\in\mathbb{S}^p\).

For notational convenience, we define
\[
\kappa_i:=\frac{|\mathbf{\Lambda}_i|^{\nu_i/2}}{2^{\nu_ip/2}\Gamma_p(\nu_i/2)}, \quad s_i:=\frac{\nu_i-p-1}{2}.
\]

Using the derivative formulas for the Wishart density from Section~\ref{sec:grad}, we have
\[
\boldsymbol{u}_i^\top\nabla_{\boldsymbol{\eta}}f(\mathbf{W}\mid \boldsymbol{\theta}_i) = \frac{1}{2}f(\mathbf{W}\mid\boldsymbol{\theta}_i)\tr(\mathbf{H}_i(\nu_i\mathbf{\Lambda}_i^{-1}-\mathbf{W})),
\]
and,
\[
v_i\nabla_{\nu}f(\mathbf{W}\mid \boldsymbol{\theta}_i)=v_if(\mathbf{W}\mid \boldsymbol{\theta}_i)\left[\frac{1}{2}\log|\mathbf{\Lambda}_i|-\frac{p}{2}\log2 -\frac{1}{2}\psi_p\left(\frac{\nu_i}{2}\right)+\frac{1}{2}\log|\mathbf{W}|\right]
\]

Hence, there exist coefficients \(a_i,b_i\in\mathbb{R}\) and \(\mathbf{C}_i\in\mathbb{S}^p\) such that
\begin{align}
\sum^k_{i=1}|\mathbf{W}|^{s_i}\exp{\left\{-\frac{1}{2}\tr(\mathbf{\Lambda}_i\mathbf{W})\right\}}\Big(a_i+b_i\log{|\mathbf{W}|}-\tr(\mathbf{C}_i\mathbf{W})\Big)=0,\quad \forall\,\mathbf{W}\in\mathbb{S}_{++}^p \label{eq:identity}
\end{align}
where one may take
\begin{align*}
    &a_i = \kappa_i\left[\alpha_i + \frac{\nu_i}{2}\tr(\mathbf{H}_i\mathbf{\Lambda}_i^{-1})+v_i\left(\frac{1}{2}\log|\mathbf{\Lambda}_i|-\frac{p}{2}\log2 -\frac{1}{2}\psi_p\left(\frac{\nu_i}{2}\right)\right)\right],\\
    &\qquad \qquad \qquad \qquad \qquad b_i=\frac{\kappa_iv_i}{2}, \quad \mathbf{C}_i=\frac{\kappa_i}{2}\mathbf{H}_i.
\end{align*}

We now group the indices by the distinct precision matrices. Let \(\mathbf{\Omega}_1,\dots,\mathbf{\Omega}_M\) be the distinct matrices among \(\mathbf{\Lambda}_1,\dots,\mathbf{\Lambda}_k\), and define
\[
J_g:=\{i:\mathbf{\Lambda}_i=\mathbf{\Omega}_g\},\quad g=1,\dots,M
\]

For each \(g\), define
\[
G_g(\mathbf{W}):=\sum_{i\in J_g}|\mathbf{W}|^{s_i}\Big(a_i+b_i\log{|\mathbf{W}|}-\tr(\mathbf{C}_i\mathbf{W})\Big).
\]
Then, (\ref{eq:identity}) becomes:
\begin{align}
    \sum^M_{g=1}\exp{\left\{-\frac{1}{2}\tr(\mathbf{\Omega}_g\mathbf{W})\right\}}G_g(\mathbf{W})=0,\quad \forall\, \mathbf{W}\in\mathbb{S}_{++}^p \label{eq:identity_G}
\end{align}

Note that, within each fixed group \(J_g\), the exponents \(\{s_i:i\in J_g\}\) are pairwise distinct, since \(\boldsymbol{\theta}_1,\dots, \boldsymbol{\theta}_k\) are pairwise distinct. 

We first assume \(M\ge2\) and eliminate one group at a time. Choose \(\mathbf{U}_0\in\mathbb{S}_{++}^p\) such that
\[
\tr(\mathbf{\Omega}_g\mathbf{U}_0) \neq \tr(\mathbf{\Omega}_{g'}\mathbf{U}_0),\quad \text{for all }g\neq g'.
\]
We note that such a choice is always possible when \(p\ge2\), since, for any \(g\neq g'\), the set \(\{\mathbf{U}\in\mathbb{S}_{++}^p:\tr\left[(\mathbf{\Omega}_g-\mathbf{\Omega}_{g'})\mathbf{U}]=0\right\}\) is a proper hyperplane, and finitely many such hyperplanes cannot cover the whole \(\mathbb{S}_{++}^p\). Then, let \(g_*\) be the unique index minimizing \(\tr(\mathbf{\Omega}_g\mathbf{U}_0)\). By continuity, there exists a nonempty open neighborhood \(\mathcal{O}\subset \mathbb{S}_{++}^p\) of \(\mathbf{U}_0\) such that
\[
\tr(\mathbf{\Omega}_{g_*}\mathbf{U})<\tr(\mathbf{\Omega}_g\mathbf{U}),\quad \forall\,\mathbf{U}\in\mathcal{O},\;\forall\,g\neq g_*
\]

Fix \(\mathbf{U}\in\mathcal{O}\) and set \(\mathbf{W}=t\mathbf{U}\) with \(t>0\). From (\ref{eq:identity_G}), 
\begin{align}
    \exp{\left\{-\frac{t}{2}\tr(\mathbf{\Omega}_{g_*}\mathbf{U})\right\}}G_{g_*}(t\mathbf{U}) + \sum_{g\neq g_*}\exp{\left\{-\frac{t}{2}\tr(\mathbf{\Omega}_{g}\mathbf{U})\right\}}G_{g}(t\mathbf{U}) = 0 
\end{align}
and therefore,
\begin{align}
    G_{g_*}(t\mathbf{U}) = - \sum_{g\neq g_*}\exp\left\{-\frac{t}{2}\left[\tr(\mathbf{\Omega}_g\mathbf{U})-\tr(\mathbf{\Omega}_{g_*}\mathbf{U})\right]\right\}G_{g}(t\mathbf{U}) \label{eq:G_g*}
\end{align}

For each fixed \(\mathbf{U}\in \mathcal{O}\), every \(G_g(t\mathbf{U})\) is a finite sum of terms of the form \(t^{ps_i}\), \(t^{ps_i+1}\) and \(t^{ps_i}\log{t}\), so it grows at most polynomially in \(t\) times \((1+\log{t})\). Since all exponential gaps in (\ref{eq:G_g*}) are strictly positive, there exists \(c(\mathbf{U})>0\) and \(N<\infty\) such that
\begin{align}
G_{g_*}(t\mathbf{U})=O\left(e^{-c(\mathbf{U})t}t^N(1+\log{t})\right),\quad t\rightarrow \infty \label{eq:G_asymptotic}
\end{align}

We now show that all coefficients indexed by \(J_{ g_*}\) vanish. First, let \(\mathbf{U},\mathbf{V}\in\mathcal{O}\) satisfy \(|\mathbf{U}|=|\mathbf{V}|\). Subtracting (\ref{eq:G_asymptotic}) for \(\mathbf{U}\) and \(\mathbf{V}\), we obtain 
\begin{align}
    \sum_{i\in J_{g_*}} |\mathbf{U}|^{s_i}\Big(\tr(\mathbf{C}_i\mathbf{U})-\tr(\mathbf{C}_i\mathbf{V})\Big)t^{ps_i+1} = O\left(e^{-ct}t^N(1+\log{t})\right),\quad t\rightarrow \infty, \label{eq:U-V}
\end{align}
for some constant \(c>0\). Order the exponents \(\{s_i:i\in J_{g_*}\}\) and let \(s_{\max}\) denote the largest one. Dividing (\ref{eq:U-V}) by \(t^{ps_{\max}+1}\) and letting \(t\rightarrow \infty\), the smaller-order terms on the left-hand side vanish, while the right-hand side goes to \(0\). Hence, the coefficients of the leading power must be \(0\). Repeating this argument successively over the remaining exponents shows that
\begin{align}
    \tr(\mathbf{C}_i\mathbf{U})=\tr(\mathbf{C}_i\mathbf{V}),\quad \forall\,i\in J_{g_*},\;\forall\, \mathbf{U},\mathbf{V}\in\mathcal{O} \text{ with }|\mathbf{U}|=|\mathbf{V}|, \label{eq:tr_equal}
\end{align}

Now, fix an index \(i\in J_{g_*}\) and a \(\mathbf{U}_{\star}\in\mathcal{O}\). Define \(\mathbf{D}_i:=\mathbf{U}_{\star}^{\frac{1}{2}}\mathbf{C}_i\mathbf{U}_{\star}^{\frac{1}{2}}\), and we note that \(\mathbf{D}_i\in\mathbb{S}^p\). Since \(\mathcal{O}\) is open, we can choose a sufficiently small \(\varepsilon >0\) such that, for any \(r\neq s\), we have 
\[
\mathbf{U}_{+}^{(rs)}:= \mathbf{U}_{\star}^{\frac{1}{2}}\big(\mathbf{I}_p + \varepsilon(\mathbf{E}_{rs}+\mathbf{E}_{sr})\big)\mathbf{U}_{\star}^{\frac{1}{2}}\in\mathcal{O};\quad \mathbf{U}_{-}^{(rs)}:= \mathbf{U}_{\star}^{\frac{1}{2}}\big(\mathbf{I}_p - \varepsilon(\mathbf{E}_{rs}+\mathbf{E}_{sr})\big)\mathbf{U}_{\star}^{\frac{1}{2}}\in\mathcal{O}
\]
where \(\mathbf{E}_{rs}\) is the \((r,s)\)-th standard basis matrix, that is, a suitable square matrix with the \((r,s)\)-th entry to be \(1\), and \(0\) otherwise. 

We note that \(|\mathbf{U}_{+}^{(rs)}|=|\mathbf{U}_{-}^{(rs)}|=|\mathbf{U}_{\star}|(1-\varepsilon^2)\). Thus, by (\ref{eq:tr_equal}), we have \(\tr(\mathbf{C}_i\mathbf{U}_{+}^{(rs)})=\tr(\mathbf{C}_i\mathbf{U}_{-}^{(rs)})\). So, 
\[
0=\tr\big(\mathbf{C}_i(\mathbf{U}_{+}^{(rs)}-\mathbf{U}_{-}^{(rs)})\big)=2\varepsilon\tr\big(\mathbf{D}_i(\mathbf{E}_{rs}+\mathbf{E}_{sr})\big)
\]

Since \(\mathbf{D}_i\) is symmetric, we have 
\[
\tr\big(\mathbf{D}_i(\mathbf{E}_{rs}+\mathbf{E}_{sr})\big)=2(\mathbf{D}_i)_{rs}=2(\mathbf{D}_i)_{sr}
\]

So, \((\mathbf{D}_i)_{rs}=0\) for any \(r\neq s\), and thus \(\mathbf{D}_i\) must be a diagonal matrix. 

Next, for any \(r\neq s\), we can pick an \(x\neq 1\) but sufficiently close to \(1\) so that
\(
\mathbf{V}_{x}^{(rs)}:=\mathbf{U}_{\star}^{\frac{1}{2}}\mathbf{M}_x\mathbf{U}_{\star}^{\frac{1}{2}}\in\mathcal{O}
\)
where \(\mathbf{M}_x=\text{diag}(1,\dots,x,\dots,x^{-1},\dots,1)\) is the diagonal matrix with the \(r\)-th diagonal element to be \(x\), \(s\)-th diagonal element to be \(x^{-1}\), and remaining diagonal elements to be \(1\).

We note that \(|\mathbf{V}_x^{(rs)}|=|\mathbf{U}_{\star}|\). So, by (\ref{eq:tr_equal}), we have \(\tr(\mathbf{C}_i\mathbf{V}_x^{(rs)})=\tr(\mathbf{C}_i\mathbf{U}_{\star})\). Then, we have \(\tr(\mathbf{D}_i\mathbf{M}_x)=\tr(\mathbf{D}_i)\). If we write \(\mathbf{D}_i=\text{diag}(d_1,\dots,d_p)\), then we can obtain \(d_rx-d_s=0\). This holds for all \(x\) in a sufficiently small neighborhood of \(1\), so it can only be the case that \(d_r=d_s=0\). Since \(r\neq s\) are chosen arbitrarily, we conclude that all diagonal elements of \(\mathbf{D}_i\) are \(0\). So, \(\mathbf{D}_i=\mathbf{0}\), and thus \(\mathbf{C}_i=\mathbf{0}\) for all \(i\).

With \(\mathbf{C}_i=\mathbf{0}\) for all \(i\in J_{g_*}\), equation (\ref{eq:G_asymptotic}) simplifies to
\begin{align}
    \sum_{i\in J_{g_*}}|\mathbf{U}|^{s_i}t^{ps_i}\big(a_i+b_i(\log{|\mathbf{U}|}+p\log{t})\big)=O\left(e^{-c(\mathbf{U})t}t^N(1+\log{t})\right),\quad t\rightarrow \infty \label{eq:G_simple}
\end{align}

Again order the exponents \(\{s_i:i\in J_{g_*}\}\) increasingly. Dividing (\ref{eq:G_simple}) first by \(t^{ps_{\max}}\log{t}\) and letting \(t\rightarrow \infty\), we obtain \(b_i=0\) for the index corresponding to the largest exponent \(s_{\max}\). Then, dividing by \(t^{ps_{\max}}\) and letting \(t\rightarrow\infty\) gives \(a_i=0\) for the same index. Repeating this procedure over the remaining exponents gives \(a_i=b_i=0\) for all \(i \in J_{g_*}\).

Thus, every coefficient in the group \(J_{g_*}\) vanishes. Then, we remove this group from equation (\ref{eq:identity_G}) and repeat the same argument for the remaining groups. If \(M=1\), then \eqref{eq:identity_G} gives \(G_1(\mathbf W)\equiv 0\) on \(\mathbb S_{++}^p\). In this case, we may take \(g_*=1\) and choose any nonempty open neighborhood \(\mathcal O\subset\mathbb S_{++}^p\). The argument
above also applies when the right-hand side of (\ref{eq:G_asymptotic}) is replaced by an exact zero. Thus, we conclude that
\(
a_i=b_i=0,\;\mathbf{C}_i=\mathbf{0},\;i=1,\dots,k.
\)

Finally, since \(b_i=\frac{\kappa_iv_i}{2}\) and \(\mathbf{C}_i=\frac{\kappa_i}{2}\mathbf{H}_i\) (recall that \(\kappa_i>0\)), we obtain that \(v_i=0\) and \(\mathbf{H}_i=\mathbf{0}\) for \(i=1,\dots,k\). Hence, \(\boldsymbol{u}_i=\boldsymbol{0}\) and therefore \(\boldsymbol{\beta}_i=\boldsymbol{0}\) for all \(i\). Returning to the definition of \(a_i\), we now have \(a_i=\kappa_i\alpha_i\), so \(a_i=0\) implies \(\alpha_i=0\). This completes the proof.
\qed

\begin{remark}[Why the condition \(p\ge2\) is imposed]\label{rem:p1_gamma_failure}
The restriction \(p\ge2\) is not merely a technical artifact of the proof. When \(p=1\), the Wishart family reduces to the Gamma family
\[
f(w\mid \lambda,\nu)
=
\frac{\lambda^{\nu/2}}{2^{\nu/2}\Gamma(\nu/2)}
w^{\nu/2-1}\exp\{-\lambda w/2\},
\qquad w>0,
\]
where \(\lambda>0\) is the scalar precision parameter. Suppose that both \((\lambda,\nu)\) and \((\lambda,\nu+2)\) belong to the parameter space. Then
\[
w f(w\mid \lambda,\nu)=\frac{\nu}{\lambda} f(w\mid \lambda,\nu+2).
\]
Hence
\[
\frac{\partial}{\partial\lambda}f(w\mid \lambda,\nu)=\left(\frac{\nu}{2\lambda}-\frac{w}{2}\right)f(w\mid \lambda,\nu)=\frac{\nu}{2\lambda}f(w\mid \lambda,\nu)-\frac{\nu}{2\lambda}f(w\mid \lambda,\nu+2).
\]
Equivalently,
\[
-\frac{\nu}{2\lambda}f(w\mid \lambda,\nu)+\frac{\nu}{2\lambda}f(w\mid \lambda,\nu+2)+\frac{\partial}{\partial\lambda}f(w\mid \lambda,\nu)=0, \qquad \forall\,w>0.
\]
This is a nontrivial first-order linear relation involving two distinct parameter points and the derivative at one of them. Therefore, the univariate Wishart, or equivalently Gamma, family is not first-order identifiable on any parameter space containing such a pair. This explains why Lemma~\ref{lem:first_order_identifiability} is stated for \(p\ge2\).
\end{remark}

\subsection{Proof of Lemma 2}

We first bound all the terms in the Hessian that depend only on the parameters. By the spectral bounds defining the compact parameter space \(\Theta^*\), we have
\[
\|\mathbf{\Lambda}^{-1}\|_{\text{op}}\le \frac{1}{\underline{\lambda}},
\qquad
\|\mathbf{\Lambda}^{-1}\|_{F}\le \frac{\sqrt{p}}{\underline{\lambda}},
\qquad
|\log|\mathbf{\Lambda}||\le p\max\{|\log \underline{\lambda}|,\ |\log \overline{\lambda}|\}.
\]
Moreover,
\[
|\mathbf{\Lambda}|^{\nu/2}
=
\exp\!\left(\frac{\nu}{2}\log|\mathbf{\Lambda}|\right)
\le
\exp\!\left(
\frac{\overline{\nu} p}{2}\max\{|\log \underline{\lambda}|,\ |\log \overline{\lambda}|\}
\right).
\]
In addition, there exist positive constants \(C_\psi\) and \(C_{\psi'}\) such that
\(
\sup_{\nu\in[\underline{\nu},\overline{\nu}]}\left|\psi_p\left(\frac{\nu}{2}\right)\right|\le C_\psi
\)
and 
\(
\sup_{\nu\in[\underline{\nu},\overline{\nu}]}\left|\psi_p'\left(\frac{\nu}{2}\right)\right|\le C_{\psi'},
\)
since \(\psi_p\left(\frac{\nu}{2}\right)\) and \(\psi_p'\left(\frac{\nu}{2}\right)\) are both continuous on the interval \([\underline{\nu},\overline{\nu}]\).
The function
\(
c(\nu):=\frac{1}{2^{\frac{\nu p}{2}}\Gamma_p\left(\frac{\nu}{2}\right)}
\)
is positive and continuous on \([\underline{\nu},\overline{\nu}]\), so \(0<c_{\min}\le c(\nu)\le c_{\max}<\infty\).
Finally, there exists a constant \(C_0\) such that for all \(\mathbf{W}\in\mathbb{S}_{++}^p\) and all \(\boldsymbol{\theta}\in\Theta^*\),
\begin{equation}\label{eq:f_bound}
f(\mathbf{W}\mid \boldsymbol{\theta})
\le
C_0\,|\mathbf{W}|^{\frac{\nu-p-1}{2}}
\exp\!\left[-\frac{\underline{\lambda}}{2}\tr(\mathbf{W})\right].
\end{equation}

Using the Hessian formulas derived in Section~\ref{sec:hessian},
\begin{equation}
\nabla^2_{\boldsymbol{\theta}} f=
\begin{bmatrix}
\nabla^2_{\boldsymbol{\eta}\boldsymbol{\eta}}f
&
\nabla^2_{\boldsymbol{\eta}\nu}f
\\[6pt]
(\nabla^2_{\boldsymbol{\eta}\nu}f)^\top
&
\nabla^2_{\nu\nu}f
\end{bmatrix}.
\label{eq:Hess-f}
\end{equation}

\paragraph{(1) \(\nabla^2_{\nu\nu}f\):}
We first have
\[
\nabla^2_{\nu\nu}f
=
f\cdot\left[l_\nu^2-\frac{1}{4}\psi_p'\left(\frac{\nu}{2}\right)\right]
\qquad \Rightarrow \qquad
|\nabla^2_{\nu\nu}f|
\le
f\cdot \left(l_\nu^2 + \frac{1}{4}C_{\psi'}\right),
\]
where
\[
l_\nu:=\frac{\partial l}{\partial \nu}
=
\frac12\log|\mathbf{\Lambda}|
-\frac{p}{2}\log 2
-\frac12\psi_p(\nu/2)
+\frac12\log|\mathbf{W}|.
\]
Hence there is a constant \(C_\nu<\infty\), depending only on \(\Theta^*\), such that
\[
|l_\nu|
\le
C_\nu+\frac12|\log|\mathbf{W}||
\qquad\Rightarrow\qquad
l_\nu^2
\le
2C_\nu^2+\frac12|\log|\mathbf{W}||^2.
\]
Thus
\[
|\nabla^2_{\nu\nu}f|
\le
f\cdot\left[C_1+C_2|\log|\mathbf{W}||^2\right]
\le
C_0|\mathbf{W}|^{\frac{\nu-p-1}{2}}
\exp\!\left[-\frac{\underline{\lambda}}{2}\tr(\mathbf{W})\right]
\left[C_1+C_2|\log|\mathbf{W}||^2\right]
\]
for some positive constants \(C_1\) and \(C_2\) depending only on \(\Theta^*\).

Let 
\(
\underline{\alpha}:=\frac{\underline{\nu}-p-1}{2}\), 
\(
\overline{\alpha}:=\frac{\overline{\nu}-p-1}{2}.
\)
We note that \(\overline{\alpha}>\underline{\alpha}>0\). Then, for every \(\nu\in[\underline{\nu},\overline{\nu}]\), 
\(
|\mathbf{W}|^{\frac{\nu-p-1}{2}}
\le
|\mathbf{W}|^{\underline{\alpha}}+|\mathbf{W}|^{\overline{\alpha}}.
\)
 Moreover,
\(
C_1+C_2|\log|\mathbf{W}||^2
\le
\bigl(\sqrt{C_1}+\sqrt{C_2}\,|\log|\mathbf{W}||\bigr)^2.
\)
 Therefore,
\[
|\nabla^2_{\nu\nu}f|
\le
C_0\sum_{\alpha\in\{\underline{\alpha},\overline{\alpha}\}}
|\mathbf{W}|^{\alpha}
\exp\!\left[-\frac{\underline{\lambda}}{2}\tr(\mathbf{W})\right]
\bigl(\sqrt{C_1}+\sqrt{C_2}\,|\log|\mathbf{W}||\bigr)^2.
\]
Using Lemma~\ref{eq:amgm} with \(a=\underline{\lambda}/2\), \(m=0\), \(r=2\),
\(c=\sqrt{C_1}\), and \(b=\sqrt{C_2}\), we have
\begin{equation}\label{eq:nunu_bound}
    \sup_{\mathbf{W}\in\mathbb{S}_{++}^p,\ \boldsymbol{\theta}\in\Theta^*}|\nabla^2_{\nu\nu}f(\mathbf{W}\mid\boldsymbol{\theta})|<\infty.
\end{equation}

\paragraph{(2) \(\nabla^2_{\boldsymbol{\eta}\nu}f\):}
We have 
\[
\nabla^2_{\boldsymbol{\eta}\nu}f
=
D_p^\top \nabla^2_{\operatorname{vec}(\mathbf{\Lambda}),\nu}f
=
D_p^\top f\left(l_\nu\,\boldsymbol{g}+\frac12\,\operatorname{vec}(\mathbf{\Lambda}^{-1})\right),
\]
and hence
\[
\|\nabla^2_{\boldsymbol{\eta}\nu}f\|_2
\le
\|D_p\|_{\text{op}}\,f\cdot\left(|l_\nu|\|\boldsymbol{g}\|_2+\frac{1}{2}\|\mathbf{\Lambda}^{-1}\|_F\right).
\]

The terms \(\|D_p\|_{\text{op}}\) and \(\|\mathbf{\Lambda}^{-1}\|_F\) are uniformly bounded over \(\Theta^*\). Thus, it remains to control \(f\,|l_\nu|\|\boldsymbol{g}\|_2\). Note that
\[
\|\boldsymbol{g}\|_2
=
\frac{1}{2}\|\nu\mathbf{\Lambda}^{-1}-\mathbf{W}\|_F
\le
\frac{1}{2}\nu\|\mathbf{\Lambda}^{-1}\|_F+\frac{1}{2}\|\mathbf{W}\|_F
\le
C_{\boldsymbol{g}}+\frac{1}{2}\|\mathbf{W}\|_F
\le
C_{\boldsymbol{g}}+\frac{1}{2}\tr(\mathbf{W}),
\]
for some constant \(C_{\boldsymbol{g}}<\infty\) depending only on \(\Theta^*\), where we used that
\(\|\mathbf{W}\|_F\le \tr(\mathbf{W})\) for \(\mathbf{W}\in\mathbb{S}_{++}^p\).

As shown above, there exists a constant \(C_\nu<\infty\), depending only on \(\Theta^*\), such that 
\(
|l_\nu|
\le
C_\nu+\frac{1}{2}|\log|\mathbf{W}||.
\)
 Therefore,
\[
|l_\nu|\|\boldsymbol{g}\|_2
\le
\left(C_\nu+\frac{1}{2}|\log|\mathbf{W}||\right)
\left(C_{\boldsymbol{g}}+\frac{1}{2}\tr(\mathbf{W})\right).
\]
Using the elementary inequality \(xy\le \frac{1}{2}(x^2+y^2)\), it follows that there exists a constant \(C<\infty\) such that
\[
|l_\nu|\|\boldsymbol{g}\|_2
\le
C\left(1+|\log|\mathbf{W}||^2+\tr(\mathbf{W})^2\right).
\]
Hence,
\[
\|\nabla^2_{\boldsymbol{\eta}\nu}f\|_2
\le
C\,f(\mathbf{W}\mid\boldsymbol{\theta})
\left(1+|\log|\mathbf{W}||^2+\tr(\mathbf{W})^2\right).
\]

Using the bound (\ref{eq:f_bound}) from above, together with
\(
|\mathbf{W}|^{\frac{\nu-p-1}{2}}
\le
|\mathbf{W}|^{\underline{\alpha}}+|\mathbf{W}|^{\overline{\alpha}},
\)
we obtain
\[
\|\nabla^2_{\boldsymbol{\eta}\nu}f\|_2
\le
C\sum_{\alpha\in\{\underline{\alpha},\overline{\alpha}\}}
|\mathbf{W}|^\alpha
\exp\!\left[-\frac{\underline{\lambda}}{2}\tr(\mathbf{W})\right]
\left(1+|\log|\mathbf{W}||^2+\tr(\mathbf{W})^2\right).
\]
Finally, since
\(
1+|\log|\mathbf{W}||^2+\tr(\mathbf{W})^2
\le
(1+|\log|\mathbf{W}||)^2
+
\tr(\mathbf{W})^2(1+|\log|\mathbf{W}||)^2,
\)
Lemma~\ref{eq:amgm} applied separately with \(m=0\) and \(m=2\), and with \(a=\underline{\lambda}/2\), \(r=2\), \(c=b=1\), yields
\begin{equation}\label{eq:etanu_bound}
\sup_{\mathbf{W}\in\mathbb{S}_{++}^p,\ \boldsymbol{\theta}\in\Theta^*}\|\nabla^2_{\boldsymbol{\eta}\nu}f(\mathbf{W}\mid\boldsymbol{\theta})\|_2<\infty.   
\end{equation}

\paragraph{(3) \(\nabla^2_{\boldsymbol{\eta}\boldsymbol{\eta}}f\):}

First, we have
\[
\nabla^2_{\boldsymbol{\eta}\boldsymbol{\eta}}f
=
D_p^\top f\left(\boldsymbol{g}\,\boldsymbol{g}^\top-\frac{\nu}{2}\,(\mathbf{\Lambda}^{-1}\otimes\mathbf{\Lambda}^{-1})\right)D_p,
\]
and hence
\[
\|\nabla^2_{\boldsymbol{\eta}\boldsymbol{\eta}}f\|_{\text{op}}
\le
\|D_p\|_{\text{op}}^2
f\left(
\|\boldsymbol{g}\boldsymbol{g}^\top\|_{\text{op}}
+
\frac{\nu}{2}\|\mathbf{\Lambda}^{-1}\otimes\mathbf{\Lambda}^{-1}\|_{\text{op}}
\right).
\]

Note that
\[
\|\boldsymbol{g}\boldsymbol{g}^\top\|_{\text{op}}
=
\|\boldsymbol{g}\|_2^2,
\qquad
\|\mathbf{\Lambda}^{-1}\otimes\mathbf{\Lambda}^{-1}\|_{\text{op}}
=
\|\mathbf{\Lambda}^{-1}\|_{\text{op}}^2
\le
\frac{1}{\underline{\lambda}^2}.
\]
Moreover,
\[
\|\boldsymbol{g}\|_2
=
\frac{1}{2}\|\nu\mathbf{\Lambda}^{-1}-\mathbf{W}\|_F
\le
\frac{1}{2}\nu\|\mathbf{\Lambda}^{-1}\|_F+\frac{1}{2}\|\mathbf{W}\|_F
\le
C_{\boldsymbol{g}}+\frac{1}{2}\|\mathbf{W}\|_F
\le
C_{\boldsymbol{g}}+\frac{1}{2}\tr(\mathbf{W}),
\]
for some constant \(C_{\boldsymbol{g}}<\infty\) depending only on \(\Theta^*\), where we used that
\(\|\mathbf{W}\|_F\le \tr(\mathbf{W})\) for \(\mathbf{W}\in\mathbb{S}_{++}^p\).
Therefore, there exists a constant \(C<\infty\) such that
\[
\|\boldsymbol{g}\|_2^2
\le
C\left(1+\tr(\mathbf{W})^2\right).
\]

Since \(\nu\le \overline{\nu}\), it follows that
\[
\|\nabla^2_{\boldsymbol{\eta}\boldsymbol{\eta}}f\|_{\text{op}}
\le
C\,f(\mathbf{W}\mid\boldsymbol{\theta})
\left(1+\tr(\mathbf{W})^2\right)
\]
for some constant \(C<\infty\) depending only on \(\Theta^*\).
Using the bound (\ref{eq:f_bound}), together with
\(
|\mathbf{W}|^{\frac{\nu-p-1}{2}}
\le
|\mathbf{W}|^{\underline{\alpha}}+|\mathbf{W}|^{\overline{\alpha}},
\)
we obtain
\[
\|\nabla^2_{\boldsymbol{\eta}\boldsymbol{\eta}}f\|_{\text{op}}
\le
C\sum_{\alpha\in\{\underline{\alpha},\overline{\alpha}\}}
|\mathbf{W}|^\alpha
\exp\!\left[-\frac{\underline{\lambda}}{2}\tr(\mathbf{W})\right]
\left(1+\tr(\mathbf{W})^2\right).
\]
Finally, since
\[
1+\tr(\mathbf{W})^2
\le
\bigl(1+|\log|\mathbf{W}||\bigr)
+
\tr(\mathbf{W})^2\bigl(1+|\log|\mathbf{W}||\bigr),
\]
Lemma~\ref{eq:amgm} applied separately with \(m=0\) and \(m=2\), and with \(a=\underline{\lambda}/2\), \(r=1\), \(c=b=1\), yields
\begin{equation}\label{eq:etaeta_bound}
\sup_{\mathbf{W}\in\mathbb{S}_{++}^p,\ \boldsymbol{\theta}\in\Theta^*}\|\nabla^2_{\boldsymbol{\eta}\boldsymbol{\eta}}f(\mathbf{W}\mid\boldsymbol{\theta})\|_{\text{op}}<\infty.
\end{equation}

\paragraph{Summary:} By the bounds (\ref{eq:nunu_bound}), (\ref{eq:etanu_bound}) and (\ref{eq:etaeta_bound}) for the three Hessian blocks, and the finite dimension of the parameter space, there exists a constant \(C_H>0\) such that
\[
\sup_{\mathbf{W}\in\mathbb{S}_{++}^p,\ \boldsymbol{\theta}\in\Theta^*}
\|\nabla^2_{\boldsymbol{\theta}}f(\mathbf{W}\mid\boldsymbol{\theta})\|_{\text{op}}
\le C_H.
\]

For any \(\boldsymbol{\theta}_1,\boldsymbol{\theta}_2\in \Theta^*\), define
\(
\boldsymbol{\theta}(t):=\boldsymbol{\theta}_2 + t(\boldsymbol{\theta}_1-\boldsymbol{\theta}_2), t\in[0,1].
\)
Since \(\Theta^*\) is convex, we have \(\boldsymbol{\theta}(t)\in\Theta^*\) for all \(t\in[0,1]\).
Then, by the fundamental theorem of calculus,
\[
\nabla_{\boldsymbol{\theta}}f(\mathbf{W}\mid \boldsymbol{\theta}_1)
-
\nabla_{\boldsymbol{\theta}}f(\mathbf{W}\mid \boldsymbol{\theta}_2)
=
\int_0^1
\nabla^2_{\boldsymbol{\theta}}f(\mathbf{W}\mid\boldsymbol{\theta}(t))
(\boldsymbol{\theta}_1-\boldsymbol{\theta}_2)\,dt.
\]
Therefore,
\begin{align*}
\|\nabla_{\boldsymbol{\theta}}f(\mathbf{W}\mid \boldsymbol{\theta}_1)
-
\nabla_{\boldsymbol{\theta}}f(\mathbf{W}\mid \boldsymbol{\theta}_2)\|_2
&\le
\int_0^1
\|\nabla^2_{\boldsymbol{\theta}}f(\mathbf{W}\mid\boldsymbol{\theta}(t))\|_{\text{op}}\,
\|\boldsymbol{\theta}_1-\boldsymbol{\theta}_2\|_2\,dt\\
&\le
C_H\|\boldsymbol{\theta}_1-\boldsymbol{\theta}_2\|_2.
\end{align*}
Thus, the first-order uniformly Lipschitz property holds with \(\delta=1\).
\qed


\subsection{Proof of Theorem~\ref{thm:contraction}}

It remains to verify the regularity conditions corresponding to (P.1)--(P.4) in \citet{guha21}. Condition (P.1) holds because \(\Theta^*\) is compact under the half-vectorization parameterization, and Lemmas~\ref{lem:first_order_identifiability} and \ref{lem:first_order_uniform_lipschitz} establish first-order identifiability and the first-order uniform Lipschitz property of the Wishart kernel, respectively. Condition (P.2) holds for the restricted model on \(\Theta^*\), because the base density of the component parameter \(\boldsymbol{\theta}=(\boldsymbol{\eta}^\top,\nu)^\top\) is continuous and strictly positive on \(\Theta^*\), and hence is bounded below by a positive constant by compactness. Condition (P.3) is the local \(\chi^2\)-type bound proved in Lemma~\ref{thm:P3}. Finally, condition (P.4) holds because the shifted Poisson prior on \(K\) assigns positive probability to every positive integer. Therefore, the proof of Theorem 3.1 in \citet{guha21} carries over to our MFM--Wishart model, yielding the claimed posterior consistency for \(K\) and the \(W_1\)-posterior contraction rate.

It should be noted that \citet{guha21} state their MFM prior with conditional weights \(\operatorname{Dirichlet}_k(\gamma/k,\ldots,\gamma/k)\), whereas our model uses \(\operatorname{Dirichlet}_k(\gamma,\ldots,\gamma)\). This difference does not affect the validity of the argument. Indeed, in the proof of Theorem 3.1 in \citet{guha21}, the Dirichlet prior on the weights is used only through the fact that it assigns polynomially positive mass to sufficiently small neighborhoods of the true weight vector. The same property holds for \(\operatorname{Dirichlet}_k(\gamma,\ldots,\gamma)\), with a different polynomial exponent. Hence, replacing \(\operatorname{Dirichlet}_k(\gamma/k,\ldots,\gamma/k)\) by \(\operatorname{Dirichlet}_k(\gamma,\ldots,\gamma)\) changes only constants and polynomial exponents in the prior-mass lower bound, and does not affect either the posterior consistency for \(K\) or the posterior contraction rate.
\qed

\subsection{Proof of Corollary~\ref{cor:cluster_consistency}}

By \citet[Theorem 5.2]{miller18}, applied with $k=k_0$, for every fixed realization of the data as $n\to\infty$,
\[
\left|
\Pi_n(K_{+,n}=k_0\mid \mathbf W_{1:n})
-
\Pi_n(K=k_0\mid \mathbf W_{1:n})
\right|
\longrightarrow 0 .
\]
Since this pointwise convergence holds for every realized data path, it holds in particular almost surely under \(P_{G_0}\).

By Theorem~\ref{thm:contraction}(a),
\[
\Pi_n(K=k_0\mid \mathbf W_{1:n})\longrightarrow 1
\qquad P_{G_0}\text{-a.s.}
\]
Hence,
\[
\begin{aligned}
\left|
\Pi_n(K_{+,n}=k_0\mid \mathbf W_{1:n})-1
\right|
&\le
\left|
\Pi_n(K_{+,n}=k_0\mid \mathbf W_{1:n})
-
\Pi_n(K=k_0\mid \mathbf W_{1:n})
\right| \\
&\quad+
\left|
\Pi_n(K=k_0\mid \mathbf W_{1:n})-1
\right|.
\end{aligned}
\]
Both terms on the right-hand side converge to zero almost surely under \(P_{G_0}\), which proves the claim. \qed


\section{Derivation of the Cluster-Label Update}\label{sec:mcmc_dev}

Throughout the calculations below, for notational convenience, \(\nu\) denotes a fixed degrees-of-freedom value for the cluster under consideration. The resulting expressions are evaluated at \(\nu_c\) for an existing cluster and at \(\widetilde\nu_r\) for an auxiliary new-cluster candidate. 

We first derive the prior collapsed predictive density \(m(\mathbf{W}\mid \nu)\). Recall that the likelihood of one sample is 
\[
f(\mathbf{W}\mid \mathbf{\Sigma},\nu)=\frac{|\mathbf{\Sigma}|^{-\frac{\nu}{2}}}{2^{\frac{\nu p}{2}}\Gamma_p\left(\frac{\nu}{2}\right)}|\mathbf{W}|^{\frac{\nu-p-1}{2}}\exp{\left\{-\frac{1}{2}\tr\left(\mathbf{\Sigma}^{-1}\mathbf{W}\right)\right\}}.
\]

Recall that the prior is \(\mathbf{\Sigma}\sim \mathcal{IW}_p(\mathbf{\Psi}_0,\kappa_0)\) with \(\mathbf{\Psi}_0\in\mathbb{S}_{++}^p\) and \(\kappa_0>p-1\), with the density
\[
p(\mathbf{\Sigma})=\frac{|\mathbf{\Psi}_0|^{\frac{\kappa_0}{2}}}{2^{\frac{\kappa_0p}{2}}\Gamma_p\left(\frac{\kappa_0}{2}\right)}|\mathbf{\Sigma}|^{-\frac{\kappa_0+p+1}{2}}\exp{\left\{-\frac{1}{2}\tr\left(\mathbf{\Psi}_0\mathbf{\Sigma}^{-1}\right)\right\}}.
\]

The prior predictive density is
\begin{align*}
    m(\mathbf{W}\mid \nu) &= \int f(\mathbf{W}\mid \mathbf{\Sigma},\nu)p(\mathbf{\Sigma}) d\mathbf{\Sigma}\\
    &= \int \frac{|\mathbf{\Sigma}|^{-\frac{\nu}{2}}}{2^{\frac{\nu p}{2}}\Gamma_p\left(\frac{\nu}{2}\right)}|\mathbf{W}|^{\frac{\nu-p-1}{2}}\exp{\left\{-\frac{1}{2}\tr\left(\mathbf{\Sigma}^{-1}\mathbf{W}\right)\right\}}\\
    &\qquad\qquad\qquad\cdot\frac{|\mathbf{\Psi}_0|^{\frac{\kappa_0}{2}}}{2^{\frac{\kappa_0p}{2}}\Gamma_p\left(\frac{\kappa_0}{2}\right)}|\mathbf{\Sigma}|^{-\frac{\kappa_0+p+1}{2}}\exp{\left\{-\frac{1}{2}\tr\left(\mathbf{\Psi}_0\mathbf{\Sigma}^{-1}\right)\right\}}d\mathbf{\Sigma}\\
    &= \frac{1}{2^{\frac{(\kappa_0+\nu) p}{2}}\Gamma_p\left(\frac{\nu}{2}\right)\Gamma_p\left(\frac{\kappa_0}{2}\right)}|\mathbf{W}|^{\frac{\nu-p-1}{2}}|\mathbf{\Psi}_0|^{\frac{\kappa_0}{2}}\int\underbrace{ |\mathbf{\Sigma}|^{-\frac{(\kappa_0+\nu)+p+1}{2}}\exp{\left\{-\frac{1}{2}\tr\!\left((\mathbf{\Psi}_0+\mathbf{W})\mathbf{\Sigma}^{-1}\right)\right\}}}_{\text{kernel of \(\mathcal{IW}_p(\mathbf{\Psi}_0+\mathbf{W},\kappa_0+\nu)\)}}d\mathbf{\Sigma}\\
    &= \frac{1}{2^{\frac{(\kappa_0+\nu) p}{2}}\Gamma_p\left(\frac{\nu}{2}\right)\Gamma_p\left(\frac{\kappa_0}{2}\right)}|\mathbf{W}|^{\frac{\nu-p-1}{2}}|\mathbf{\Psi}_0|^{\frac{\kappa_0}{2}}\cdot\frac{2^{\frac{(\kappa_0+\nu)p}{2}}\Gamma_p\left(\frac{\kappa_0+\nu}{2}\right)}{|\mathbf{\Psi}_0+\mathbf{W}|^{\frac{\kappa_0+\nu}{2}}}\\
    &= \frac{\Gamma_p\left(\frac{\nu+\kappa_0}{2}\right)}{\Gamma_p\left(\frac{\nu}{2}\right)\Gamma_p\left(\frac{\kappa_0}{2}\right)}
\cdot
\frac{|\mathbf{W}|^{\frac{\nu-p-1}{2}}|\mathbf{\Psi}_0|^{\frac{\kappa_0}{2}}}{|\mathbf{W} + \mathbf{\Psi}_0|^{\frac{\nu+\kappa_0}{2}}}.
\end{align*}

Next, we derive the collapsed posterior predictive distribution \(p(\mathbf{W}_i\mid c,\nu, \{\mathbf{W}_j\}_{j:z_j=c,\,j\neq i})\) for any occupied cluster \(c\). To do so, we first derive the posterior of \(\mathbf{\Sigma}_c\) given the \(n_c\) samples within the cluster \(c\). Denote the index set \(\mathcal{I}_c:=\{i:z_i=c\}\) and \(\mathbf{S}_c:=\sum_{i\in\mathcal{I}_c}\mathbf{W}_i\). The conditional posterior of \(\mathbf{\Sigma}_c\) given the samples and \(\nu\) is
\begin{align*}
    p(\mathbf{\Sigma}_c\mid \{\mathbf{W}_i\}_{i\in\mathcal{I}_c},\nu)&\propto p(\mathbf{\Sigma}_c)\prod_{i\in\mathcal{I}_c}f(\mathbf{W}_i\mid \mathbf{\Sigma}_c,\nu)\\
    & \propto |\mathbf{\Sigma}_c|^{-\frac{\kappa_0+p+1}{2}}\exp{\left\{-\frac{1}{2}\tr\left(\mathbf{\Psi}_0\mathbf{\Sigma}_c^{-1}\right)\right\}}\cdot |\mathbf{\Sigma}_c|^{-\frac{n_c\nu}{2}}\exp{\left\{-\frac{1}{2}\tr{\left(\mathbf{\Sigma}_c^{-1}\mathbf{S}_c\right)}\right\}}\\
    & \propto |\mathbf{\Sigma}_c|^{-\frac{(\kappa_0+n_c\nu)+p+1}{2}}\exp{\left\{-\frac{1}{2}\tr{\left[(\mathbf{\Psi}_0+\mathbf{S}_c)\mathbf{\Sigma}_c^{-1}\right]}\right\}}
\end{align*}

So, we have
\[
\mathbf{\Sigma}_c \mid \{\mathbf{W}_i\}_{i\in\mathcal{I}_c},\nu\sim \mathcal{IW}_p(\mathbf{\Psi}_0+\mathbf{S}_c,\kappa_0+n_c\nu)
\]

Integrating out \(\mathbf\Sigma_c\) similarly gives the collapsed likelihood for cluster \(c\):
\begin{align*}
    m(\{\mathbf{W}_i\}_{i\in\mathcal{I}_c}\mid \nu)=\frac{\Gamma_p\left(\frac{\kappa_0+n_c\nu}{2}\right)}{\Gamma_p\left(\frac{\kappa_0}{2}\right)\Gamma_p\left(\frac{\nu}{2}\right)^{n_c}}\cdot \frac{|\mathbf{\Psi}_0|^{\frac{\kappa_0}{2}}\prod_{i\in\mathcal{I}_c}|\mathbf{W}_i|^{\frac{\nu-p-1}{2}}}{|\mathbf{\Psi}_0+\mathbf{S}_c|^{\frac{\kappa_0+n_c\nu}{2}}}
\end{align*}

We next obtain the predictive density needed for the allocation update. Temporarily remove observation \(i\) from its current cluster and write \(\boldsymbol{z}_{-i}=(z_1,\dots,z_{i-1},z_{i+1},\dots,z_n)^\top\). For any existing cluster \(c\) given \(\boldsymbol{z}_{-i}\), we denote \(n_{c,-i}=\sum_{j\neq i}\mathbbm{1}(z_j=c)\) and \(\mathbf{S}_{c,-i}=\sum_{j:z_j=c,\,j\neq i}\mathbf{W}_j\). Then, the posterior predictive density, for an existing cluster \(c\), is
\begin{align*}
    p(\mathbf{W}_i\mid c,\nu,\{\mathbf{W}_j:z_j=c,j\neq i\}) &= \frac{m(\{\mathbf{W}_j:z_j=c,j\neq i\}\bigcup \{\mathbf{W}_i\}\mid c, \nu)}{m(\{\mathbf{W}_j:z_j=c,j\neq i\}\mid c, \nu)}\\
    &=\frac{\Gamma_p\left(\frac{\kappa_0+(n_{c,-i}+1)\nu}{2}\right)}{\Gamma_p\left(\frac{\kappa_0+n_{c,-i}\nu}{2}\right)\Gamma_p\left(\frac{\nu}{2}\right)}\cdot
\frac{|\mathbf{W}_i|^{\frac{\nu-p-1}{2}}|\mathbf{\Psi}_0+\mathbf{S}_{c,-i}|^{\frac{\kappa_0+n_{c,-i}\nu}{2}}}{|\mathbf{\Psi}_0+\mathbf{S}_{c,-i}+\mathbf{W}_i|^{\frac{\kappa_0+(n_{c,-i}+1)\nu}{2}}}.
\end{align*}

For an existing cluster \(c\), setting \(\nu=\nu_c\) in the preceding predictive density gives \eqref{eq:inf-predictive-existing}. For an auxiliary new-cluster candidate, setting \(\nu=\widetilde\nu_r\) in \(m(\mathbf W_i\mid\nu)\) gives \eqref{eq:inf-predictive-new}.

We now derive the allocation update. This follows directly from the Algorithm~8 of \citet{neal00}, and interested readers can refer to that paper for more details. Here we give a brief justification of the algorithm in our MFM--Wishart context. Let \(h\) denote the prior density of the degrees-of-freedom parameter, and let \(b_i\) denote the MFM prior weight for creating a new cluster:
\[
h(\nu)=\frac{\mathbf 1\{\nu\in[\nu_L,\nu_U]\}}{\nu_U-\nu_L},
\qquad
b_i=\gamma\frac{V_n(K^*+1)}{V_n(K^*)},
\]
where \(K^*=|\mathcal C_{-i}|\). Throughout the following derivation, we condition on \(\boldsymbol z_{-i}\), the parameters \(\{\nu_c:c\in\mathcal C_{-i}\}\), and all observations.

Combining the MFM prior weights in (\ref{eq:prior_mfm}) with the predictive densities gives the existing-cluster weight \(w_c\) in (\ref{eq:inf-allocation-existing}). For a new cluster, the unnormalized joint posterior density of its allocation and degrees-of-freedom parameter \(\nu\) is \(b_i h(\nu)m(\mathbf W_i\mid\nu)\). Integrating over \(\nu\) therefore gives the new-cluster allocation weight
\[
w_{\mathrm{new}}=
b_i\int_{\nu_L}^{\nu_U}h(\nu)m(\mathbf W_i\mid\nu)\,d\nu.
\]

We note that, by using auxiliary parameters, we can avoid evaluating this integral. Introduce \(M\) auxiliary values \(\widetilde{\boldsymbol\nu}=(\widetilde\nu_1,\ldots,\widetilde\nu_M)\) and an indicator \(J\) specifying the destination of observation \(i\). Its possible values are
\(
J\in\mathcal C_{-i}\cup\{(\mathrm{new},r):r=1,\ldots,M\},
\)
where \(J=(\mathrm{new},r)\) assigns observation \(i\) to a new cluster with parameter \(\widetilde\nu_r\). With the conditioning above suppressed, define the augmented density
\[
\widetilde p(J,\widetilde{\boldsymbol\nu})
\propto
\left\{\prod_{s=1}^{M}h(\widetilde\nu_s)\right\}
\begin{cases}
w_c,
& J=c\in\mathcal C_{-i},\\[3pt]
\dfrac{b_i}{M}m(\mathbf W_i\mid\widetilde\nu_r),
& J=(\mathrm{new},r).
\end{cases}
\]

For an existing-cluster destination, integrating out all auxiliary values recovers the weight \(w_c\). For a new-cluster destination, retain the selected parameter, denote its value by \(\nu\), and integrate out the other auxiliary values. Summing over the \(M\) possible selected positions gives
\[
\sum_{r=1}^{M}\frac{b_i}{M}h(\nu)m(\mathbf W_i\mid\nu)
=
b_i h(\nu)m(\mathbf W_i\mid\nu),
\]
which is the required joint posterior density for the new cluster and its parameter, up to normalization. Thus, the augmented distribution has the correct marginal posterior. Given the auxiliary values, the common factor \(\prod_{s=1}^{M}h(\widetilde\nu_s)\) cancels, leaving the allocation weights \(w_c\) and \((b_i/M)m(\mathbf W_i\mid\widetilde\nu_r)\) in (\ref{eq:inf-allocation-existing}) and (\ref{eq:inf-allocation-auxiliary}).

We emphasize that the auxiliary values must also be generated conditional on the current state. If observation \(i\) is not a singleton before removal, all \(M\) values are drawn independently from \(h\). If it is a singleton with parameter \(\nu_{\mathrm{old}}\), this parameter is retained in a uniformly selected auxiliary position, and the remaining \(M-1\) values are drawn independently from \(h\). Permuting the auxiliary positions and their destination labels does not change the distribution of the resulting partition and its cluster parameters. We can therefore retain \(\nu_{\mathrm{old}}\) in the first position, as specified in Section~\ref{subsubsec:inference-labels}.

We then sample \(J\) using the normalized allocation weights. If a new-cluster candidate is selected, its parameter becomes the parameter of the new cluster. All unused auxiliary values are discarded. These steps give the partially collapsed MFM adaptation of Algorithm~8 of \citet{neal00} and preserve the target posterior for every fixed \(M\ge1\).

\section{Additional Simulation Details and Results}\label{sec:app_addition}

\subsection{Implementation details of baseline methods}\label{app:baseline}

Hierarchical clustering is implemented using the \texttt{hclust} function in base R, with Ward's linkage method. PAM is implemented using the R package \texttt{cluster} \citep{maec19}. For hierarchical clustering and PAM, we use the non-Euclidean Riemannian distance implemented in the R package \texttt{shapes}, which was developed in \citet{dryden09}.

For the unpenalized Wishart FMM and the penalized Wishart mixture model, we set the number of clusters to be in \(\{1,2,3,4,5,6,7,8,9,10\}\). For Penalized FMM, we tune \(\lambda \in \{0,2,4,6,8,10\}\). For both FMM and Penalized FMM, the model with the best BIC score is selected as the final model.

\subsection{Cluster-specific scale matrices \(\mathbf{\Sigma}_k\) for medium- and large-matrix settings}

Figures~\ref{fig:medium_sigma} and Figure~\ref{fig:large_S} provide additional details on the cluster-specific scale matrices used in the simulations. All matrices have unit diagonal entries, and the off-diagonal entries determine the cluster-specific covariance patterns. In the medium-matrix setting in Figure~\ref{fig:medium_sigma}, the components are designed to differ in both the magnitude and the sign of the off-diagonal associations. For \(k_0=3\), the three components represent broadly positive, weakly associated, and mixed-sign covariance structures. For \(k_0=5\), the design further increases heterogeneity by including strongly positive, weakly correlated, diagonal, and sign-changing patterns. These settings create a range of clustering difficulties.

\begin{figure}[h!]
\centering

\begin{subfigure}{0.95\textwidth}
  \centering
  \includegraphics[width=\textwidth]{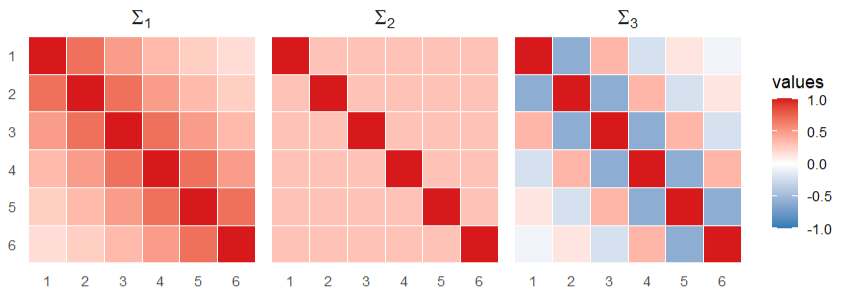}
  \caption{\(k_0=3\): cluster-specific \(\mathbf{\Sigma}_k\)'s.}
\end{subfigure}

\begin{subfigure}{0.95\textwidth}
  \centering
  \includegraphics[width=\textwidth]{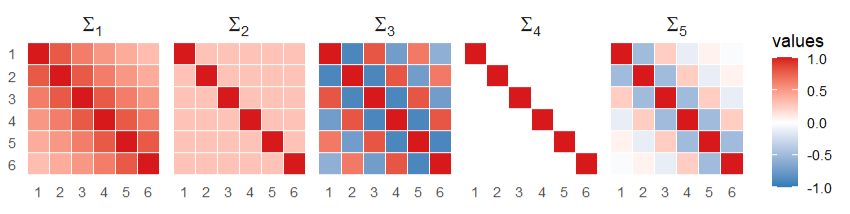}
  \caption{\(k_0=5\): cluster-specific \(\mathbf{\Sigma}_k\)'s.}
\end{subfigure}

\caption{Cluster-specific \(\mathbf{\Sigma}_k\) settings in the medium-matrix (\(p=6\)) simulations.}
\label{fig:medium_sigma}
\end{figure}

In the large-matrix setting, Figure~\ref{fig:large_S} considers a more structured design. The first two scale matrices, \(\mathbf{\Sigma}_1\) and \(\mathbf{\Sigma}_2\), are fixed across replicates and exhibit block-sparse correlation patterns. The third component, \(\mathbf{\Sigma}_3\), is generated randomly in each replicate by standardizing a Wishart draw into a correlation matrix.

\begin{figure}[h!]
\centering
\includegraphics[width=0.99\textwidth]{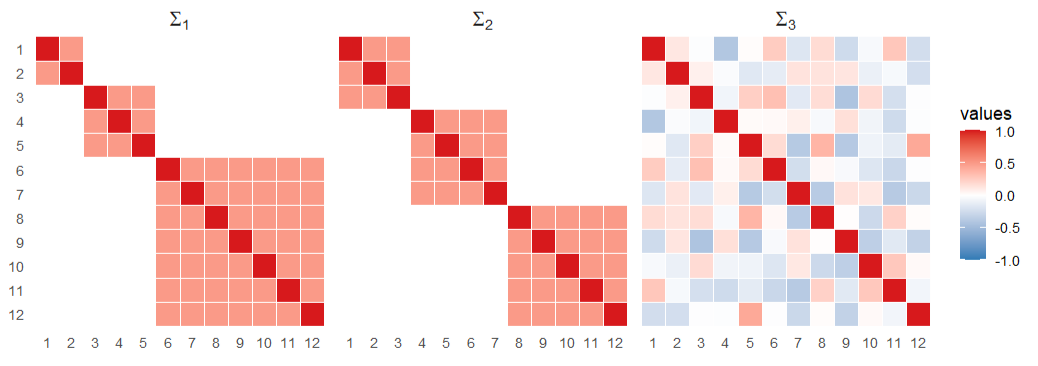}
\caption{Cluster-specific \(\mathbf{\Sigma}_k\) settings in the large-matrix (\(p=12\)) simulations. We note that \(\mathbf{\Sigma}_1\) and \(\mathbf{\Sigma}_2\) are fixed across the 100 replicates, while \(\mathbf{\Sigma}_3\) is randomly generated for each replicate, as explained in Section~\ref{sec:sim_setting}.}
\label{fig:large_S}
\end{figure}

\clearpage
\newpage

\subsection{Additional results for the well-specified simulations}\label{app:add_sim}

Table~\ref{tab:large-matrix-k-accuracy} summarizes recovery of the true number of clusters in the large-matrix simulations. We report the accuracy of the Dahl-based point estimate, together with the mean and standard deviation of the posterior probability assigned to the true number of clusters across \(100\) simulation replicates. MFM--Wishart and DPM--Wishart exhibit similar point-estimation accuracy, which increases with the sample size.

\begin{table}[!htbp]
\centering
\caption{Recovery of the true number of clusters in the large-matrix setting (\(p=12\), \(k_0=3\)). For each method and sample size, Accuracy is the proportion of \(100\) simulation replicates for which the Dahl-based estimate of \(K_{+,n}\) equals \(k_0\). Mean and SD are the mean and standard deviation, respectively, of the posterior probability \(\mathbb{P}(K_{+,n}=k_0\mid\mathrm{data})\) across the same replicates.}
\label{tab:large-matrix-k-accuracy}

\small
\setlength{\tabcolsep}{4pt}
\renewcommand{\arraystretch}{1.1}

\begin{tabular}{lccc ccc ccc}
\toprule
& \multicolumn{3}{c}{\(n=50\)} & \multicolumn{3}{c}{\(n=100\)} & \multicolumn{3}{c}{\(n=200\)} \\
\cmidrule(lr){2-4} \cmidrule(lr){5-7} \cmidrule(lr){8-10}
Method
& Accuracy & Mean & SD
& Accuracy & Mean & SD
& Accuracy & Mean & SD \\
\midrule
MFM--Wishart (ours)
& 0.80 & 0.80 & 0.40
& 0.92 & 0.92 & 0.27
& 0.99 & 0.99 & 0.10 \\

DPM--Wishart
& 0.76 & 0.77 & 0.42
& 0.91 & 0.91 & 0.29
& 0.99 & 0.99 & 0.10 \\
\bottomrule
\end{tabular}
\end{table}

In the large-matrix setting, the estimated posterior probabilities assigned to the true number of clusters tend to be close to either zero or one across replicates. One possible explanation is that the information contained in larger matrices leads to more concentrated posterior support for particular cluster configurations. However, limited movement between competing partitions under single-observation allocation updates may also contribute to this pattern. The present results do not distinguish between these statistical and computational explanations. A systematic investigation of their relative contributions, including comparisons with samplers incorporating split--merge moves, is left for future work.

We next compare the MCMC runtimes of MFM--Wishart and DPM--Wishart across the simulation settings, reporting results for the balanced designs. Figure~\ref{fig:mcmc-time-balanced-small-medium} shows the runtimes for the small-matrix and medium-matrix simulations, while Table~\ref{tab:large-matrix-mcmc-time} reports the results for the large-matrix setting. MFM--Wishart has computational costs comparable to those of DPM--Wishart across the settings considered. For both models, runtime increases with sample size, consistent with the sample-wise allocation updates used in our MCMC algorithms. However, Variations with \(p\) and \(k_0\) are comparatively modest and not consistently monotone over the settings considered. These findings support the computational feasibility of MFM--Wishart, although further experiments would be needed to assess its scalability to substantially higher dimensions.

\begin{figure}[!htbp]
\centering
\includegraphics[width=0.95\textwidth]{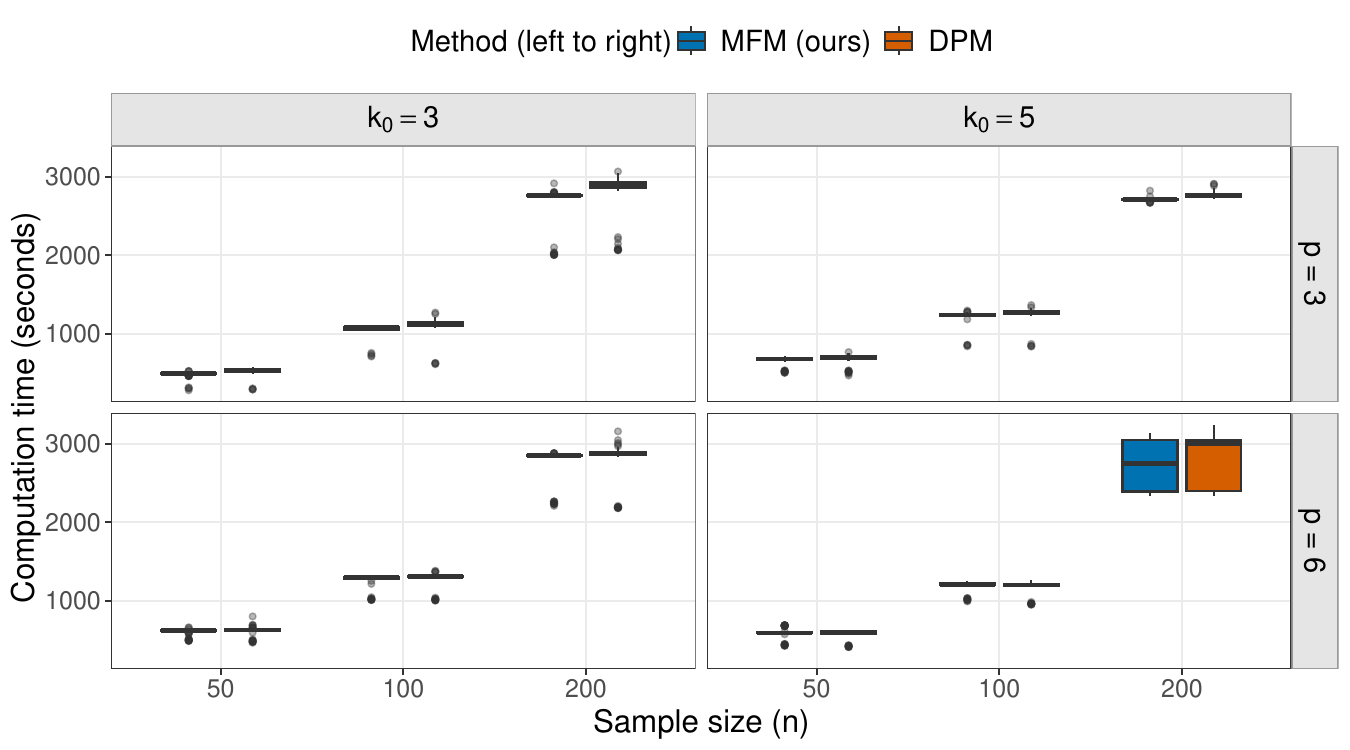}
\caption{
MCMC computation time under the balanced cluster-size configuration for the small-matrix \((p=3)\) and medium-matrix \((p=6)\) simulations. The two columns correspond to \(k_0=3\) and \(k_0=5\), respectively, and the rows correspond to \(p=3\) and \(p=6\). The x-axis shows the sample size \(n\). Boxplots summarize 100 replicated datasets. Methods are ordered from left to right as MFM--Wishart and DPM--Wishart.
}
\label{fig:mcmc-time-balanced-small-medium}
\end{figure}

\begin{table}[!htbp]
\centering
\caption{MCMC computation time in the large-matrix setting. For each method and each sample size, the table reports the mean computation time and standard deviation, in seconds, over 100 replicated datasets.}
\label{tab:large-matrix-mcmc-time}
\small
\begin{tabular}{l|cc|cc|cc}
\hline
& \multicolumn{2}{c|}{\(n=50\)} & \multicolumn{2}{c|}{\(n=100\)} & \multicolumn{2}{c}{\(n=200\)}  \\
\cline{2-7}
Method & Mean & SD & Mean & SD & Mean & SD  \\
\hline
MFM (ours)    & 517.66 & 47.52 & 1190.39 & 113.51 & 2518.71 & 202.41 \\
DPM           & 503.70 & 43.81 & 1184.01 & 114.25 & 2515.57 & 210.49 \\
\hline
\end{tabular}
\end{table}

\clearpage
\newpage
\subsection{Simulations of misspecified cases}\label{app:misspecified}

\subsubsection{Wishart density as a working likelihood}
\label{app:temporal_moments}

The discussion in Section~\ref{sec:wishart} can be made precise in the present misspecified setting as follows. Fix a cluster \(k\), and suppose that \(\{\boldsymbol{x}_t\}_{t=1}^T\) is a jointly Gaussian, second-order stationary \(p\)-variate process with \(E(\boldsymbol{x}_t)=0\) and \(\operatorname{Cov}(\boldsymbol{x}_t)=\mathbf{\Sigma}_k\). Define
\(
\mathbf{S}=\sum_{t=1}^T \boldsymbol{x}_t\boldsymbol{x}_t^\top\), and
\(
\widehat{\mathbf{\Sigma}}=\frac{1}{T}S.
\)
Then \(E(\widehat{\mathbf{\Sigma}})=\mathbf{\Sigma}\). Moreover, if we write
\[
\mathbf{\Gamma}(h)=\operatorname{Cov}(\boldsymbol{x}_t,\boldsymbol{x}_{t+h}),
\]
it can be shown that, for any \(1 \le i,j,r,s \le p\),
\[
\operatorname{Cov}(\widehat{\mathbf{\Sigma}}_{ij},\widehat{\mathbf{\Sigma}}_{rs})
=
\frac{1}{T^2}
\sum_{h=-(T-1)}^{T-1}
(T-|h|)
\left\{
\mathbf{\Gamma}_{ir}(h)\mathbf{\Gamma}_{js}(h)
+
\mathbf{\Gamma}_{is}(h)\mathbf{\Gamma}_{jr}(h)
\right\},
\]
where we write \(\mathbf{\Gamma}_{ir}(h)\) as the \((i,r)\)-th entry of \(\mathbf{\Gamma}(h)\). In the i.i.d.\ Gaussian case, \(\mathbf{\Gamma}(0)=\mathbf{\Sigma}\) and \(\mathbf{\Gamma}(h)=0\) for \(h\neq 0\), so the above expression reduces to
\[
\operatorname{Cov}(\widehat{\mathbf{\Sigma}}_{ij},\widehat{\mathbf{\Sigma}}_{rs})
=
\frac{1}{T}
\left\{
(\mathbf{\Sigma})_{ir}(\mathbf{\Sigma})_{js}
+
(\mathbf{\Sigma})_{is}(\mathbf{\Sigma})_{jr}
\right\},
\]
which is the usual covariance formula associated with the Wishart distribution.

Consider the following stationary vector autoregressive model of order one, VAR(1):
\[
\boldsymbol{x}_t=\phi \boldsymbol{x}_{t-1}+\boldsymbol{\varepsilon}_t,
\qquad
\boldsymbol{\varepsilon}_t \overset{\mathrm{i.i.d.}}{\sim} \mathcal{N}_p\!\left(0,(1-\phi^2)\mathbf{\Sigma}\right),\qquad -1<\phi<1
\]
with stationary initial distribution \(\boldsymbol{x}_1\sim \mathcal{N}_p(0,\mathbf{\Sigma})\), we have
\[
\mathbf{\Gamma}(h)=\phi^{|h|}\mathbf{\Sigma}.
\]
Substituting this into the previous display yields
\[
\operatorname{Cov}(\widehat{\mathbf{\Sigma}}_{ij},\widehat{\mathbf{\Sigma}}_{rs})
=
\frac{1}{T}
\left[
1+2\sum_{h=1}^{T-1}\left(1-\frac{h}{T}\right)\phi^{2h}
\right]
\left\{
(\mathbf{\Sigma})_{ir}(\mathbf{\Sigma})_{js}
+
(\mathbf{\Sigma})_{is}(\mathbf{\Sigma})_{jr}
\right\}.
\]
Thus, relative to the i.i.d.\ case, temporal dependence changes the second-order fluctuations of \(\widehat{\mathbf{\Sigma}}\) by the multiplicative factor
\(
1+2\sum_{h=1}^{T-1}\left(1-\frac{h}{T}\right)\phi^{2h},
\)
which is larger than one when \(\phi\neq0\). Equivalently, one may define the covariance-level effective sample size
\[
\nu_{\mathrm{eff}}(T,\phi)
=
\frac{T}{
1+2\sum_{h=1}^{T-1}\left(1-\frac{h}{T}\right)\phi^{2h}
},
\]
so that
\[
\operatorname{Cov}(\widehat{\mathbf{\Sigma}}_{ij},\widehat{\mathbf{\Sigma}}_{rs})
=
\frac{1}{\nu_{\mathrm{eff}}(T,\phi)}
\left\{
(\mathbf{\Sigma})_{ir}(\mathbf{\Sigma})_{js}
+
(\mathbf{\Sigma})_{is}(\mathbf{\Sigma})_{jr}
\right\}.
\]
This shows that temporal dependence preserves the first-order mean structure, but reduces the effective sample size from \(T\) to \(\nu_{\mathrm{eff}}(T,\phi)\). In this sense, the Wishart degrees of freedom may be interpreted as an effective sample size parameter under temporal dependence.

This also explains the rescaling used in the simulation below. Since
\(
\mathbf{W}=\frac{\nu}{T}\mathbf{S}=\nu\widehat{\mathbf{\Sigma}},
\)
we have
\(
E(\mathbf{W})=\nu\mathbf{\Sigma}.
\)
In addition,
\[
\operatorname{Cov}((\mathbf{W})_{ij},(\mathbf{W})_{rs})
=
\frac{\nu^2}{\nu_{\mathrm{eff}}(T,\phi)}
\left\{
(\mathbf{\Sigma})_{ir}(\mathbf{\Sigma})_{js}
+
(\mathbf{\Sigma})_{is}(\mathbf{\Sigma})_{jr}
\right\}.
\]
Therefore, when \(T\) is chosen so that \(\nu_{\mathrm{eff}}(T,\phi)\) is close to the benchmark value \(\nu\), the mean structure of \(\mathbf{W}\) matches that of a \(\mathcal{W}_p(\mathbf{\Sigma},\nu)\) distribution, and its covariance scale is also comparable; indeed, the covariance formulas coincide exactly when \(\nu_{\mathrm{eff}}(T,\phi)=\nu\). This yields a useful misspecified setting in which the observed matrices within each cluster are not exactly Wishart-distributed, but remain comparable to the correctly specified Wishart setting in terms of first-order scale and effective information.

\subsubsection{Simulation settings}

In this study, we consider a misspecified setting in which the observed matrices within each cluster do not follow Wishart distributions exactly, but for which the Wishart model can still serve as a useful approximation. This setting is motivated by covariance matrices computed from multivariate time-series data, such as fNIRS signals, where temporal dependence is typically present. In such cases, the resulting lag-0 covariance matrices generally do not have an exact Wishart distribution, even though a Wishart model may still provide a reasonable working likelihood.

To mimic this situation, we consider the balanced setting with \(k_0=3\) true clusters and total sample sizes \(n=50,100,200\). We consider the small- and medium-matrix settings with \(p=3\) and \(p=6\). As in the \(k_0=3\) well-specified settings, we use the same three cluster-specific scale matrices \(\mathbf{\Sigma}_1\), \(\mathbf{\Sigma}_2\), and \(\mathbf{\Sigma}_3\), shown in Figures~\ref{fig:small_sigma} and~\ref{fig:medium_sigma}. We consider two levels of temporal dependence, \(\phi=0.5\) and \(\phi=0.8\), where \(\phi\) is the autoregressive coefficient in a stationary VAR(1) process.

For each subject assigned to cluster \(k\), we generate a \(p\)-variate stationary VAR(1) process
\begin{align*}
\boldsymbol{x}_t = \phi \boldsymbol{x}_{t-1} + \boldsymbol{\varepsilon}_t,
\qquad
\boldsymbol{\varepsilon}_t \sim_{i.i.d} \mathcal{N}_p\bigl(\mathbf{0},(1-\phi^2)\mathbf{\Sigma}_k\bigr),
\end{align*}
with stationary initial distribution \(\boldsymbol{x}_1 \sim \mathcal{N}_p(\mathbf{0},\mathbf{\Sigma}_k)\), so that the marginal covariance of \(\boldsymbol{x}_t\) is \(\mathbf{\Sigma}_k\). For each subject, we then compute the scatter matrix \(\mathbf{S}=\sum_{t=1}^{T}\boldsymbol{x}_t\boldsymbol{x}_t^\top\) and rescale it as \(\mathbf{W}=\frac{\nu_k}{T}\mathbf{S}\), where \(\nu_k\) is a benchmark degrees of freedom parameter. This rescaling ensures that \(\mathbb{E}(\mathbf{W})=\nu_k\mathbf{\Sigma}_k\), matching the mean structure of a Wishart distribution with parameters \((\mathbf{\Sigma}_k,\nu_k)\), while still preserving the temporal dependence induced by the VAR(1) process. 

To make the comparison with the correctly specified Wishart setting meaningful, we choose the time-series length \(T\) so that the covariance-level effective sample sizes under the VAR(1) process are close to the benchmark values \((\nu_1,\nu_2,\nu_3)=(8,10,12)\) as they are the degrees of freedom of the three clusters in the well-specified case in Section~\ref{sec:sim_setting}. The resulting choices of \(T\) for the three clusters are \((13, 16, 19)\) for \(\phi=0.5\) and \((34, 43, 52)\) for \(\phi=0.8\). With this construction, the resulting matrices are not exactly Wishart-distributed, but they remain comparable to the Wishart setting in terms of first-order scale and effective information. This allows us to assess the robustness of MFM--Wishart and the competing methods under model misspecification caused by temporal dependence.

We keep all other simulation settings the same as in the well-specified small- and medium-matrix settings, including prior hyperparameters, MCMC settings, and baseline-model specifications.

\subsubsection{Results of misspecified cases}

Figure~\ref{fig:ari_misspecified} reports the ARI results under the misspecified setting for the small- and medium-matrix simulations. Across all experimental conditions, MFM--Wishart remains among the best-performing methods, suggesting that its clustering performance is robust to model misspecification induced by temporal dependence. In the small-matrix setting (\(p=3\)), MFM--Wishart and DPM--Wishart generally achieve the highest ARI values, whereas the finite-mixture and distance-based methods perform less well, especially when the sample size is small. In particular, FMM and Penalized FMM exhibit substantial variability at \(n=50\), while HC and PAM have generally lower ARI than our MFM--Wishart. As the sample size increases, the ARI of all methods improves, but MFM--Wishart remains highly competitive throughout. In the medium-matrix setting (\(p=6\)), the clustering task becomes much easier, and both MFM--Wishart and DPM--Wishart attain very high ARI values, often near one, even at relatively small sample sizes. The results under \(\phi=0.5\) and \(\phi=0.8\) are qualitatively very similar, indicating that the relative performance of the competing methods is stable across the two temporal-dependence settings after controlling the covariance-level effective sample size.

\begin{figure}[h!]
\centering
\includegraphics[width=0.99\textwidth]{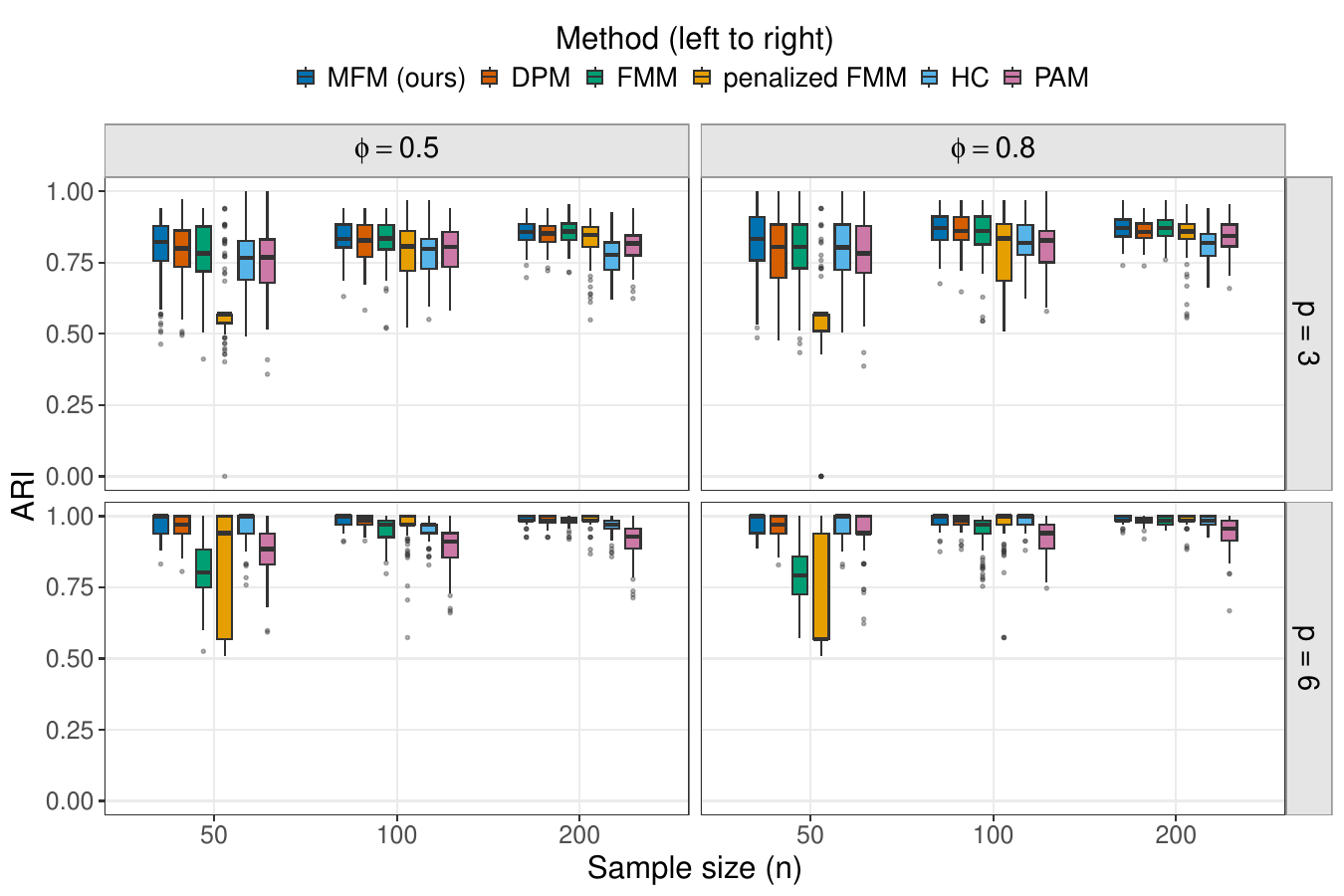}
\caption{ARI under the misspecified setting for the small-matrix (\(p=3\)) and medium-matrix (\(p=6\)) simulations. The rows correspond to \(p=3\) and \(p=6\). The two columns correspond to the autoregressive coefficients \(\phi=0.5\) and \(\phi=0.8\), respectively, and the x-axis shows the sample size \(n\). The true number of clusters is fixed at \(k_0=3\).}
\label{fig:ari_misspecified}
\end{figure}

Figure~\ref{fig:acc_misspecified} reports the accuracy of recovering the true number of clusters under the same misspecified setting. Here the advantage of MFM--Wishart over DPM--Wishart is more pronounced. MFM--Wishart recovers the true number of clusters substantially more often than DPM--Wishart across all settings. Moreover, the accuracy of MFM--Wishart increases with the sample size and becomes high at \(n=200\), whereas the corresponding accuracy of DPM--Wishart remains clearly lower. Thus, although the two methods can yield similar ARI values, MFM--Wishart provides much more reliable inference on the number of clusters. 

Overall, these results show that MFM--Wishart is robust to the misspecification scenarios considered here and that its advantage over DPM--Wishart is especially clear for estimating the number of clusters in the more difficult low-dimensional settings.

\begin{figure}[h]
\centering
\includegraphics[width=0.99\textwidth]{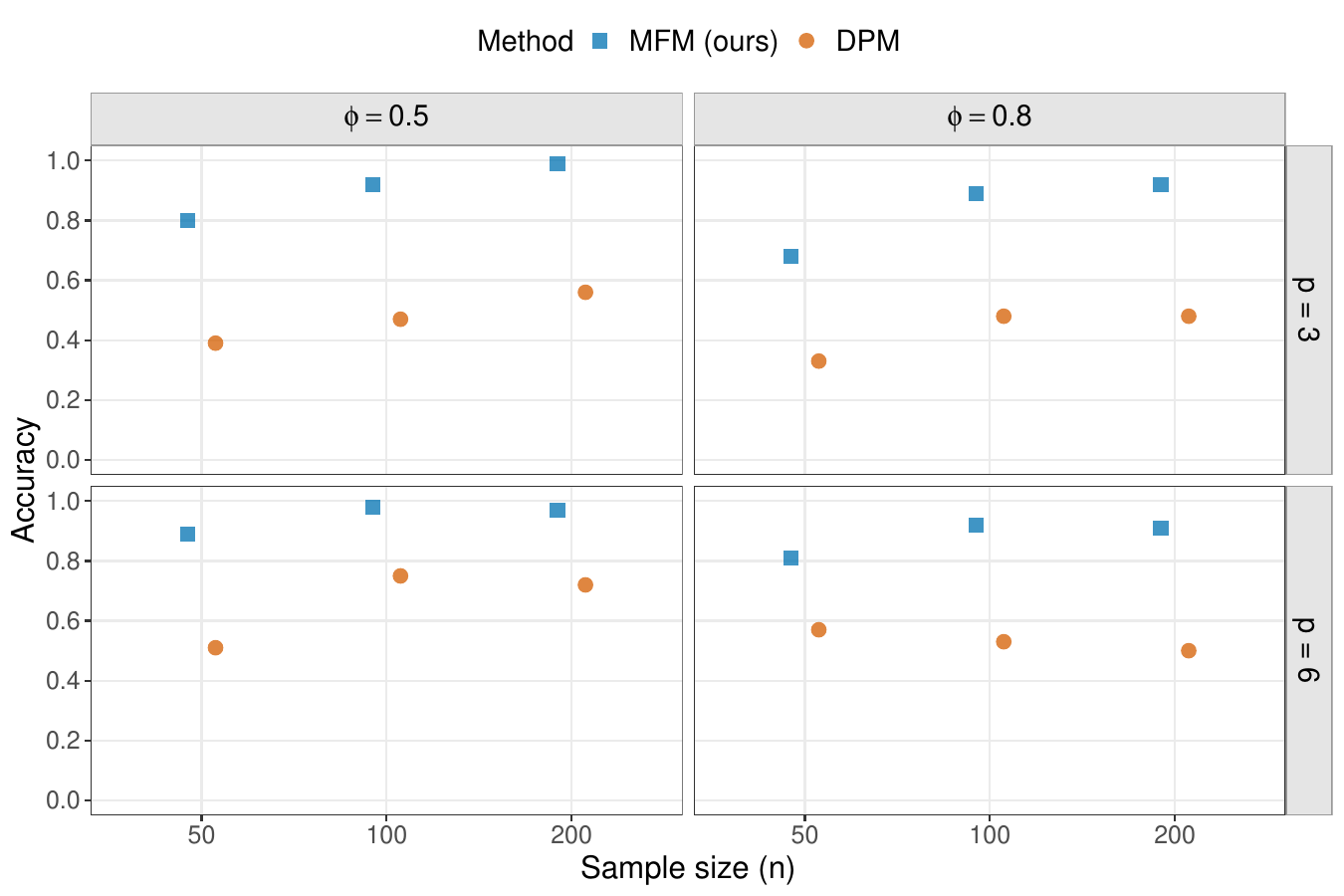}
\caption{Accuracy of recovering the true number of clusters under the misspecified setting for the small-matrix (\(p=3\)) and medium-matrix (\(p=6\)) simulations. For each replicate, accuracy is defined by whether the Dahl-based estimate of \(K_{+,n}\) equals the true number of clusters \(k_0\). The two columns correspond to the autoregressive coefficients \(\phi=0.5\) and \(\phi=0.8\), respectively, and the x-axis shows the sample size \(n\). The true number of clusters is fixed at \(k_0=3\).}
\label{fig:acc_misspecified}
\end{figure}

\subsection{MCMC diagnostics}
\label{app:application_trace}

For the application analysis, we run a single MCMC chain for \(10{,}000\) iterations and discard the first \(4{,}000\) iterations as burn-in. We assess sampling efficiency using the bulk and tail effective sample sizes (ESS) proposed by \citet{veh21}, computed from the \(6{,}000\) post-burn-in draws. Bulk ESS is computed from rank-normalized draws using split chains, whereas tail ESS is the minimum of the ESS estimates associated with the \(5\%\) and \(95\%\) quantiles.

Because the number of occupied clusters can vary across iterations and component labels are arbitrary, we monitor \(K_{+,n}\) and three label-invariant summaries of the cluster-specific degrees of freedom parameters. At iteration \(t\), these summaries are defined as
\begin{align*}
\nu_{\min}^{(t)}= \min_{1 \leq k \leq K_{+,n}^{(t)}} \nu_k^{(t)},\quad
\nu_{\max}^{(t)}
=\max_{1 \leq k \leq K_{+,n}^{(t)}} \nu_k^{(t)},\quad
\nu_{\mathrm{w}}^{(t)}&=\frac{1}{n}\sum_{k=1}^{K_{+,n}^{(t)}}n_k^{(t)} \nu_k^{(t)},
\end{align*}
where \(n_k^{(t)}\) denotes the size of occupied cluster \(k\) at iteration \(t\). Thus, \(\nu_{\mathrm{w}}\) is the cluster-size-weighted average of the degrees of freedom parameters.

Figure~\ref{fig:application-k-trace} displays the full trace of \(K_{+,n}\), with the end of the burn-in period indicated by a vertical dashed line. Figure~\ref{fig:application-nu-traces} presents the post-burn-in traces of \(\nu_{\min}\), \(\nu_{\max}\), and \(\nu_{\mathrm{w}}\). These summaries are computed over the occupied clusters at each iteration and therefore do not track individual clusters with fixed labels.

\begin{figure}[!htbp]
    \centering
    \includegraphics[width=0.95\linewidth]{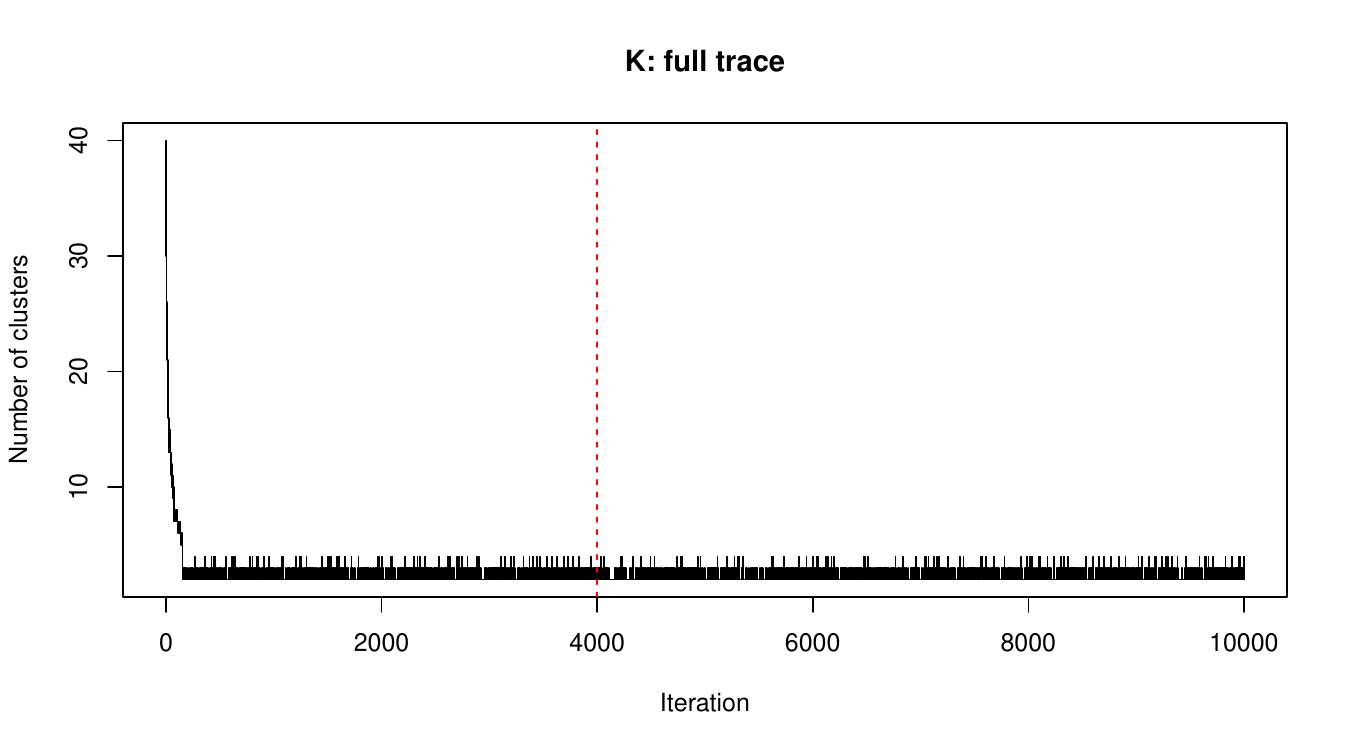}
    \caption{Trace of the number of occupied clusters \(K_{+,n}\) across all \(10{,}000\) MCMC iterations in the application analysis. The red dashed vertical line marks the end of the \(4{,}000\)-iteration burn-in period.}
    \label{fig:application-k-trace}
\end{figure}

\begin{figure}[!htbp]
    \centering

    \begin{subfigure}{0.95\textwidth}
        \centering
        \includegraphics[width=\linewidth]{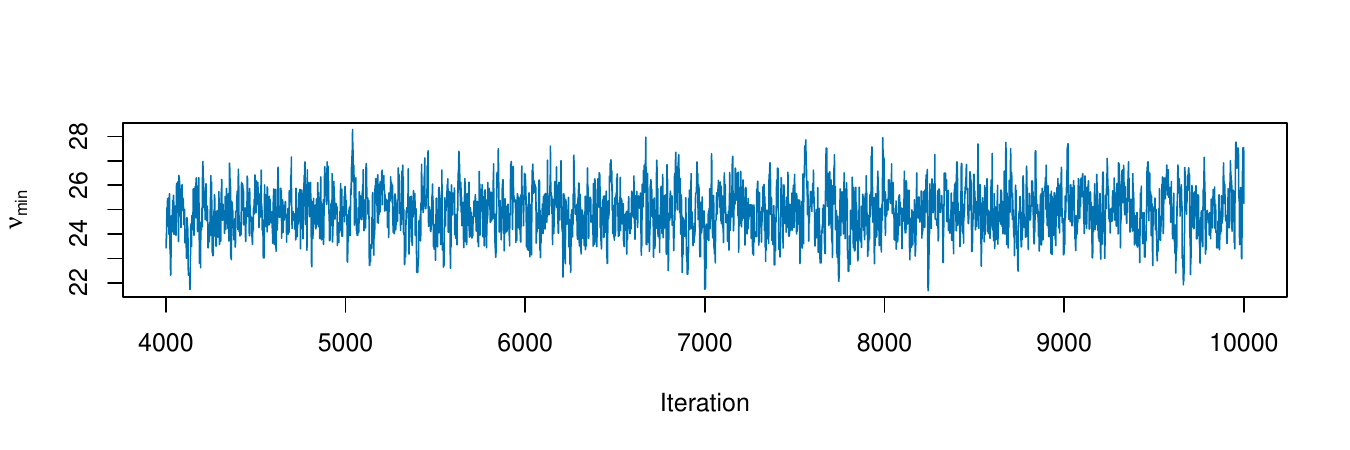}
        \caption{Minimum degrees of freedom, \(\nu_{\min}\).}
        \label{fig:application-nu-min-trace}
    \end{subfigure}

    \par\vspace{0.5em}

    \begin{subfigure}{0.95\textwidth}
        \centering
        \includegraphics[width=\linewidth]{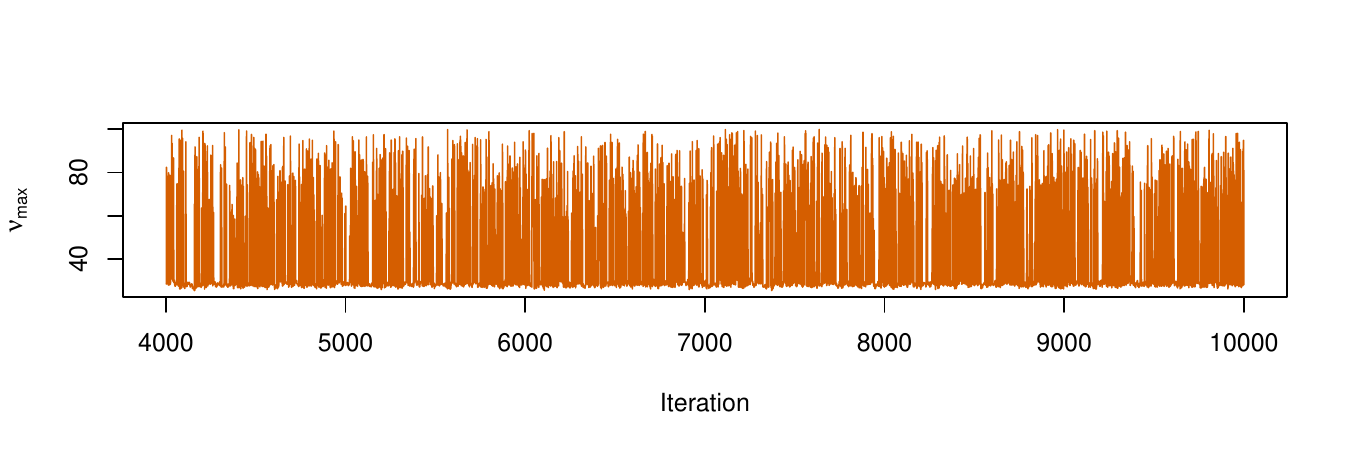}
        \caption{Maximum degrees of freedom, \(\nu_{\max}\).}
        \label{fig:application-nu-max-trace}
    \end{subfigure}

    \par\vspace{0.5em}

    \begin{subfigure}{0.95\textwidth}
        \centering
        \includegraphics[width=\linewidth]{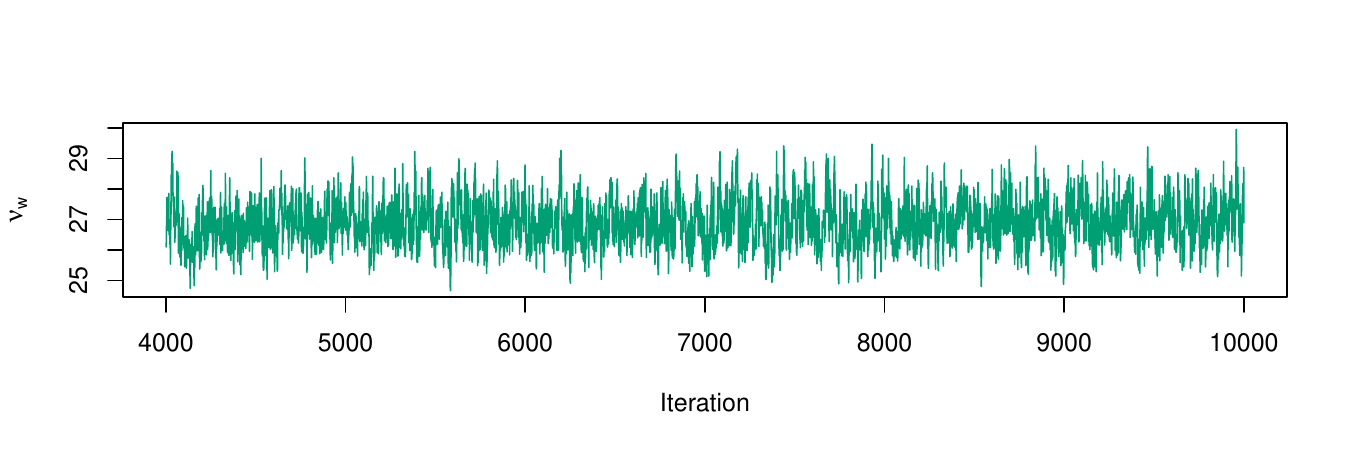}
        \caption{
        Cluster-size-weighted average degrees of freedom,
        \(\nu_{\mathrm{w}}\).
        }
        \label{fig:application-nu-weighted-trace}
    \end{subfigure}

    \caption{
    Post-burn-in traces of three label-invariant summaries of the cluster-specific degrees-of-freedom parameters in the application analysis. Panels (a)--(c) show \(\nu_{\min}\), \(\nu_{\max}\), and \(\nu_{\mathrm{w}}\), respectively, over iterations \(4{,}001\)--\(10{,}000\).}
    \label{fig:application-nu-traces}
\end{figure}

\begin{table}[h]
\centering
\caption{Effective sample sizes for the application analysis, computed from \(6{,}000\) post-burn-in draws. The monitored quantities are the number of occupied clusters and three label-invariant summaries of the cluster-specific degrees-of-freedom parameters. ESS estimates are rounded to the nearest integer.}
\label{tab:application-ess}
\begin{tabular}{lrr}
\toprule
Quantity & Bulk ESS & Tail ESS \\
\midrule
\(K_{+,n}\)         & 2583 & 2633 \\
\(\nu_{\min}\)      & 820  & 1183 \\
\(\nu_{\max}\)      & 1703 & 1332 \\
\(\nu_{\mathrm{w}}\) & 859  & 1777 \\
\bottomrule
\end{tabular}
\end{table}

As shown in Table~\ref{tab:application-ess}, the bulk ESS ranges from approximately \(820\) to \(2{,}583\), and the tail ESS ranges from approximately \(1{,}183\) to \(2{,}633\). These values suggest reasonable sampling efficiency for the monitored summaries. However, these single-chain ESS estimates do not by themselves establish convergence or adequate exploration of the full posterior distribution.

\end{appendices}
\end{document}